\documentclass[11pt,a4paper]{article}
\usepackage{jheppub}
\usepackage[utf8]{inputenc}
\usepackage[normalem]{ulem}
\makeatletter
\long\def\emph#1{\ifmmode\nfss@text{\em #1}\else\hmode@bgroup\text@command{#1}\em\check@icl #1\check@icr\expandafter\egroup\fi}
\makeatother
\usepackage{epsfig}
\usepackage{graphicx}
\usepackage{amsfonts,amsmath,amsthm,amssymb}
\usepackage{epiolmec}
\usepackage{phaistos}
\usepackage{adjustbox}
\newtheorem{theorem}{Theorem}
\newtheorem{definition}{Definition}
\newtheorem{notation}{Notation}
\newtheorem{remark}{Remark}
\newtheorem{proposition}[]{Proposition}

\newcommand{\tr}{\operatorname{tr}}
\usepackage{ulem}
\usepackage{multicol}
\usepackage{latexsym}
\usepackage{mathrsfs} 
\newcommand\Vtextvisiblespace[1][.5em]{%
  \mbox{\kern.06em\vrule height.3ex}%
  \vbox{\hrule width#1}%
  \hbox{\vrule height.3ex}}
\usepackage{cancel}

\usepackage{xparse}
\usepackage{float}
\usepackage{todonotes}
\usepackage{etoolbox}

\usepackage{tikz}
\usepackage{upgreek}
\usetikzlibrary{arrows,positioning,cd,calc,math,decorations.pathreplacing,decorations.markings,decorations.pathmorphing,patterns
}
\tikzset{wave/.style={decorate, decoration=snake}}

\usepackage{xcolor}
\definecolor{MyBlue}{rgb}{0.25,0.5,0.75}
\colorlet{NextBlue}{MyBlue!20}
\colorlet{SecondBlue}{MyBlue!40}
\usepackage{subcaption}
\NewDocumentCommand{\tens}{t_}
 {%
  \IfBooleanTF{#1}
   {\tensop}
   {\otimes}%
 }
\NewDocumentCommand{\tensop}{m}
 {%
  \mathbin{\mathop{\otimes}\displaylimits_{#1}}%
 }

\newcommand{\eq}[1]{\begin{equation}\begin{gathered}#1\end{gathered}\end{equation}}
\newcommand{\eqs}[1]{\begin{equation}\begin{gathered}\begin{split}#1\end{split}\end{gathered}\end{equation}}

\DeclareMathOperator{\res}{Res}
\newcommand{\sfa}{{\sf a}}
\newcommand{\sfb}{{\sf b}}
\newcommand{\sfc}{{\sf c}}
\newcommand{\sfd}{{\sf d}}
\newcommand{\sfQ}{{\sf Q}}
\newcommand{\sfY}{{\sf Y}}
\newcommand{\mc}{\mathcal}
\newcommand{\bs}{\boldsymbol}

\newcommand{\cP}{\mathcal{P}}
\newcommand{\cin}{\mathcal{C}_{in}}
\newcommand{\cout}{\mathcal{C}_{out}}
\newcommand{\cC}{\mathcal{C}}
\newcommand{\cH}{\mathcal{H}}
\newcommand{\Yt}{\widetilde{Y}}
\newcommand{\CS}{\underline{\mathcal{C}}}
\newcommand{\CP}{\overline{\mathcal{C}}}
\newcommand{\T}{\mathcal{T}}

\newcommand{\diag}{\text{diag}}
\usepackage{tikz}
\usepackage{upgreek}
\usetikzlibrary{decorations.markings}
\definecolor{darkspringgreen}{rgb}{0.05, 0.5, 0.06}
\definecolor{MyBlue}{rgb}{0.25,0.25,0.75}
\definecolor{MyRed}{rgb}{0.75,0.25,0.25}
\colorlet{NextBlue}{MyBlue!20}
\colorlet{SecondBlue}{MyBlue!40}
 
\title{Isomonodromic tau functions on a torus as  Fredholm determinants, and charged partitions}
\author[a]{Fabrizio Del Monte,}
\author[b,c]{Harini Desiraju,}
\author[d,e]{Pavlo Gavrylenko}
\affiliation[a]{Centre de Recherches Math\'ematiques (CRM), Universit\'e de Montr\'eal C. P. 6128, succ. Centre Ville, Montr\'eal, Qu\'ebec, Canada H3C 3J7}
\affiliation[b]{Scuola Internazionale Superiore di Studi Avanzati (SISSA),\\
Via Bonomea 265, Trieste}
\affiliation[c]{Institute for Geometry and Physics (IGAP), \\
via Beirut 2/1, 34151 Trieste, Italy}
\affiliation[d]{Center for Advanced Studies, Skolkovo Institute of Science and Technology, Nobel Street 1, 121205 Moscow, Russia}
\affiliation[e]{Faculty of Mathematics, HSE University, Usacheva 6, 119048 Moscow, Russia}
\emailAdd{fabrizio.delmonte@sissa.it}
\emailAdd{harini.desiraju@sissa.it}
\emailAdd{pasha.145@gmail.com}
 \abstract{We prove that the isomonodromic tau function on a torus with Fuchsian singularities and generic monodromies in $GL(N,\mathbb{C})$ can be written in terms of a Fredholm determinant of Plemelj operators. We further show that the minor expansion of this Fredholm determinant is described by a series labeled by charged partitions. As an example, we show that in the case of $SL(2,\mathbb{C})$ this combinatorial expression takes the form of a dual Nekrasov-Okounkov partition function, or equivalently of a free fermion conformal block on the torus. Based on these results we also propose a definition of the tau function of the Riemann-Hilbert problem on a torus with generic jump on the A-cycle.
 }
\begin{document}
\maketitle

\section{Introduction}
A central object in the study of monodromy preserving equations
is the so-called \textit{tau function} \cite{Jimbo:1981zz},  which encodes important information about the system, e.g. by generating the Hamiltonians governing its dynamics. Tau functions, when expressed as Fredholm determinants, bring together concepts from mathematics and physics: notable examples are the relation between Painlev\'e III and the two-dimensional Ising model \cite{PhysRevB.13.316}, Painlev\'e V and quantum correlations of Bose gases \cite{1990IJMPB...4.1003I}, and between Painlev\'e III and self-avoiding polymers \cite{Zamolodchikov:1994uw}, among others. An important example is the formulation of the general tau function of Painlev\'e VI as a Fredholm determinant \cite{Gavrylenko:2016zlf,Gavrylenko:2016moe}, and consequentially taking the form of Nekrasov partition functions \cite{Nekrasov:2003af,Nekrasov:2003rj} or equivalently free fermion conformal blocks.

While this correspondence was originally derived by using methods from two-dimensional Conformal Field Theory (CFT) \cite{Gamayun:2012ma,Iorgov:2013uoa,Iorgov:2014vla,Bershtein:2014yia,Bershtein:2016uov,Gavrylenko:2016moe,Gavrylenko:2018ckn}, it was later directly proven by representing the tau functions as Fredholm determinants in \cite{Gavrylenko:2016zlf,Gavrylenko:2017lqz,Cafasso:2017xgn,Gavrylenko:2018fsm} for the cases of Painlev\'e III,V,VI, and for Fuchsian systems on the sphere. Not only does the Fredholm determinant representation of the tau function provide an explicit formulation of the general solution to the Painlev\'e equations (Painlev\'e transcendents), but also reveals the combinatorial structure in terms of charged partitions underlying the tau functions, that are organized as a convergent power series in the isomonodromic time: such a representation is especially remarkable given the transcendental nature of these solutions.

A fascinating property of Painlev\'e equations is that they can be expressed as time-dependent Hamiltonian systems with a Calogero-type potential: this is the so-called Painlev\'e-Calogero correspondence \cite{Levin2000,Takasaki:2000zd}. In particular the sixth Painlev\'e equation with the parameters $\alpha_{0},\alpha_{1},\alpha_{2},\alpha_{3}$, can be expressed as a Hamiltonian system with an elliptic potential \cite{manin1996sixth},
    \begin{gather}
    (2\pi i )^2\frac{d^2 Q(\tau)}{d\tau^2} = \sum_{n=0}^{3} \alpha_{n} \wp'(Q+\omega_{n}), \label{step3}
\end{gather}
where
\begin{align}
    \omega_0=0, && \omega_1=\frac{1}{2}, && \omega_2=\frac{\tau}{2}, && \omega_3=\frac{\tau+1}{2}.
\end{align}
For a specific set of parameters, $\alpha_i=\frac{m^2}{8}$, $i=0,1,2,3$, 
this is the equation of the 2-particle nonautonomous Calogero-Moser system whose particles are positioned at $\pm Q$ in their center-of-mass frame, that we write below in \eqref{eq:EllPainleve}, and the associated Lax pair lives on a torus with one puncture at $z=0$, making it the simplest example to study isomonodromic deformations on a torus. 

In this paper, we first extend the determinant formalism of \cite{Gavrylenko:2016zlf} to construct the tau function of the 2-particle nonautonomous Calogero-Moser system as a Fredholm determinant that is explicitly determined by hypergeometric functions. We then generalise our construction and show that isomonodromic tau functions on a torus with an arbitrary number of Fuchsian singularities \cite{Korotkin:1995yi,levin1999hierarchies,takasaki1999elliptic,Korotkin:1999xx,Levin:2013kca} have a  Fredholm determinant representation, and its minor expansion can be written in terms of Nekrasov functions \cite{Nekrasov:2003af,Nekrasov:2003rj}. This extends and completes the analysis of \cite{Bonelli:2019boe,Bonelli:2019yjd}, where these cases were studied by using CFT methods. We conclude by outlining the general ideas behind the extension of the Widom constant's approach of \cite{Cafasso:2017xgn} to the present case.

\subsection*{Overview of the results}
Our starting point is the equation of motion for the 2-particle nonautonomous Calogero-Moser system \cite{takasaki1999elliptic}
\begin{align}\label{eq:EllPainleve}
    (2\pi i)^2\frac{d^2Q(\tau)}{d\tau^2}=m^2\wp'(2Q(\tau)|\tau),
\end{align}
where $m\in \mathbb{C}$ is an arbitrary complex parameter, and the Weierstrass $\wp$ function is defined in terms of the theta function $\theta_1$ by
\begin{align}
    \wp(z|\tau):=-\frac{\partial^2}{\partial z^2}\log\theta_1(z|\tau)-\frac{1}{6}\frac{\theta_1'''(0|\tau)}{\theta_1'(0|\tau)}\equiv -\frac{\partial^2}{\partial z^2}\log\theta_1(z|\tau)-\frac{1}{6}\frac{\theta_1'''}{\theta_1'},  
\end{align}
\begin{equation}
    \theta_1(z|\tau):=\sum_{n\in\mathbb{Z}}(-1)^{n-\frac{1}{2}}e^{i\pi\left(n+\frac{1}{2} \right)^2}e^{2\pi i\left(n+\frac{1}{2} \right)z}, \label{eq:Theta1def}
\end{equation}
with the theta function satisfying the following periodicity properties:
\begin{align}
    \theta_1(z+1 \vert \tau)=-\theta_1(z\vert \tau), && \theta_1(z+\tau \vert \tau)=-e^{-2\pi i(z+\frac{\tau}{2})}\theta_1(z \vert \tau). \label{eq:theta1transf}
\end{align}
 The modular parameter of the torus $\tau$  lies in the upper-half plane $\mathbb{H}$ and assumes the role of the isomonodromic time. The equation \eqref{eq:EllPainleve} arises as the compatibility condition of the following linear system on a torus with one puncture \cite{Korotkin:1995yi,levin1999hierarchies,takasaki1999elliptic},
\eqs{\label{eq:linear_systemCM}
\partial_z\mathcal{Y}_{CM}(z,\tau)=\mathcal{Y}_{CM}(z,\tau)L_{CM}(z,\tau), \\
2\pi i\partial_\tau \mathcal{Y}_{CM}(z,\tau)=\mathcal{Y}_{CM}(z,\tau) M_{CM}(z,\tau),
}
where $(L_{CM},M_{CM})$ is the Lax pair of 2-particle non-autonomous Calogero-Moser system
\begin{equation}
\begin{split}
    L_{CM}(z,\tau) &=\left(\begin{array}{cc}
        P(\tau) & mx(-2Q(\tau),z) \\
        mx(2Q(\tau),z) & -P(\tau)
    \end{array}\right),\\
    M_{CM}(z, \tau) &=m\left(\begin{array}{cc} 
        0 & y(-2Q(\tau),z) \\
        y(2Q(\tau),z) & 0
    \end{array}\right).
    \end{split}\label{linear_system}
\end{equation}
The functions $x(\xi,z)$, $y(\xi,z)$ and $P$  in \eqref{linear_system}  are respectively,
\begin{align}\label{eq:xyP}
    x(\xi,z)=\frac{\theta_1(z-\xi \vert \tau)\theta_1'(0\vert \tau)}{\theta_1(z\vert \tau)\theta_1(\xi \vert \tau)}, && y(\xi,z)=\partial_\xi x(\xi,z), && P=2\pi i\frac{dQ(\tau)}{d\tau}.
\end{align}
As opposed the behaviour of the Lax matrices on the sphere, the Lax matrix $L_{CM}$ in \eqref{linear_system} is not single-valued, and satisfies the relations
\begin{align}\label{L_1}
L_{CM}(z+1, \tau)=L_{CM}(z, \tau), &&  L_{CM}(z+\tau, \tau)=e^{-2\pi iQ(\tau)\sigma_3} L_{CM}(z, \tau) e^{2\pi iQ(\tau)\sigma_3}.
\end{align}
Subsequently, the solution of the linear system \eqref{linear_system} has the following monodromy properties around A,B cycles of the torus and around the puncture:
\begin{gather}
\begin{array}{c}
 \mathcal{Y}_{CM}(z+1, \tau)=M_A \mathcal{Y}_{CM}(z, \tau), \qquad \mathcal{Y}_{CM}(z+\tau, \tau)=M_B \mathcal{Y}_{CM}(z, \tau) e^{2\pi iQ(\tau)\sigma_3}, \\ \\
\mathcal{Y}_{CM}(e^{2\pi i}z, \tau)=M_0 \mathcal{Y}_{CM}(z, \tau), 
\end{array}
\end{gather}
under the constraint
\begin{equation}
    M_0=M_A^{-1} M_B^{-1} M_A M_B, \label{mon_cons}
\end{equation}
 and without loss of generality, it is always possible to set $M_A$ to be diagonal by conjugation. Introducing the monodromy exponent $a\notin\mathbb{Z} + \frac{1}{2}$ around the A-cycle, we have
\begin{align}\label{eq:1ptmonodromy}
    M_A=e^{2\pi i a\sigma_3}, && M_0\sim e^{2\pi im\sigma_3},
\end{align}
where $\sim$ means "in the same conjugacy class of", $\sigma_{3}$ is the Pauli sigma matrix, and $m$ is the free parameter of the equation \eqref{eq:EllPainleve}.
Furthermore, the Hamiltonian of the system \eqref{linear_system} is the A-cycle contour integral \cite{Korotkin:1995yi,Levin:2013kca}
\begin{equation} \label{eq:Ham_CM}
    H_{CM}(\tau)=\oint_A dz\frac{1}{2}\tr L_{CM}^2(z, \tau)=P(\tau)^2-m^2\wp(2Q(\tau)|\tau)+4\pi im^2\partial_\tau\log\eta(\tau),
\end{equation}
where $\eta(\tau)$ is Dedekind's eta function
\begin{equation}
    \eta(\tau):=\left(\frac{\theta_1'(0\vert \tau)}{2\pi} \right)^{1/3}.
\end{equation}
The generator of the Hamiltonian $H_{CM}$ is called the (isomonodromic) tau function $\T_{CM}$ of the 2-particle non-autonomous Calogero-Moser system, and is defined by
\begin{gather}\label{eq:IsomTau11}
    2\pi i \partial_{\tau} \log \T_{CM}(\tau) := H_{CM}(\tau).
\end{gather}
Our first result, proving a conjecture (equations 3.47, 4.10) of \cite{Bonelli:2019boe}, is that the tau function $\T_{CM}$ in \eqref{eq:IsomTau11} is proportional to the Fredholm determinant of an operator whose entries are determined solely by hypergeometric functions.
\begin{theorem}\label{prop:CM_Ham}
The isomonodromic tau function $\T_{CM}$ for the one-punctured torus is given by the following expression:
\begin{equation}\label{eq:prop1}
    \T_{CM}(\tau)=\det\left[\mathbb{I}- K_{1,1}\right]\,\left( \frac{e^{-2\pi i\rho}\eta(\tau)^2}{\theta_{1}(Q(\tau)-\rho)\theta_{1}(Q(\tau)+\rho)}\right)e^{2\pi i\tau\left(a^2+\frac{1}{6}\right)}\Upsilon_{1,1}(a,m),
\end{equation}
where $\rho$ is an arbitrary constant, $ Q(\tau)$ is the solution of the equation of motion for the 2-particle nonautonomous Calogero-Moser system \eqref{eq:EllPainleve}. The kernel $K_{1,1}(z,w; \tau)$ reads
\begin{gather}
  K_{1,1} (z,w; \tau)=
  \left(\begin{array}{cc}-e^{-2\pi i \rho}\frac{\Yt_{out}(z+\tau)\Yt_{in}(w)^{-1}}{e^{2\pi i(z-w+\tau)}-1}&
    \frac{\Yt_{out}(z+\tau)\Yt_{out}(w+\tau)^{-1}-\mathbb{I}}{1-e^{-2\pi i(z-w)}}\\
    \frac{\mathbb{I}-\Yt_{in}(z)\Yt_{in}(w)^{-1}}{1-e^{-2\pi i(z-w)}} &
    e^{2\pi i\rho}\frac{\Yt_{in}(z)\Yt_{out}(w+\tau)^{-1}}{e^{2\pi i(z-w-\tau)}-1}\end{array}\right).
\end{gather}
and the corresponding operator acts on $L^{2}(S^{1})\otimes \left(\mathbb{C}^2 \oplus \mathbb{C}^{2}\right)$. The function
\eq{\label{eq:in_hyp}
\Yt_{in}(z)=(1-e^{-2\pi iz})^{m}\times\operatorname{diag}(e^{2\pi i az},e^{-2\pi i az})\times\\\times
\begin{pmatrix}
{}_2F_1(m,m-2a,-2a,e^{-2\pi iz})&
-\frac{m}{2a}{}_2F_1(1+m,m-2a,1-2a,e^{-2\pi iz})\\\frac{m e^{-2\pi iz}}{2a+1}{}_2F_1(1+m,1+m+2a,2+2a,e^{-2\pi iz})&{}_2F_1(m,1+m+2a,1+2a,e^{-2\pi iz})
\end{pmatrix},
}
is the local behavior of the solution to the associated three-point spherical problem for $z\rightarrow-i\infty$, normalized in such a way that the monodromy around $-i\infty$ is diagonal and equal to $e^{2\pi ia\sigma_3}$, well-defined as a series in $e^{-2\pi iz}$, convergent for $|e^{-2\pi iz}|<1$, ${_2F_1}$ are hypergeometric functions, and the function $\Yt_{out}$ is defined by 
\begin{align}\label{eq:out_hyp}
\Yt_{out}(z):=e^{2\pi i(\nu+\delta\nu(a,m))\sigma_3}\sigma_1 \Yt_{in}(-z) \sigma_1, &&
e^{2\pi i\delta\nu(a,m)}=\frac{\Gamma(-2a)\Gamma(1+2a-m)}{\Gamma(1+2a)\Gamma(-2a-m)}.
\end{align}
This expression for $\Yt_{out}$, which is well-defined as a series in $e^{2\pi iz}$, was obtained in \cite{Bonelli:2016idi}, where $\nu$ parametrizes the B-cycle monodromy $M_B$ and $\delta\nu(a,m)$ is a shift depending on $a,m$. $\sigma_1$
is a Pauli sigma matrix, $m$ is the monodromy exponent around the puncture, $a$ is the monodromy exponent around the A-cycle of the torus, and $\Upsilon_{1,1}$ is an arbitrary function of the monodromy data.
\end{theorem}
An important consequence of the Fredholm determinant representation of the tau function in theorem \ref{prop:CM_Ham} is the combinatorial expansion in terms of Nekrasov partition functions, or equivalently free fermion conformal blocks, that we show in theorem \ref{prop:Minor_CM}.
The results for the 2-particle nonautonomous Calogero-Moser system are further generalized to the isomonodromic problem on an $n$-punctured torus $C_{1,n}$ which is characterised by the following $N\times N$ system of linear differential equations \cite{Takasaki:2001fr,Levin:2013kca}
\begin{align}
  \begin{array}{c}
     \frac{\partial}{\partial z}\mathcal{Y}\left(z;\tau,\lbrace z_i \rbrace_{_1}^{^n}\right) =\mathcal{Y}\left(z;\tau,\lbrace z_i \rbrace_{_1}^{^n}\right) L\left(z;\tau,\lbrace z_i \rbrace_{_1}^{^n}\right),\\ \\
     \quad (2\pi i) \frac{\partial}{\partial\tau} \mathcal{Y}\left(z;\tau,\lbrace z_i \rbrace_{_1}^{^n}\right) =  \mathcal{Y}\left(z;\tau,\lbrace z_i \rbrace_{_1}^{^n}\right) M_\tau\left(z;\tau,\lbrace z_i \rbrace_{_1}^{^n}\right), \\ \\
     \frac{\partial}{\partial z_k} \mathcal{Y}\left(z;\tau,\lbrace z_i \rbrace_{_1}^{^n}\right)=\mathcal{Y}\left(z;\tau,\lbrace z_i \rbrace_{_1}^{^n}\right)M_{z_k}\left(z;\tau,\lbrace z_i \rbrace_{_1}^{^n}\right),
\end{array}
&& 
\begin{array}{c}
z,z_1,\dots,z_n\in C_{1,n} ; \tau\in \mathbb{H}, \\
 k=1,\dots,n
\end{array}
\label{lax_general}
\end{align}
where\footnote{The dependence on the variables $\tau$, $z_{1},..., z_{n}$ of the functions $\mathcal{Y}(z)$, $L(z)$, $M(z)$, $H_{k}$, $H_{\tau}$, $\T_{H}$ is dropped henceforth for brevity.} $\mathcal{Y}(z) \in GL(N)$, and $L,M_{\tau}, M_{z_{k}}\in\mathfrak{gl}_N$ are the Lax matrices. The isomonodromic time evolution in this case is generated by $n+1$ Poisson commuting Hamiltonians, that can be obtained as before from contour integrals of $\frac{1}{2}\tr L^2$, and are generated by the isomonodromic tau function $\T_H$:
\begin{align}
    2\pi i\partial_\tau \log \T_{H}:=H_\tau =\frac{1}{2}\oint_A\tr L^2(z)dz, &&   \partial_{z_k}\log\T_{H}:=H_{k} =\res_{z=z_k}\frac{1}{2}\tr L^2(z).\label{eq:IsomHam}
\end{align}
In theorem \ref{thm:GL_Ham} we show that the isomonodromic tau function for the linear system \eqref{lax_general} is also described by a Fredholm determinant \eqref{eq:Thm2}. Furthermore, theorem \ref{thm:GL_minor} generalizes theorem \ref{prop:Minor_CM}, describing the tau function of the elliptic Garnier system in terms of Nekrasov partition functions.

\subsection*{Outline of the paper}
This paper is organised as follows. We introduce our main motivating example, the 2-particle nonautonomous Calogero-Moser system, in Section \ref{sec:CM}. We then introduce the pants decomposition for the one-punctured torus, and construct Plemelj operators acting on functions holomorphic on the annuli of the pants decomposition, in Section \ref{subsec:Plemelj_CM}. In Section \ref{subsec:Fred_CM}, we show that the Fredholm determinant of the Plemelj operators constructed in Section \ref{subsec:Plemelj_CM} is described by hypergeometric functions, and show its relation to the isomonodromic tau function $\T_{CM}$ proving Theorem \ref{prop:CM_Ham}, in Section \ref{subsec:Ham_CM}.

We extend the construction in Section \ref{sec:CM} to the case of a $GL(N)$ linear problem over a torus with $n$ Fuchsian singularities, in Section \ref{sec:GL}. In proving theorem \ref{thm:GL_Ham}, we show that the isomonodromic tau function $\T_H$ in \eqref{eq:IsomHam} can be written in terms of the Fredholm determinant of 3-point Plemelj operators constructed on boundary spaces of the pants decomposition of the $n$-punctured torus.

In Section \ref{subsec:Minor_Pants}, we perform the explicit minor expansion of the Fredholm determinant of the $n$-point in Theorem \ref{thm:GL_Ham}. Using previously obtained results for the tau functions associated to rank-1 linear systems, we write the explicit Nekrasov sum  representation for the tau functions of the 2-particle nonautonomous Calogero-Moser system in Theorem \ref{prop:Minor_CM}, and the elliptic Garnier system in Theorem \ref{thm:GL_minor}. Finally, in Section \ref{sec:RHP}, we outline a possible representation of the tau function on the one point torus in terms of determinant of a particular combination of Toeplitz operators and \(e^{2\pi i \tau\partial_z}\), which we propose to be the generalization of the Widom constant.

\section{The 2-particle nonautonomous Calogero-Moser system: a toy model}\label{sec:CM}
We defined the isomonodromic tau function $\T_{CM}$ for the equation \eqref{eq:EllPainleve} in equation \eqref{eq:Ham_CM} as the generator of the corresponding Hamiltonian. Another notion of a tau function describes it as a Fredholm determinant (if it exists) of an operator whose vanishing locus, called the {\it Malgrange divisor} \cite{Malgrange1982}, defines the non-solvability of some linear problem \cite{2010CMaPh.294..539B,2016arXiv160104790B}. In this spirit, following the construction in \cite{Gavrylenko:2016zlf}, we define a tau function as the Fredholm determinant of certain Plemelj operators. The overview of the construction for the one-punctured torus is as follows:
\begin{itemize}
    \item The pants decomposition \cite{hatcher1999pants} of the one-punctured torus consists of a trinion with two legs identified \cite{Goldman2009arXiv0901.1404G}, whose boundaries become the A-cycle of the torus;
    \item A linear system with 3 Fuchsian singularities, whose solution is explicitly described by hypergeometric functions, is associated to the trinion;
    \item Boundary (Hilbert) spaces are defined on the two legs of the trinion; 
    \item Two Plemelj operators, $P_{\Sigma}$ and $P_{\oplus}$, are defined in terms of the solutions to the linear systems on the torus and on the trinion respectively. The Plemelj operators project one boundary space on to the other, effectively 'gluing' the cut along the A-cycle and giving us the one-punctured torus.
    \item A tau function is then defined in \eqref{eq:TAU11} as a determinant of some combination of (restrictions of) the operators $\cP_{\Sigma}$ and $\cP_{\oplus}$.
\end{itemize}

\subsection{Pants decomposition and Plemelj operators}\label{subsec:Plemelj_CM}
Let us introduce the $2\times2$ matrix-valued function $\Yt(z)$ that solves the following auxiliary linear system on a cylinder with 3 punctures at $-i\infty,0,+i\infty$:
\begin{align}\label{eq:3ptcyl}
    \partial_z\widetilde{\mathcal{Y}}(z)=\widetilde{\mathcal{Y}}(z)L_{3pt}(z), && L_{3pt}(z)=-2\pi i A_{0}-2\pi i\frac{A_1}{1-e^{2\pi i z}},
\end{align}
whose fundamental solution $\widetilde{\mathcal{Y}}(z)$ is described by hypergeometric functions, see \cite{Gavrylenko:2016zlf,Bonelli:2019boe}. The local monodromy exponents of the Lax matrix in \eqref{eq:3ptcyl} are chosen so that they coincide with those  on the torus \eqref{eq:1ptmonodromy}:
\begin{align}\label{eq:3ptexpts}
    A_{0}\sim a\sigma_3, && A_{1}\sim m\sigma_3,
\end{align} 
and \(\widetilde{\mathcal{Y}}(z)\) itself is chosen in such a way that
\begin{equation*}
\widetilde{\mathcal{Y}}(z)^{-1}\mathcal{Y}_{CM}(z)
\end{equation*}
is regular and single-valued around \(z=0\) and has no monodromy around the closest A-cycles.
In other word, \(\widetilde{\mathcal{Y}}(z)\) ``approximates'' analytic behavior of \(\mathcal{Y}(z)\) in the fundamental domain having the same monodromies around puncture and around two closest A-cycles.

The trinion $\mathscr{T}$ can then be viewed as being obtained by cutting the torus along its A-cycle, see Figure \ref{fig:TorusTrinion}, inducing a homomorphism of monodromy groups $\pi_1(C_{3,0})\rightarrow\pi_1(C_{1,1})$
\begin{gather}
   M_{A}M_0M_{B}^{-1}M_{A}^{-1}M_{B}=1=(M_{A})M_0(M_{B}^{-1}M_{A}M_{B})^{-1}  := M_{-i\infty}^{3pt}M_0^{3pt}M_{i\infty}^{3pt},
\end{gather}
that defines the monodromies of the three-punctured cylinder around $-i\infty,0,+i\infty$ in terms of the monodromy representation of the torus as in Figure \ref{fig:Trinion}.

\begin{figure}[H]
\begin{center}
\begin{subfigure}{.45\textwidth}
\centering
\begin{tikzpicture}[scale=1.5]
\draw[thick,decoration={markings, mark=at position 0.25 with {\arrow{>}}}, postaction={decorate}] (0,0) circle [x radius=0.5, y radius =0.2];
\node at (0,0) {$M_0$};
\draw(-0.5,0)  to[out=270,in=90] (-1.5,-1.5) to[out=-90,in=-90] (1.5,-1.5);
\draw(0.5,0) to[out=270,in=90] (1.5,-1.5);

\draw(-0.7,-1.4) to[out= -30,in=210] (0.7,-1.4);
\draw(-0.55,-1.469) to[out= 30,in=-210] (0.55,-1.469);

\fill[red!30!white] (-0.2,-2.375) to[out=135,in=225] (-0.2,-1.605)
to (0.2,-1.6) to[out=190,in=170] (0.2,-2.373) --cycle;

\draw[dashed,color=black!60!white](-0.2,-2.375) to[out=135,in=225] (-0.2,-1.605)
to (0.2,-1.6) to[out=190,in=170] (0.2,-2.373) --cycle;

\fill[red](-0.2,-2.375) to[out=0,in=0] (-0.2,-1.605)
to (0.2,-1.6) to[out=10,in=-10] (0.2,-2.373) --cycle;

\draw(-0.2,-2.375) to[out=0,in=0] (-0.2,-1.605)
to (0.2,-1.6) to[out=10,in=-10] (0.2,-2.373) --cycle;
\node at ($(0,-0.7)$) {\Large $\mathscr{T}$};
\node at ($(0,-2.6)$) {\color{red}\Large $\mathscr{A}$};
\node at ($(-0.5,-1.8)$) {$\cC_{in}$};
\node at ($(0.7,-1.8)$) {$\cC_{out}$};
\end{tikzpicture}
\caption{Pants decomposition of $C_{1,1}$}
\label{fig:Torus}
\end{subfigure}
\hfill
\begin{subfigure}{.45\textwidth}
\centering
\begin{tikzpicture}[scale=1.5]
\draw[thick,decoration={markings, mark=at position 0.25 with {\arrow{>}}}, postaction={decorate}](0,0) circle[x radius=0.5, y radius =0.2];
\draw[red,thick,decoration={markings, mark=at position 0.25 with {\arrow{>}}}, postaction={decorate}] (-1,-2) circle[x radius=0.5, y radius =0.2];
\draw[red,thick,decoration={markings, mark=at position 0.25 with {\arrow{<}}}, postaction={decorate}] (1,-2) circle[x radius=0.5, y radius =0.2];

\draw(-0.5,0)  to[out=270,in=90] (-1.5,-2);
\draw(0.5,0) to[out=270,in=90] (1.5,-2);
\draw(-0.5,-2) to[out=90,in=180] (0,-1.5) to[out=0, in=90] (0.5,-2);
\node at ($(0,-0.7)$) {\Large $\mathscr{T}$};

\node at ($(0.6,0.3)$) {$M_{0}$};
\node at ($(-1,-2.5)$) {$M_A$};
\node at ($(-1.8,-2.0)$) {$\cC_{in}$};
\node at ($(1,-2.5)$) {$M_{B}^{-1} M_{A}^{-1} M_{B}$};
\node at ($(1.8,-2.0)$) {$\cC_{out}$};
\end{tikzpicture}
\caption{Trinion}
\label{fig:Trinion}
\end{subfigure}
\end{center}
\vspace{-0.5cm}
\caption{}
\label{fig:TorusTrinion}
\end{figure}

\begin{remark}
The linear system \eqref{eq:3ptcyl} is simply the usual three-point Fuchsian problem on the sphere, having mapped the sphere to a cylinder by $z\rightarrow e^{-2\pi iz}$. The punctures at $0,1,\infty$ become punctures at $-i\infty,0,i\infty$ respectively.
\end{remark}
\begin{definition}\label{def:YcalY}
Out of the solutions $\mathcal{Y}_{CM}(z)$, $\widetilde{\mathcal{Y}}(z)$ of the linear problems \eqref{eq:linear_systemCM}, \eqref{eq:3ptcyl} respectively, we define two matrix-valued functions $Y_{CM}(z)$, $\Yt(z)$ with diagonal monodromies around the boundary circles $\cC_{in}$ and $\cC_{out}$ in Figure \ref{fig:TorusTrinion}, by the following equations:
 \begin{align}
    Y_{CM}(z)\vert_{\cC_{in}}:=\mathcal{Y}_{CM}(z)\vert_{\cC_{in}}\in\cH_{in}, && Y_{CM}(z)\vert_{\cC_{out}}:=M_B^{-1}\mathcal{Y}_{CM}(z)\vert_{\cC_{out}}\in\cH_{out},
\end{align}
 \begin{align}
    \Yt(z)\vert_{\cC_{in}}\equiv \Yt_{in}(z):=\widetilde{\mathcal{Y}}(z)\vert_{\cC_{in}}\in\cH_{in}, && \Yt(z)\vert_{\cC_{out}}\equiv \Yt_{out}(z):=M_B^{-1}\widetilde{\mathcal{Y}}(z)\vert_{\cC_{out}}\in\cH_{out}.
\end{align}
\end{definition}
Notice that $Y_{CM}(z)$ and $\widetilde{Y}(z)$ also solve \eqref{eq:linear_systemCM}, \eqref{eq:3ptcyl} respectively.
Moreover,
\begin{equation}
    \widetilde{Y}(z)^{-1}Y_{CM}(z)=\widetilde{\mathcal{Y}}(z)^{-1}\mathcal{Y}_{CM}(z),
\end{equation}
so effectively they can be exchanged in the formulas where they appear in the form of such ratios.
Notice also that under such definition
\begin{equation}
\label{eq:YCMtwist}
Y_{CM}(z+\tau)=Y_{CM}(z)e^{2\pi i \boldsymbol{Q}},\quad z\in \mathcal{C}_{in}.
\end{equation}

The Hilbert spaces $\cH_{in}$, $\cH_{out}$ on the boundaries of the pants $\mathcal{C}_{in},  \cC_{out}$ respectively (see Figure \ref{fig:TorusTrinion}) have an orthogonal decomposition into spaces of positive and negative Fourier modes. A Hilbert space $\cH$ defined as the direct sum of $\cH_{in}$ and $\cH_{out}$ is then associated to the trinion $\mathscr{T}$:
\begin{equation}
    \cH := \cH_{in}\oplus\cH_{out}  = \left( \cH_{in, -} \oplus \cH_{out, +} \right) \oplus \left( \cH_{in, +} \oplus \cH_{out, -} \right):=\cH_{+} \oplus \cH_{-},
\end{equation}
where
\begin{equation}
\label{eq:2}
\mathcal{H}_+=\cH_{in, -} \oplus \cH_{out, +},\quad
\mathcal{H}_-=\cH_{in, +} \oplus \cH_{out, -}.
\end{equation}
The functions $f(z) \in \cH$ then have the decomposition
\begin{gather}
    f(z) = \left(\begin{array}{c}
         f_{in} \\
         f_{out}
    \end{array}  \right) = \left(\begin{array}{c}
         f_{in, -} \\
         f_{out, +}
    \end{array}  \right) \oplus \left(\begin{array}{c}
         f_{in, +} \\
         f_{out, -}
    \end{array}  \right) \equiv f_{+} \oplus f_{-}, \label{eq:f_dec}
\end{gather}
where
\begin{equation}
f_+=\left(\begin{array}{c}
         f_{in, -} \\
         f_{out, +}
    \end{array}  \right)\in\mathcal{H}_+,\quad
f_-=\left(\begin{array}{c}
         f_{in, +} \\
         f_{out, -}
    \end{array}  \right)\in\mathcal{H}_-,
\end{equation}
and the $\pm$ parts of the function are defined by their Fourier expansions:
\begin{equation}
\begin{split}\label{eq:functionComponents}
    f_{in,+}=e^{2\pi i az \sigma_3}\sum_{n=0}^\infty  f_{in,n}e^{-2\pi inz}, \qquad\qquad f_{in,-}=e^{2\pi i az\sigma_3}\sum_{n=1}^\infty f_{in,-n}e^{2\pi inz}, \\
    f_{out,+}=e^{2\pi i az\sigma_3}\sum_{n=0}^\infty f_{out,n}e^{-2\pi inz}, \qquad\qquad f_{out,-}=e^{2\pi i az\sigma_3}\sum_{n=1}^\infty f_{out,n}e^{2\pi inz},
\end{split}
\end{equation}
where the coefficients $f_{in,\pm n}, f_{out, \pm n}$ are column vectors.
On the space $\cH$ we introduce two Plemelj projectors in terms of the solutions to the linear systems \eqref{linear_system}, \eqref{eq:3ptcyl} respectively.

\begin{definition}\label{def:cP_Sigma} 
The Plemelj operator $\cP_{\Sigma_{1,1}}:\cH \rightarrow \cH$ is defined in terms of the solution to the linear system on the torus \eqref{linear_system} as
\begin{gather} 
    \left(\cP_{\Sigma_{1,1}} f\right)(z) = \int_{\cin \cup \cout} \frac{dw}{2\pi i}\, Y_{CM}(z;\tau)\Xi_2(z,w; \tau) Y_{CM}(w;\tau)^{-1} f(w) \nonumber\\
    \equiv\int_{\mathcal{C}} \frac{dw}{2\pi i} \, Y_{CM}(z;\tau)\Xi_2(z,w; \tau) Y_{CM}(w;\tau)^{-1} f(w),\label{ps}
\end{gather}
where
\begin{align}\label{eq:twistedCauchy} 
    \Xi_2(z,w; \tau)= \left( \begin{array}{cc}
        \frac{\theta_1(z-w+Q- \rho)\theta_1'(0)}{\theta_1(z-w)\theta_1(Q-\rho)} &  0 \\
        0 & -\frac{\theta_1(z-w-Q- \rho)\theta_1'(0)}{\theta_1(z-w)\theta_1(Q+\rho)}
    \end{array} \right).
\end{align}
\end{definition}
The function $\Xi_2(z,w; \tau)dw$ in \eqref{eq:twistedCauchy} is a twisted Cauchy kernel, with the properties
\begin{align}
    \Xi_2(z+\tau,w; \tau)=e^{-2\pi iQ\sigma_3+2\pi i\rho} \Xi_2(z,w; \tau), && \Xi_2(z,w+\tau; \tau)=\Xi_2(z,w; \tau) e^{2\pi i Q\sigma_3-2\pi i\rho},\label{eq:TauShiftXi2}
\end{align}
The variable\footnote{Here on, we drop the $\tau$ dependence of $Q$ for brevity.} $Q\equiv Q(\tau)$ is the solution of the non-autonomous Calogero-Moser system \eqref{eq:EllPainleve}, and $\rho$ is a parameter encoding a $U(1)$ B-cycle monodromy of the twisted Cauchy kernel as can be seen in \eqref{eq:TauShiftXi2}. It does not appear in the linear problem \eqref{eq:linear_systemCM}, but rather it is an arbitrary parameter whose role will become clear later (see remark \ref{rmk:zerostau11}).
The expansion of \(\Xi_2(z,w; \tau)\) for $z\sim w$ reads
\begin{equation}\label{Xi_expansion}
    \begin{split}\Xi_2(z,w; \tau) = \frac{\mathbb{I}}{z-w} + \diag\left[ \frac{\theta_{1}'(Q-\rho)}{\theta_{1}(Q-\rho)}, -\frac{\theta_{1}'(Q+\rho)}{\theta_{1}(Q+\rho)} \right] &+\frac{1}{2}(z-w)\,\diag\left[\frac{\theta_1''(Q-\rho)}{\theta_1(Q-\rho)}, \frac{\theta_1''(Q+\rho)}{\theta_1(Q+\rho)}\right]\\
    &  -\frac{\mathbb{I}}{6}(z-w)\frac{\theta_1'''}{\theta_1'} + \mc{O}\left( (z-w)^2 \right).  \end{split}
    \end{equation}
    \begin{definition}\label{notation:contour}
Since the integrand in \eqref{ps} has a singularity at \(w=z\), we define the following rule: each time \(w\) approaches \(z\), we go around the singularity in \emph{\it clockwise} direction. 
Sometimes it is also useful to use the notation $\mathcal{C}=\mathcal{C}_{in}\cup\mathcal{C}_{out}$, and $\CS$, $\CP$ for the shifted contours as in Figure \ref{fig:contour_PS}.
\end{definition}
One can verify that $
    \cP_{\Sigma_{1,1}}^2 = \cP_{\Sigma_{1,1}}$, and that the space of functions on the annulus $\mathscr{A}$, which is defined by the equation \eqref{eq:A_id} (see also Figure \ref{fig:Torus}), is
    \eq{
    \label{eq:ker11} \cH_{\mathscr{A}}\,\subseteq \textrm{ker}\cP_{\Sigma_{1,1}}.
    }
\begin{definition}\label{def:cP_oplus}
The  Plemelj operator $ \cP_{\oplus}: \cH \rightarrow \cH $ is defined in terms of the solution of the 3--point linear system \eqref{eq:3ptcyl} as
\begin{gather}
    \left(\cP_{\oplus} f\right)(z) = \int_{\cin \cup \cout} dw\, \frac{\Yt(z) \Yt (w){^{-1}}}{1-e^{-2\pi i(z-w)}} f(w) \nonumber \\
    = \int_{\mathcal{C}} dw\, \frac{\Yt(z) \Yt (w){^{-1}}}{1-e^{-2\pi i(z-w)}} f(w). \label{pp}
\end{gather}
\end{definition}
For $z\sim w$, 
 \eq{\frac{1}{1- e^{-2\pi i(z-w)}} = \frac{1}{2\pi i (z-w)} + \frac{1}{2} + \frac{2\pi i }{12}(z-w)+\mathcal{O}\left((z-w)^2\right). \label{eq:ExpExpansion}}
It can be verified that $ \cP_{\oplus}^2 = \cP_{\oplus}$, and
\eq{\label{eq:keroplus}\textrm{ker}~\cP_{\oplus} = \cH_{-}.}
Furthermore, one can prove that
    \eq{
        \cP_{\oplus}\cP_{\Sigma_{1,1}} = \cP_{\Sigma_{1,1}} ,\quad \cP_{\Sigma_{1,1}}\cP_{\oplus}= \cP_{\oplus},}
and therefore, the space of functions on the trinion $\mathscr{T}$ in Figure \ref{fig:Torus} is defined as
    \begin{gather}
   \cH_{\mathscr{T}}:= \textrm{im}~ \cP_{\oplus} = \textrm{im}~ \cP_{\Sigma_{1,1}} \label{eq:H_T}.
\end{gather}

 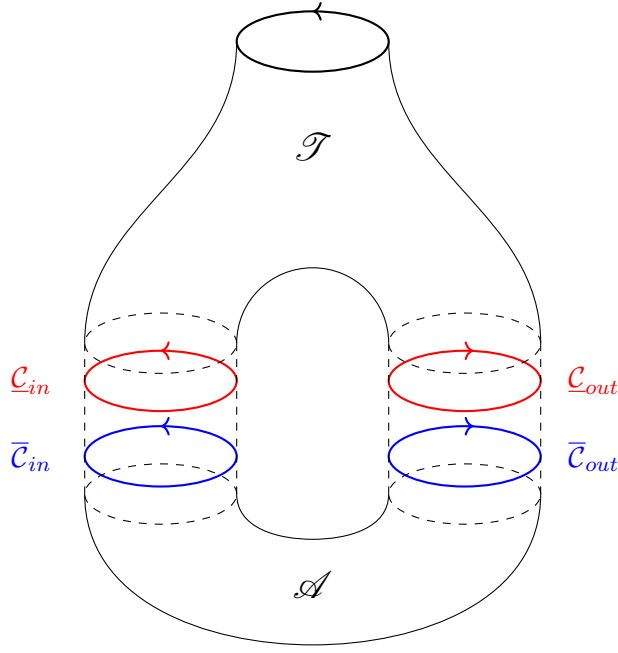
\begin{figure}[H]
\centering
\begin{tikzpicture}[scale=2]
\draw[thick,decoration={markings, mark=at position 0.25 with {\arrow{>}}}, postaction={decorate}](0,0) circle[x radius=0.5, y radius =0.2];
\draw[thick,red,decoration={markings, mark=at position 0.25 with {\arrow{>}}}, postaction={decorate}] (-1,-2.25) circle[x radius=0.5, y radius =0.2];
\draw[thick,red,decoration={markings, mark=at position 0.25 with {\arrow{<}}}, postaction={decorate}] (1,-2.25) circle[x radius=0.5, y radius =0.2];

\draw(-0.5,0)  to[out=270,in=90] (-1.5,-2);
\draw(0.5,0) to[out=270,in=90] (1.5,-2);
\draw(-0.5,-2) to[out=90,in=180] (0,-1.5) to[out=0, in=90] (0.5,-2);

\node at ($(0,-0.7)$) {\Large $\mathscr{T}$};\textbf{}

\node at ($(-1.85,-2.25)$) {{\color{red}$\underline{\mathcal{C}}_{in}$}};

\node at ($(1.85,-2.25)$) {{\color{red}$\underline{\mathcal{C}}_{out}$}};

\draw[thick,blue,decoration={markings, mark=at position 0.25 with {\arrow{>}}}, postaction={decorate}] (-1,-2.75) circle[x radius=0.5, y radius =0.2];
\draw[thick,blue,decoration={markings, mark=at position 0.25 with {\arrow{<}}}, postaction={decorate}] (1,-2.75) circle[x radius=0.5, y radius =0.2];

\node at ($(-1.85,-2.75)$) {{\color{blue}$\overline{\mathcal{C}}_{in}$}};
\node at ($(1.85,-2.75)$) {{\color{blue}$\overline{\mathcal{C}}_{out}$}};

\draw (-1.5,-3) to[out=270,in=180] (0,-4) to[out=0,in=270] (1.5,-3);
\draw (-0.5,-3) to[out=270,in=180] (0,-3.3) to[out=0,in=270] (0.5,-3);
\draw[dashed] (-1.5,-2) to (-1.5,-3);
\draw[dashed] (-0.5,-2) to (-0.5,-3);
\draw[dashed] (1.5,-2) to (1.5,-3);
\draw[dashed] (0.5,-2) to (0.5,-3);

\draw[dashed] (-1,-3) circle[x radius=0.5, y radius =0.2];
\draw[dashed] (1,-3) circle[x radius=0.5, y radius =0.2];
\draw[dashed] (-1,-2) circle[x radius=0.5, y radius =0.2];
\draw[dashed] (1,-2) circle[x radius=0.5, y radius =0.2];

\node at ($(0,-3.6)$) {\Large $\mathscr{A}$};

\end{tikzpicture}
\caption{Contours}
\label{fig:contour_PS}
\end{figure}
The components of $\mc P_\oplus$ under the orthogonal decomposition are obtained by computing its action on the function $f(z)\in\cH$:
\eq{\label{eq:Pofin}
 (\mc P_\oplus f)(z)_{in}=\oint_{\mc C_{in}}dw\frac{1}{1- e^{-2\pi i(z-w)}}f_{in}(w)\\
  +\oint_{\mc C_{in}}dw\frac{\Yt_{in}(z) \Yt_{in}(w)^{-1}-1}{1-e^{-2\pi i(z-w)}}f_{in}(w) \\
  +\oint_{\mc C_{out}}dw\frac{\Yt_{in}(z)\Yt_{out}(w)^{-1}}{1-e^{-2\pi i(z-w)}}f_{out}(w),
  }
\eq{ \label{eq:Pofout}
    (\mc P_\oplus f)(z)_{out}=\oint_{\mc C_{out}}dw\frac1{1- e^{-2\pi i(z-w)}}f_{out}(w)  \\
  +\oint_{\mc C_{out}} dw\frac{\Yt_{out}(z)\Yt_{out}(w)^{-1}-1}{1- e^{-2\pi i(z-w)}}f_{out}(w) \\
  +\oint_{\mc C_{in}}dw\frac{\Yt_{out}(z)\Yt_{in}(w)^{-1}}{1- e^{-2\pi i(z-w)}}f_{in}(w).
  }
To analyze the formulas above we notice that
\begin{equation} \label{eq:Notice1}
  \Yt_{out}(z)\in  \mathbb{C}[[e^{2\pi iz}]]\otimes e^{2\pi iaz\sigma_3} End(\mathbb C^2),\quad\Yt_{in}(z)\in \mathbb C[[e^{-2\pi iz}]]\otimes e^{2\pi iaz\sigma_3} End(\mathbb C^2),
\end{equation}
and
\eq{ \label{eq:Notice2}
  \int_{-\frac{1}{2}+ic}^{\frac{1}{2}+ic}dw \frac1{1-e^{-2\pi i (z-w)}}f(w)=  \left\{\begin{array}{l}
f_{+}(z),\quad \textrm{Im}~ z<c\\
         -f_{-}(z),\quad \textrm{Im}~ z>c.
        \end{array}
      \right.
}
Because of \eqref{eq:Notice1}, \eqref{eq:Notice2}, the action of $\cP_{\oplus}$ on $f(z)$ in \eqref{eq:Pofin}, \eqref{eq:Pofout} can be rewritten as
\eq{\label{eq:Pplusabcd}
  (\mc P_\oplus f)(z)=\begin{pmatrix}(\mc P_\oplus f)_{in,-}\\(\mc P_\oplus f)_{out,+}  \end{pmatrix}\oplus\begin{pmatrix}  (\mc P_\oplus f)_{in,+}\\(\mc P_\oplus f)_{out,-}
  \end{pmatrix}
  =\\=
  \begin{pmatrix}f_{in,-}\\f_{out,+}  \end{pmatrix}
  \oplus\begin{pmatrix}{\sf a}&{\sf b}\\{\sf c}&{\sf d}\end{pmatrix}\begin{pmatrix}f_{in,-}\\f_{out,+}\end{pmatrix},
}
where ${\sfa}, {\sfb}, {\sfc}, {\sfd} $ are the components of $\cP_\oplus$ with respect to the decomposition $\cH=\cH_{in} \oplus\cH_{out} $:
\eqs{
  ({\sfa} f)(z)=&\oint_{\mc C_{in}}dw\frac{\Yt_{in}(z)\Yt_{in}(w)^{-1}-\mathbb{I}}{1-e^{-2\pi i(z-w)}}f_{in}(w), && z\in\cin,\\
  ({\sfb} f)(z)=&\oint_{\mc C_{out}}dw\frac{\Yt_{in}(z)\Yt_{out}(w)^{-1}}{1-e^{-2\pi i(z-w)}}f_{out}(w), && z\in\cin,\\
  ({\sfc} f)(z)=&\oint_{\mc C_{in}}dw\frac{\Yt_{out}(z)\Yt_{in}(w)^{-1}}{1- e^{-2\pi i(z-w)}}f_{in}(w), && z\in\cout,\\
  ({\sfd} f)(z)=&\oint_{\mc C_{out}} dw \frac{\Yt_{out}(z)\Yt_{out}(w)^{-1}-\mathbb{I}}{1- e^{-2\pi i(z-w)}}f_{out}(w), && z\in\cout.
  \label{eq:kernels}
}
The functions $\Yt_{in},\Yt_{out}$ are the local solutions of the three-point problem \eqref{eq:3ptcyl} around $\mp i\infty$, defined in Definition \ref{def:YcalY}.
They are given by, respectively \eqref{eq:in_hyp},   
which is well-defined as a series in $e^{-2\pi iz}$, convergent for $|e^{-2\pi iz}|<1$, and \eqref{eq:out_hyp}, 
which is well-defined as a series in $e^{2\pi iz}$. 
\begin{definition}\label{def:tau_11} 
The tau function $\T^{(1,1)} $ is defined, in terms of the Plemelj operators $\cP_{\oplus},\cP_{\Sigma_{1,1}}$ in definitions \ref{def:cP_oplus} and \ref{def:cP_Sigma}, as:
\begin{equation}\label{eq:TAU11}
    \T^{(1,1)}(\tau):=\det_{\cH_+}\left[\cP_{\Sigma_{1,1},+}^{-1}\cP_{\oplus,+} \right],
\end{equation}
where
\begin{equation}
    \cP_{\cdot,+}:=\cP_\cdot\vert_{\cH_+}.
\end{equation}
\end{definition}
In general, it is useful to introduce the following notation:
\begin{notation}\label{not:tau_gn}
$\T^{(g,n)}$ denotes the determinant tau function on genus $g$ Riemann Surfaces with $n$ Fuchsian singularities.
\end{notation}
\subsection{Constructing the Fredholm determinant}\label{subsec:Fred_CM} 
As a stepping stone to theorem \ref{prop:CM_Ham}, that links the determinant tau function \eqref{eq:TAU11} to the isomonodromic tau function \eqref{eq:IsomTau11}, in the following proposition we show that the tau function $\T^{(1,1)}$ of Definition \ref{def:tau_11} depends solely on the operators $\sfa,\sfb,\sfc,\sfd$ defined by the three-point problem.
\begin{proposition}\label{prop:CM_tau}
The tau function $\T^{(1,1)}(\tau)$ is the Fredholm determinant of an operator acting on $L^2(S^{1})\otimes \left(\mathbb{C}^2 \oplus \mathbb{C}^{2}\right)$, 
explicitly determined by hypergeometric functions
\begin{equation} \label{eq:propdetT11}
\T^{(1,1)}(\tau) = \det\left[ \mathbb{I} -K_{1,1} \right],
\end{equation}
where
\begin{gather}
   K_{1,1}(z,w)=
  \left(\begin{array}{cc}-e^{-2\pi i \rho}\frac{\Yt_{out}(z+\tau)\Yt_{in}(w)^{-1}}{1-e^{-2\pi i(z-w+\tau)}}&
    \frac{\Yt_{out}(z+\tau)\Yt_{out}(w+\tau)^{-1}- \mathbb{I}}{1-e^{-2\pi i(z-w)}}\\
    \frac{\mathbb{I}-\Yt_{in}(z)\Yt_{in}(w)^{-1}}{1-e^{-2\pi i(z-w)}} &
    e^{2\pi i \rho}\frac{\Yt_{in}(z)\Yt_{out}(w+\tau)^{-1}}{1-e^{-2\pi i(z-w-\tau)}}
  \end{array} \right), 
\end{gather}
$\Yt_{in}$ and $\Yt_{out}$ are the solutions of the three-point problem on the cylinder \eqref{eq:3ptcyl}, given by \eqref{eq:in_hyp} and \eqref{eq:out_hyp} respectively, $\rho$ parametrizes the \(U(1)\) shift of the B-cycle monodromy of \(\mathcal{P}_{\Sigma}\), and $\tau$ is the modular parameter of the torus.
\end{proposition}
\begin{proof}
Starting from the definition \eqref{eq:TAU11} of $\T^{(1,1)}$, we compute the action of $\cP_{\Sigma_{1,1},+}^{-1}\cP_{\oplus,+}$ on a function $f\in\cH_+$:
\begin{align}
F:=\mc P_{\Sigma_{1,1},+}^{-1}\mc P_{\oplus,+}f 
 \quad \Rightarrow \quad \mc P_{\Sigma_{1,1}}F=\mc P_{\oplus}f \, , \qquad\quad F \in \cH_{+}.
\end{align}
Noting that for any projector $\mc P$ acting on a vector $x$, one has $x-\mc P x\in \textrm{ker } \mc P$, and that\footnote{When $\left(\mathbb{I}- K_{1,1}\right)$ is invertible, $\cH = \cH_{\mathscr{T}} \oplus \cH_{A}$, and therefore $\textrm{ker } \cP_{\Sigma_{1,1}}=\cH_\mathscr{A}$. } $ \textrm{ker } \cP_{\Sigma_{1,1}}=\cH_\mathscr{A}$:
\begin{align} \label{eq:F_11}
    F=(F-\mc P_{\Sigma_{1,1}} F) + \mc P_{\Sigma_{1,1}} F=A+\mc P_\oplus f, && A:=F-\mc P_{\Sigma_{1,1}} F\in \mc H_{\mathscr{A}}.
\end{align}
In components, $A$ reads
\eq{
  A=\begin{pmatrix}
    A_{in,-}(z)\\A_{out,+}(z)
\end{pmatrix}\oplus
\begin{pmatrix}
  A_{in,+}(z)\\A_{out,-}(z)
\end{pmatrix}.
}
The identification of $\cC_{in}$ with $\cC_{out}$, that produces the torus from the trinion as in Figure \ref{fig:TorusTrinion}, is implemented at the level of functional spaces by setting
\begin{align}\label{eq:A_id}
    A_{in, \pm} = \nabla^{-1} A_{out, \pm}, 
\end{align}
where $\nabla:\cH_{in}\rightarrow\cH_{out}$ is a translation operator acting on an arbitrary function $g(z)\in\cH_{in}$ as
\eq{\label{eq:def_nabla} \nabla g (z) =e^{2\pi i\rho}g (z-\tau).} 
The factor $e^{2\pi i\rho}$ takes into account the $U(1)$ B-cycle monodromy of the Cauchy kernel in \eqref{eq:TauShiftXi2}. Using the explicit form of $\cP_\oplus$ in \eqref{eq:Pplusabcd}, together with the fact that $F\in\cH_+$, equation \eqref{eq:F_11} reads:
\eq{\label{eq:comp_F11}
  \begin{pmatrix}F_{in,-}\\F_{out,+}  \end{pmatrix}
  \oplus\begin{pmatrix}0\\0\end{pmatrix}
  =\begin{pmatrix}
    A_{in,-}\\A_{out,+}
\end{pmatrix}\oplus
\begin{pmatrix}
  A_{in,+}\\A_{out,-}
\end{pmatrix}+
  \begin{pmatrix}f_{in,-}\\f_{out,+}  \end{pmatrix}
  \oplus\begin{pmatrix}{\sfa}&{\sfb}\\{\sfc}&{\sfd}\end{pmatrix}\begin{pmatrix}f_{in,-}\\f_{out,+}.\end{pmatrix}
}
The $\cH_-$ components of \eqref{eq:comp_F11} are solved by
\eq{\label{eq:11AHminus}
  A_{out,-}=-{\sfc}f_{in,-}-{\sfd}f_{out,+}= \nabla A_{in,-}\,, \\
  A_{in,+}=-{\sfa}f_{in,-}-{\sfb}f_{out,+} = \nabla^{-1} A_{out,+}\,,
}
and substituting \eqref{eq:11AHminus} into \eqref{eq:comp_F11} gives
\eq{\label{eq:cL11}
F = \begin{pmatrix}F_{in,-}\\F_{out,+}\end{pmatrix}=\begin{pmatrix}f_{in,-}\\f_{out,+}\end{pmatrix}-\begin{pmatrix}\nabla^{-1}{\sfc}&\nabla^{-1}{\sfd}\\\nabla {\sfa}&\nabla {\sfb}\end{pmatrix}\begin{pmatrix}f_{in,-}\\f_{out,+}\end{pmatrix}:= \left( \mathbb{I} -\widehat{K}_{1,1}\right) f.
}
We note that the kernel $\widehat{K}$ in \eqref{eq:cL11}, when expressed in spherical coordinates, becomes the one appearing in Section 4 of \cite{Bonelli:2019boe}.
It is however more natural to conjugate the kernel $\widehat{K}_{1,1}$ by the operator ${\rm diag}(1,\nabla^{-1})$:
\eq{\label{eq:nabla_conj11}
K_{1,1}:={\rm diag}(1,\nabla^{-1})\, \widehat{K}_{1,1}\, {\rm diag}(1,\nabla) = \begin{pmatrix}
   \nabla^{-1} {\sfc}&\nabla^{-1}{\sfd}\nabla\\{\sfa}&{\sfb}\nabla
  \end{pmatrix}
}
The advantage of such a conjugation is the following: recall that we identify $\cin$ and $\cout$ with two copies of the A-cycle obtained by cutting the B-cycle of the torus. They are given by the segments in figure \eqref{fig:TorusParalgram_11} with endpoints identified.
 \begin{figure}[H]
\centering
\begin{tikzpicture}[scale = 4]
\draw[ultra thick, blue,decoration={markings, mark=at position 0.5 with {\arrow{<}}}, postaction={decorate}] (0,0) to (1,0);
\draw[ultra thick, red,decoration={markings, mark=at position 0.5 with {\arrow{>}}}, postaction={decorate}] (0.72,0.5) to (1.72,0.5);
\draw[thick] (0,0) to (0.72,0.5);
\draw[thick] (1,0) to (1.72,0.5);

\node at (0.5,-0.1) {{\color{blue}$\cin$}};
\node at (1.22,0.6) {{\color{red}$\cout$}};

\node at (-0,-0.06) {$0$};
\node at (1,-0.06) {$1$};
\node at (0.72,0.56) {$\tau$};
\node at (1.72,0.56) {$1+\tau$};
\end{tikzpicture}
\caption{$\cin,\cout$ in coordinates on the torus}
\label{fig:TorusParalgram_11}
\end{figure}
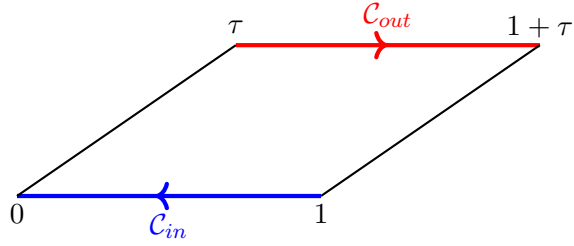
\noindent After the conjugation, $\widehat{K}_{1,1}$ is defined on a single circle, since all the functions on $\cout$ are translated by $\tau$, as is clear from the explicit expression
\eq{\label{eq:kerker_11}
   K_{1,1}(z,w)=
  \left(\begin{array}{cc}-e^{-2\pi i \rho}\frac{\Yt_{out}(z+\tau)\Yt_{in}(w)^{-1}}{1-e^{-2\pi i(z-w+\tau)}}&
    \frac{\Yt_{out}(z+\tau)\Yt_{out}(w+\tau)^{-1}-\mathbb{I}}{1-e^{-2\pi i(z-w)}}\\
    \frac{\mathbb{I} -\Yt_{in}(z)\Yt_{in}(w)^{-1}}{1-e^{-2\pi i(z-w)}} &
     e^{2\pi i \rho}\frac{\Yt_{in}(z)\Yt_{out}(w+\tau)^{-1}}{1-e^{-2\pi i(z-w-\tau)}}
  \end{array} \right).}
The tau function $\T^{(1,1)}$ in \eqref{eq:TAU11} is therefore
  \eqs{\T^{(1,1)}(\tau) = \det_{\cH_{+}}\left[\cP_{\Sigma_{1,1},+}^{-1} \cP_{\oplus,+}  \right] = \det[\mathbb{I}- K_{1,1}]. }
\end{proof}
Let us highlight the block determinant structure of the tau function 
\begin{gather} \label{eq:11blockdet}
    \T^{(1,1)}(\tau) = \det_{\cH_{+}} \left[\cP_{\Sigma_{1,1},+}^{-1} \cP_{\oplus,+}  \right] = \det \left[ \mathbb I - \begin{pmatrix}
   \nabla^{-1} {\sfc}&\nabla^{-1}{\sfd}\nabla\\{\sfa}&{\sfb}\nabla
  \end{pmatrix} \right],
\end{gather}
which will prove important in theorem \eqref{thm:GL_det}, that generalizes proposition \ref{prop:CM_tau} to the case of a genus 1 surface with $n$ punctures, with tau function $\T^{(1,n)}$. 

\subsection{Relation to the Hamiltonian: Proof of Theorem \ref{prop:CM_Ham}}\label{subsec:Ham_CM}
In this section we prove that the logarithmic derivative of the tau function \eqref{eq:TAU11}
differs from the Hamiltonian \eqref{eq:Ham_CM} by a factor that we compute. Let us recall the main statement of theorem \ref{prop:CM_Ham}:
\begin{equation}
  \T_{CM}(\tau)=\det\left[\mathbb{I}- K_{1,1}\right]\frac{e^{-2\pi i\rho}\eta(\tau)^2}{\theta_{1}(Q-\rho)\theta_{1}(Q+\rho)} e^{2\pi i\tau \left( a^2 +\frac{1}{6} \right)}\Upsilon_{1,1}(a,m),
\end{equation}
where $\Upsilon$ is an arbitrary function of the monodromy data of the system \eqref{eq:linear_systemCM}.
\begin{proof}
 Recall from \eqref{ps}, \eqref{pp} that 
 \begin{gather}
     \cP_{\Sigma_{1,1}}f(z) = \int_{\mathcal{C}} \frac{dw}{2\pi i}\, Y_{CM}(z,\tau) \Xi_2(z,w;\tau) Y_{CM}(w, \tau){^{-1}} f(w) \\
     \cP_{\oplus} f(z) = \int_{\mathcal{C}} dw\, \frac{\Yt(z) \Yt(w)^{-1}}{1-e^{-2\pi i(z-w)}} f(w),
 \end{gather}
 and since $\cP_\oplus$ does not depend on $\tau$, the logarithmic derivative of $\T^{(1,1)}$ in \eqref{eq:TAU11} is (see also pg. 20 in \cite{Gavrylenko:2016zlf})
\begin{gather} \label{eq:11Tr}
    \partial_{\tau} \log \T^{(1,1)}(\tau) = - \tr_{\cH} \cP_{\oplus} \partial_{\tau} \cP_{\Sigma_{1,1}}.
 \end{gather}
{The computation of the \(\tau\)-derivative of \(\mathcal{P}_{\Sigma}\) needs careful analysis.
In principle, the operator $\mathcal{P}_{\Sigma}$ acts on different spaces for different values of the  complex moduli: to define its derivative we need a local identification of these spaces (connection).
In the spherical case such an identification is absolutely natural, because we can keep the system of contours \(\mathcal{C}_{in,out}\) untouched while varying the complex moduli; which is no longer true in the torus case, since the position of $\cC_{out}$ depends on \(\tau\), see Figure~\ref{fig:TorusParalgram_11}.
In order to make the space \(\mathcal{H}_{out}\) \(\tau\)-independent we identify it with \(\mathcal{H}_{in}\) using the shift operator \(\nabla\) defined in \eqref{eq:def_nabla}, by setting \(\mathcal{H}_{out}=\nabla\mathcal{H}_{in}'\), where the space $\cH_{in}'$ is isomorphic to $\cH_{in}$.
This identification gives us a new operator $\cP_{\Sigma_{1,1}}'$ acting on ``time-independent'' spaces: \(\mathcal{P}_{\Sigma_{1,1}}': \mathcal{H}_{in}\oplus \mathcal{H}_{in}'\to \mathcal{H}_{in}\oplus \mathcal{H}_{in}'\).
\begin{equation}
\label{eq:Pp}
\mathcal{P}_{\Sigma_{1,1}}':=\diag(1,\nabla^{-1})\mathcal{P}_{\Sigma_{1,1}}\diag(1,\nabla).
\end{equation}
We identify \(\mathcal{H}_{in}'\) with the space of functions on \(\mathcal{C}_{in}'\), which is just another copy of \(\mathcal{C}_{in}\), introduced for convenience to describe the block structure of \(\mathcal{P}_{\Sigma}'\) by indicating the positions of the arguments of the kernel. Using these notations, the kernel of \(\mathcal{P}_{\Sigma}'\) is given by the following expressions:
\begin{equation}\begin{gathered}\begin{split}
\label{eq:PpKernel}
&\mathcal{P}_{\Sigma_{1,1}}'(w,z)=\mathcal{P}_{\Sigma_{1,1}}(w,z), \quad \text{for } w, z\in \mathcal{C}_{in},\\
&\mathcal{P}_{\Sigma_{1,1}}'(w,z)=e^{-2\pi i\rho}\mathcal{P}_{\Sigma_{1,1}}(w+\tau,z), \quad \text{for } w\in \mathcal{C}_{in}', z\in \mathcal{C}_{in},\\
&\mathcal{P}_{\Sigma_{1,1}}'(w,z)=e^{2\pi i\rho}\mathcal{P}_{\Sigma_{1,1}}(w,z+\tau), \quad \text{for } w\in \mathcal{C}_{in}, z\in \mathcal{C}_{in}',\\
&\mathcal{P}_{\Sigma_{1,1}}'(w,z)=\mathcal{P}_{\Sigma_{1,1}}(w+\tau,z+\tau), \quad \text{for } w, z\in \mathcal{C}_{in}'.
\end{split}\end{gathered}\end{equation}
Now we define the $\tau$-derivative of \(\mathcal{P}_{\Sigma_{1,1}}\) simply as
\begin{equation}
\label{eq:3}
\partial_{\tau}\mathcal{P}_{\Sigma_{1,1}}:=\diag(1,\nabla)\partial_{\tau}\mathcal{P}_{\Sigma_{1,1}}'\diag(1,\nabla^{-1}).
\end{equation}
Using \eqref{eq:PpKernel} we get the kernel of \(\partial_{\tau}\mathcal{P}_{\Sigma_{1,1}}\) explicitly:
\begin{equation}\begin{gathered}\begin{split}
\label{eq:11xizeta}
  \left(\partial_\tau\mc P_{\Sigma_{1,1}}\right)(w,z)&=\partial_\tau\mc P_{\Sigma_{1,1}}(w,z) \quad \textrm{for} \quad w, z \in \cC_{in}, \\
  \left(\partial_\tau\mc P_{\Sigma_{1,1}}\right)(w,z) &=(\partial_\tau+\partial_w)\mc P_{\Sigma_{1,1}}(w,z) \quad \textrm{for} \quad w \in \cC_{out}, \, z \in \cC_{in}, \\
  \left(\partial_\tau\mc P_{\Sigma_{1,1}}\right)(w,z) &=(\partial_\tau+\partial_{z})\mc P_{\Sigma_{1,1}}(w,z) \quad \textrm{for} \quad w \in \cC_{in}, \, z \in \cC_{out},\\
  \left(\partial_\tau\mc P_{\Sigma_{1,1}}\right)(w,z) &=(\partial_\tau+\partial_w+\partial_{z})\mc P_{\Sigma_{1,1}}(w,z) \quad w, z \in \cC_{out}.
\end{split}\end{gathered}\end{equation}}
 Therefore\footnote{We drop the $\tau$ dependence of $Y_{CM}$, $L_{CM}$ and $M_{CM}$ in this proof for brevity}, 
 \begin{gather}
     \begin{split}
&-\tr_{\cH}(\mathcal{P}_\oplus\partial_\tau\mathcal{P}_{\Sigma_{1,1}}) \\&=- \oint_{\CP}dw \oint_{\CS}\frac{dz}{2\pi i} \frac{1}{1-e^{-2\pi i(z-w)}}\tr\left\{\Yt(z)\Yt(w)^{-1} \partial_{\tau} \left( Y_{CM}(w)\Xi_2(w,z)Y_{CM}(z)^{-1} \right) \right\} \\
&  - \oint_{\CP}dw \oint_{\CS_{out}}\frac{dz}{2\pi i} \frac{1}{1-e^{-2\pi i(z-w)}}\tr\left\{\Yt(z)\Yt(w)^{-1} \partial_{z} \left( Y_{CM}(w)\Xi_2(w,z)Y_{CM}(z)^{-1} \right) \right\} \\
&  - \oint_{\CP_{out}}dw \oint_{\CS}\frac{dz}{2\pi i} \frac{1}{1-e^{-2\pi i(z-w)}}\tr\left\{\Yt(z)\Yt(w)^{-1} \partial_{w} \left( Y_{CM}(w)\Xi_2(w,z)Y_{CM}(z)^{-1} \right) \right\} \\ 
&= -I_{\tau}- I_{z} - I_{w},
\end{split} \label{eq:11_Tr}
 \end{gather}
 where 
\begin{gather}
I_{\tau} :=  \oint_{\CP}dw  \oint_{\CS}\frac{dz}{2\pi i}  \frac{1}{1- e^{-2\pi i(z-w)}}\tr\left\{\Yt(z)\Yt(w)^{-1} \partial_{\tau} \left( Y_{CM}(w)\Xi_2(w,z)Y_{CM}(z \right)^{-1} \right\}, \label{eq:11_I_tau}\\
I_{z} :=   \oint_{\CP}dw  \oint_{\CS_{out}}\frac{dz}{2\pi i}  \frac{1}{1- e^{-2\pi i(z-w)}}\tr\left\{\Yt(z)\Yt(w)^{-1} \partial_{z} \left( Y_{CM}(w)\Xi_2(w,z)Y_{CM}(z)^{-1} \right) \right\} \label{eq:11_I_z},\\
I_{w} :=  \oint_{\CP_{out}}dw  \oint_{\CS}\frac{dz}{2\pi i}  \frac{1}{1- e^{-2\pi i(z-w)}}\tr\left\{\Yt(z)\Yt(w)^{-1} \partial_{w} \left( Y_{CM}(w)\Xi_2(w,z)Y_{CM}(z)^{-1} \right) \right\} \label{eq:11_I_w}.
\end{gather}
In the multiple integrals we always use the convention that $z$ is inside $w$ (recall that the notation $\CP$, $\CS$ is explained in Figure \ref{fig:contour_PS}) and we close the contours in the direction of $\mathscr{A}$.
The reason for such choice of the contour is the following: the
kernel  \(\left( \partial_{\tau}\mathcal{P}_{\Sigma_{1,1}} \right)(w,z)\) is regular at \(z=w\) since \(\partial_{\tau}\frac1{w-z}=0\) and \(\left( \partial_{\tau}+\partial_z+\partial_w \right)\frac1{w-z}=0\),
which means that the relative positions of the arguments of \(\partial_{\tau}\mathcal{P}_{\Sigma_{1,1}}\) can be arbitrary.
Keeping this in mind we first act on \(\left( \mathcal{P}_{\Sigma_{1,1}} \right)(w,z_0)\), viewed as a function of $w$, by \(\mathcal{P}_{\oplus}(z,w)\):
the action results in an integral over $w$, whose contour should be chosen according to Definition \ref{notation:contour}.
Namely, since \(\mathcal{P}_{\oplus}(z,w)\) has pole along the diagonal, we deform the contour for \(w\) to \(\overline{\mathcal{C}}\), and also move \(z\) to \(\underline{\mathcal{C}}\)  for convenience.
After this, we set \(z_0=z\) and integrate over $z$ on $\CS$ to take trace.

The integration of $w$ over $\CP$ then picks up the residue at $w=z$. 
Let us begin with the integral $I_{z}$:
 \begin{equation} \label{eq:11_I_z1p2}
 \begin{split}
     I_{z} &= \oint_{\CP}dw  \oint_{\CS_{out}}\frac{dz}{2\pi i}  \frac{1}{1- e^{-2\pi i(z-w)}}\tr\left\{\Yt(z)\Yt(w)^{-1} \partial_{z} \left( Y_{CM}(w)\Xi_2(w,z)Y_{CM}(z)^{-1} \right) \right\} \\
     &= I_{z}^{(1)} + I_{z}^{(2)},
     \end{split}
 \end{equation}
 where 
 \begin{gather}
      I_{z}^{(1)} := \oint_{\CP}dw  \oint_{\CS_{out}}\frac{dz}{2\pi i}  \frac{1}{1- e^{-2\pi i(z-w)}}\tr\left\{\Yt(w)^{-1} Y_{CM}(w) \partial_{z} \Xi_2(w,z) Y_{CM}(z)^{-1}  \Yt(z) \right\}, \label{eq:11_I_z1}\\
  \label{eq:11_I_z2} I_{z}^{(2)}:=  \oint_{\CP}dw  \oint_{\CS_{out}}dz  \frac{1}{1-e^{-2\pi i(z-w)}}\tr\left\{\Yt(w)^{-1}Y_{CM}(w) \Xi_2(w,z) \partial_{z}Y_{CM}(z)^{-1} \Yt(z)\right\}.
 \end{gather}
 To compute $I_{z}^{(1)}$, we expand $\Xi_2(w,z)$ as in \eqref{Xi_expansion}, and use \eqref{eq:ExpExpansion} 
 \begin{gather}
     I_{z}^{(1)} =\oint_{\CP}dw  \oint_{\CS_{out}}\frac{dz}{2\pi i}  \frac{1}{1-e^{-2\pi i(z-w)}}\tr\left\{\Yt(w)^{-1} Y_{CM}(w) \partial_{z} \Xi_2(w,z) Y_{CM}(z)^{-1}  \Yt(z) \right\} \nonumber \\
     =  \oint_{\CP}\frac{dw}{2\pi i}\oint_{\CS}\frac{dz}{2\pi i}\tr\left\{\Yt(w)^{-1}Y_{CM}(w)  \left[-\frac{1}{(w-z)^3}+\frac{i\pi}{(w-z)^2} \right] Y_{CM}(z)^{-1}\Yt(z) \right\}  \nonumber\\
+ \oint_{\CP}\frac{dw}{2\pi i}\oint_{\CS}\frac{dz}{2\pi i}\tr\left\{\frac{\Yt(w)^{-1}Y_{CM}(w)}{2(w-z)}\left[\left( \begin{array}{cc}
 \frac{\theta_1''(Q-\rho)}{\theta_1(Q-\rho)} & 0 \\ 0 &  \frac{\theta_1''(Q+\rho)}{\theta_1(Q+\rho)}
 \end{array} \right)-\frac{1}{3}\frac{\theta_1'''}{\theta_1'}-\frac{1}{6}(2\pi i)^2 \right] Y_{CM}(z)^{-1}\Yt(z) \right\} \nonumber\\
=\frac{1}{2}\oint_{\CS_{out}}\frac{dz}{2\pi i}\tr\left\{\partial_z^2\left(\Yt(z)^{-1}Y_{CM}(z) \right)Y_{CM}(z)^{-1}\Yt(z) \right\} -\frac{1}{2}\oint_{\CS_{out}}dz\tr\left\{L_{CM}-L_{3pt} \right\} \nonumber\\
 -\oint_{\CS_{out}}\frac{dz}{2\pi i}\tr\left\{\frac{1}{2}\left( \begin{array}{cc}
 \frac{\theta_1''(Q-\rho)}{\theta_1(Q-\rho)} & 0 \\ 0 &  \frac{\theta_1''(Q+\rho)}{\theta_1(Q+\rho)}
 \end{array} \right)-\frac{\mathbb{I}}{6}\frac{\theta_1'''}{\theta_1'}-\frac{(2\pi i)^2\mathbb{I}}{12} \right\} 
, \label{I_z1}
 \end{gather}
with $L_{CM}(z)$, $L_{3pt}(z)$ given in \eqref{linear_system}, \eqref{eq:3ptcyl} respectively. Similarly, $I_{z}^{(2)}$ reads
{\small \begin{gather}
     I_{z}^{(2)} = \oint_{\CP}dw  \oint_{\CS_{out}}\frac{dz}{2\pi i}  \frac{1}{1-e^{-2\pi i(z-w)}}\tr\left\{\Yt(w)^{-1} Y_{CM}(w) \Xi_2(w,z) \partial_{z}Y_{CM}(z)^{-1}  \Yt(z)\right\} \nonumber \\
     = -\oint_{\CP}\frac{dw}{2\pi i} \oint_{\CS_{out}}\frac{dz}{2\pi i}\tr\left\{  \Yt(w)^{-1}Y_{CM}(w)  \left[\frac{1}{(w-z)^2}\right] \partial_{z} Y_{CM}(z)^{-1}\Yt(z) \right\} \nonumber \\
 -\oint_{\CP}\frac{dw}{2\pi i} \oint_{\CS_{out}}\frac{dz}{2\pi i}\tr\left\{ \frac{ \Yt(w)^{-1}Y_{CM}(w) }{w-z}\left[ \left( \begin{array}{cc}
 \frac{\theta_1'(Q-\rho)}{\theta_1(Q-\rho)} & 0 \\ 0 &  -\frac{\theta_1'(Q+\rho)}{\theta_1(Q+\rho)}
 \end{array} \right)- i\pi \right] \partial_zY_{CM}(z)^{-1}\Yt(z) \right\}\nonumber \\
          =  \oint_{\CS_{out}}\frac{dz}{2\pi i} \tr\left\{\partial_{z}\left(\Yt(z)^{-1} Y_{CM}(z)\right)  \partial_{z}Y_{CM}(z)^{-1}  \Yt(z)\right\} \nonumber \\
          -\oint_{\CS_{out}}\frac{dz}{2\pi i} \tr\left\{ \left( \begin{array}{cc}
 \frac{\theta_1'(Q-\rho)}{\theta_1(Q-\rho)} & 0 \\ 0 &  -\frac{\theta_1'(Q+\rho)}{\theta_1(Q+\rho)}
 \end{array} \right) L_{CM}(z) \right\} +\frac{1}{2}  \oint_{\CS_{out}}dz \tr L_{CM}(z). \label{eq:Iz2F}
 \end{gather}}
Plugging the expressions for $I_z^{(1)}$ in \eqref{I_z1} and $I_z^{(2)}$ in \eqref{eq:Iz2F} into \eqref{eq:11_I_z1p2}, observing that $\tr L_{CM}=\tr L_{3pt}=0$, and rearranging the terms we find:
\begin{gather}
    I_{z} =  \frac{1}{2}\oint_{\CS_{out}}\frac{dz}{2\pi i}  \tr\left\{\partial_{z}^2 \left(\Yt(z)^{-1}Y_{CM}(z)   \right)   Y_{CM}(z)^{-1} \Yt(z) \right\} \nonumber\\
     +   \oint_{\CS_{out}}\frac{dz}{2\pi i}  \tr\left\{\partial_{z} \left( \Yt(z)^{-1}Y_{CM}(z)  \right)  \partial_{z}Y_{CM}(z)^{-1} \Yt(z)\right\} \nonumber \\
     -\oint_{\CS_{out}}\frac{dz}{2\pi i}\tr\left\{\frac{1}{2}\left( \begin{array}{cc}
 \frac{\theta_1''(Q-\rho)}{\theta_1(Q-\rho)} & 0 \\ 0 &  \frac{\theta_1''(Q+\rho)}{\theta_1(Q+\rho)}
 \end{array} \right)-\frac{\mathbb{I}}{6}\frac{\theta_1'''}{\theta_1'}-\frac{(2\pi i)^2\mathbb{I}}{12} \right\} .
 \label{I_z}
\end{gather}
Let us integrate by parts the first two terms in \eqref{I_z}:
\begin{gather}
   \frac{1}{2}\oint_{\CS_{out}}\frac{dz}{2\pi i}  \tr\left\{\partial_{z}^2 \left(\Yt(z)^{-1}Y_{CM}(z)   \right)   Y_{CM}(z)^{-1} \Yt(z) \right\} \nonumber\\
     + \oint_{\CS_{out}}\frac{dz}{2\pi i}  \tr\left\{\partial_{z} \left( \Yt(z)^{-1}Y_{CM}(z) \right)  \partial_{z}Y_{CM}(z)^{-1} \Yt(z)\right\} \nonumber \\
     = -\frac{1}{2}\oint_{\CS_{out}}\frac{dz}{2\pi i}  \tr\left\{\partial_{z} \left(\Yt(z)^{-1}Y_{CM}(z)   \right)   \partial_{z}\left(Y_{CM}(z)^{-1} \Yt(z)\right) \right\} \nonumber \\
     +  \oint_{\CS_{out}}\frac{dz}{2\pi i}  \tr\left\{\partial_{z} \left( \Yt(z)^{-1}Y_{CM}(z)  \right)  \partial_{z}Y_{CM}(z)^{-1} \Yt(z)\right\} \nonumber \\
     =- \frac{1}{2}\oint_{\CS_{out}}\frac{dz}{2\pi i}  \tr \left\{ -L_{3pt}(z)^2 + L_{CM}(z)^2 \right\}.\label{eq:11_I_zbp}
\end{gather}
Therefore, \begin{gather}
   -I_z =
\frac{1}{2}\oint_{\CS_{out}}\frac{dz}{2\pi i}  \tr \left\{ -L_{3pt}(z)^2 + L_{CM}(z)^2   + 2 \left( \begin{array}{cc}
 \frac{\theta_1'(Q-\rho)}{\theta_1(Q-\rho)} & 0 \\ 0 &  -\frac{\theta_1'(Q+\rho)}{\theta_1(Q+\rho)}
 \end{array} \right) L_{CM}(z)\right\}\nonumber\\
 +\frac{1}{2}\oint_{\CS_{out}}\frac{dz}{2\pi i}  \tr\left\{\left( \begin{array}{cc}
 \frac{\theta_1''(Q-\rho)}{\theta_1(Q-\rho)} & 0 \\ 0 &  \frac{\theta_1''(Q+\rho)}{\theta_1(Q+\rho)}
 \end{array} \right)-\mathbb{I}\left(\frac{1}{3}\frac{\theta_1'''}{\theta_1'}+\frac{(2\pi i)^2}{6} \right)\right\} \nonumber \\
 \mathop{=}^{\eqref{eq:Ham_CM}} -2\pi i a^2 + \frac{1}{2\pi i} H_{CM} + \frac{P}{2\pi i}\left( \frac{\theta_{1}'(Q-\rho)}{\theta_{1}(Q-\rho)} +\frac{\theta_{1}'(Q+\rho)}{\theta_{1}(Q+\rho)} \right) + \frac{1}{4\pi i}\left( \frac{\theta_{1}''(Q-\rho)}{\theta_{1}(Q-\rho)} +\frac{\theta_{1}''(Q+\rho)}{\theta_{1}(Q+\rho)}\right) \nonumber \\
 -\left(\frac{1}{6\pi i}\frac{\theta_1'''}{\theta_1'}+\frac{(2\pi i)}{6} \right). \label{eq:I_zF}
 \end{gather}
 To compute the first term in \eqref{eq:I_zF}, we use the explicit form \eqref{eq:3ptcyl}, \eqref{eq:3ptexpts} and recall that the contour $\CS_{out}$ is simply the interval $[\tau,\tau+1]$:
 \begin{equation}
     \oint_{\CS_{out}}\frac{dz}{4\pi i}  \tr L^2_{3pt}(z)=2\pi i a^2.
 \end{equation}
 The second term of \eqref{eq:I_zF} is simply the isomonodromic Hamiltonian, while all the other terms  are constants, that are unaffected by the integration. The term $I_{w}$ in \eqref{eq:11_I_w} vanishes because the $z$-loop is contractible. 
 \begin{gather}
     I_{w} =0. \label{eq:I_wF}
 \end{gather}Finally, we compute $I_{\tau}$:
 \begin{equation}
     \begin{split}
     I_{\tau} &= \oint_{\CP}dw  \oint_{\CS}\frac{dz}{2\pi i}  \frac{1}{1-e^{-2\pi i(z-w)}}\tr\left\{\Yt(z)\Yt(w)^{-1} \partial_{\tau} \left( Y_{CM}(w)\Xi_2(w,z)Y_{CM}(z)^{-1} \right) \right\} \\
     & = I_{\tau}^{(1)} + I_{\tau}^{(2)} + I_{\tau}^{(3)},
     \end{split}
 \end{equation}
 where
 \begin{gather}
     I_{\tau}^{(1)} :=  \oint_{\CP}dw  \oint_{\CS}\frac{dz}{2\pi i}  \frac{1}{1- e^{-2\pi i(z-w)}}\tr\left\{\Yt(w)^{-1} \partial_{\tau} \left( Y_{CM}(w)\right)  \Xi_2(w,z)Y_{CM}(z)^{-1}\Yt(z) \right\},  \\
     I_{\tau}^{(2)} :=  \oint_{\CP}dw  \oint_{\CS}\frac{dz}{2\pi i}  \frac{1}{1- e^{-2\pi i(z-w)}}\tr\left\{\Yt(w)^{-1} Y_{CM}(w) \partial_{\tau} \left( \Xi_2(w,z) \right)  Y_{CM}(z)^{-1}\Yt(z) \right\} \\
     I_{\tau}^{(3)} :=  \oint_{\CP}dw  \oint_{\CS}\frac{dz}{2\pi i}  \frac{1}{1- e^{-2\pi i(z-w)}}\tr\left\{\Yt(w)^{-1} Y_{CM}(w)\Xi_2(w,z) \partial_{\tau} \left( Y_{CM}(z)^{-1} \right)\Yt(z) \right\}.
 \end{gather}
 Expanding $\Xi_2(z,w)$ as in \eqref{Xi_expansion} and using \eqref{eq:ExpExpansion},  $I_{\tau}^{(1)}$ reads,
 \begin{gather}
     I_{\tau}^{(1)} =  \oint_{\CP}dw  \oint_{\CS}\frac{dz}{2\pi i}  \frac{1}{1-e^{-2\pi i(z-w)}}\tr\left\{\Yt(w)^{-1} \partial_{\tau} \left( Y_{CM}(w)\right)  \Xi_2(w,z)Y_{CM}(z)^{-1} \Yt(z)\right\} \nonumber \\
     = -\oint_{\CP}\frac{dw}{2\pi i}  \oint_{\CS}\frac{dz}{2\pi i}\tr\left\{ \frac{\Yt(w)^{-1} \partial_{\tau} \left( Y_{CM}(w)\right)Y_{CM}(z)^{-1} \Yt(z)}{(w-z)^2} \right\} \nonumber\\
+\oint_{\CP}\frac{dw}{2\pi i}  \oint_{\CS}\frac{dz}{2\pi i}\tr\left\{ \frac{\Yt(w)^{-1} \partial_{\tau} \left( Y_{CM}(w)\right)}{w-z}\left[i\pi\mathbb{I}-\left( \begin{array}{cc}
 \frac{\theta_1'(Q-\rho)}{\theta_1(Q-\rho)} & 0 \\ 0 &  \frac{\theta_1'(-Q-\rho)}{\theta_1(-Q-\rho)}
 \end{array} \right) \right] Y_{CM}(z)^{-1}\Yt(z)\right\} \nonumber\\
= -\oint_{\CS} \frac{dz}{2\pi i}\tr\left\{\left[i\pi\mathbb{I}-\left( \begin{array}{cc}
 \frac{\theta_1'(Q-\rho)}{\theta_1(Q-\rho)} & 0 \\ 0 &  \frac{\theta_1'(-Q-\rho)}{\theta_1(-Q-\rho)}
 \end{array} \right)  \right]  M_{CM}\right\}=0,  \label{I_tau1}
   \end{gather}
where $M_{CM}$ is the matrix in equation \eqref{eq:linear_systemCM}. In the last line we use the fact that $z$ lies inside the contour of $w$, and $M_{CM}$ has no residue at the puncture $z=0$. Now computing the integral $I_{\tau}^{(2)}$,
 \begin{gather}
      I_{\tau}^{(2)} = \oint_{\CP}dw  \oint_{\CS}\frac{dz}{2\pi i}  \frac{1}{1-e^{-2\pi i(z-w)}}\tr\left\{\Yt(z)\Yt(w)^{-1} Y_{CM}(w) \partial_{\tau} \left( \Xi_2(w,z) \right)  Y_{CM}(z)^{-1} \right\}=0, \label{I_tau2}
 \end{gather}
because $\partial_\tau\Xi_2(z,w)$ is regular at $w=z$. Since $I_{\tau}^{(1)}$ and $I_{\tau}^{(2)}$ vanish, $I_{\tau} = I_{\tau}^{(3)}$. Finally, we compute the integral $I_{\tau}^{(3)}$ by expanding $\Xi$ as before:
\begin{gather}
    I_{\tau} = I_{\tau}^{(3)} = \oint_{\CP}dw  \oint_{\CS}\frac{dz}{2\pi i}  \frac{1}{1-e^{-2\pi i(z-w)}}\tr\left\{\Yt(w)^{-1} Y_{CM}(w)\Xi_2(w,z) \partial_{\tau} Y_{CM}(z)^{-1} \Yt(z) \right\}\nonumber\\
     = -\oint_{\CP}\frac{dw}{2\pi i}  \oint_{\CS}\frac{dz}{2\pi i}  \tr\left\{\frac{\Yt(w)^{-1} Y_{CM}(w)  \partial_{\tau} \left( Y_{CM}(z)^{-1} \right) \Yt(z)}{(w-z)^2} \right\} \nonumber \\
     - \oint_{\CP}\frac{dw}{2\pi i}  \oint_{\CS}\frac{dz}{2\pi i}  \tr\left\{\frac{\Yt(w)^{-1} Y_{CM}(w)}{w-z} \,   \left[i \pi \mathbb{I} - \left( \begin{array}{cc}
 \frac{\theta_1'(Q-\rho)}{\theta_1(Q-\rho)} & 0 \\ 0 &  \frac{\theta_1'(-Q-\rho)}{\theta_1(-Q-\rho)}
 \end{array} \right) \right]\partial_{\tau} \left( Y_{CM}(z)^{-1} \right) \Yt(z) \right\} \nonumber \\
     =-\oint_{\CS}\frac{dz}{2\pi i}  \tr\left\{\partial_{z} \left(\Yt(z)^{-1} Y_{CM}(z)\right) \partial_{\tau} \left( Y_{CM}(z)^{-1} \right) \Yt(z) \right\} \nonumber \\
     +  \oint_{\CS}dz  \tr\left\{ \left[i\pi \mathbb{I} -\left( \begin{array}{cc}
 \frac{\theta_1'(Q-\rho)}{\theta_1(Q-\rho)} & 0 \\ 0 &  \frac{\theta_1'(-Q-\rho)}{\theta_1(-Q-\rho)}
 \end{array} \right)\right] M_{CM}  \right\}\nonumber \\
 = \oint_{\CS}\frac{dz}{2\pi i}  \tr\left\{\partial_{z} \left(\Yt(z)^{-1} Y_{CM}(z)\right) \partial_{\tau} \left( Y_{CM}(z)^{-1} \right) \Yt(z) \right\}. \label{I_tau3}
 \end{gather}
Again, in the last line of \eqref{I_tau3} we used the fact that $M_{CM}$ in \eqref{linear_system}, is regular at the puncture. Therefore,
 \begin{gather}
     -I_\tau =-\oint_{\CS}dz  \tr\left\{\partial_{z} \left(\Yt(z)^{-1} Y_{CM}(z)\right) \partial_{\tau} \left( Y_{CM}(z)^{-1} \right) \Yt(z) \right\} \nonumber \\
     =-\oint_{\CS}dz \tr\left\{\left(\partial_z\Yt(z)\Yt(z)^{-1}-\partial_z Y_{CM}(z)Y_{CM}(z)^{-1} \right)\partial_\tau Y_{CM}(z)Y_{CM}(z)^{-1}\right\}\nonumber \\
 =-\oint_{\CS}dz \left\{\left(Y_{CM}(z)^{-1} \Yt(z)L_{3pt}(z)(Y_{CM}(z)^{-1}\Yt(z))^{-1}-L_{CM} \right)M_{CM} \right\}. \label{eq:I_tau}
 \end{gather}
To compute the above expression, we study the behavior of $L_{CM},M_{CM},L_{3pt}$  in \eqref{linear_system} and \eqref{eq:3ptcyl} respectively, at $z=0$, using the expansions
  \begin{gather}
     x(2Q,z) = \frac{1}{z} + \frac{\theta_{1}'(2Q)}{\theta_{1}(2Q)} + \mathcal{O}(z), \label{expandx}
 \end{gather}
 and
 \begin{equation}
 \begin{split}
     y(2Q,z) & = 
     \left[ \frac{\theta_{1}''(2Q)}{\theta_{1}(2Q)} - \left(\frac{\theta_{1}'(2Q)}{\theta_{1}(2Q)}\right)^2 \right] + \mathcal{O}(z) \\
     & = \wp(2Q) + \mathcal{O}(z)
     \label{expandy}.
     \end{split}
 \end{equation}
 Substituting \eqref{expandx} and \eqref{expandy} in the Lax matrices one finds (the solutions $Y_{CM}(z)$, $\Yt(z)$ can be simultaneously re-normalized in such a way that their monodromy around $z=0$ is $e^{2\pi i m\sigma_3}$)
 \begin{gather}
     L_{CM} = \frac{Y_{CM}(0)^{-1}m\sigma_3Y_{CM}(0)}{z} - i \frac{m\theta_{1}'(2Q)}{\theta_{1}(2Q)} \sigma_{2} + P \sigma_{3} + \mathcal{O}(z), \nonumber\\ 
     M_{CM} = 
     \wp(2Q) \sigma_{1} + \mathcal{O}(z) , \label{eq:LMLexp} \\
     L_{3pt}=\frac{\Yt(0)^{-1}m\sigma_3\Yt(0) }{z}-2\pi i\left(A_0+\frac{A_1}{2}\right)+\mathcal{O}(z)\nonumber.
 \end{gather}
From equation \eqref{eq:LMLexp}, it follows that 
\begin{gather}
    Y_{CM}(z)^{-1} \Yt(z)L_{3pt}(z)(Y_{CM}(z)^{-1}\Yt(z))^{-1}-L_{CM}\nonumber\\
    =-2\pi iY_{CM}(0)^{-1} \Yt(0)\left(A_0+\frac{A_1}{2}\right)Y_{CM}(0)^{-1} \Yt(0)+ i \frac{m\theta_{1}'(2Q)}{\theta_{1}(2Q)} \sigma_{2} - P \sigma_{3},
\end{gather}
so, the integrand in equation \eqref{eq:I_tau} has no pole, and
 \begin{equation}
 I_\tau=0.\label{eq:I_tauF}
 \end{equation}
We have thus shown that the logarithmic derivative of the tau function $\T^{(1,1)}$ in \eqref{eq:11Tr} is
\begin{equation}
\begin{split}
 2\pi i \partial_{\tau} \log&\det\left[\mathbb{I}- K_{1,1}  \right] \mathop{=}^{\eqref{eq:propdetT11}} 2\pi i\partial_{\tau} \log \T^{(1,1)}    \mathop{=}^{\eqref{eq:11Tr}} -2\pi i\tr_{\cH} \cP_{\oplus} \partial_{\tau} \cP_{\Sigma} \\
 &\mathop{=}^{\eqref{eq:11_Tr}}2\pi i\left( - I_{\tau} - I_{z} - I_{w}\right)  \mathop{=}^{\eqref{eq:I_wF}, \eqref{eq:I_tauF}} - 2\pi i I_{z} \\
    & \mathop{=}^{\eqref{eq:I_zF}}   -(2\pi i)^{2} a^2 + H_{CM} -\left(\frac{1}{3}\frac{\theta_1'''}{\theta_1'}+\frac{(2\pi i)^2}{6} \right)  \\
 & + P\left( \frac{\theta_{1}'(Q-\rho)}{\theta_{1}(Q-\rho)} +\frac{\theta_{1}'(Q+\rho)}{\theta_{1}(Q+\rho)} \right) + \frac{1}{2}\left( \frac{\theta_{1}''(Q-\rho)}{\theta_{1}(Q-\rho)} +\frac{\theta_{1}''(Q+\rho)}{\theta_{1}(Q+\rho)}\right) \\
& \mathop{=}^{\eqref{eq:IsomTau11}}   2\pi i \partial_{\tau} \log \T_{CM}-(2\pi i)^2 a^2-\left(\frac{1}{3}\frac{\theta_1'''}{\theta_1'}+\frac{(2\pi i)^2}{6} \right)    \\
 & + P\left( \frac{\theta_{1}'(Q-\rho)}{\theta_{1}(Q-\rho)} +\frac{\theta_{1}'(Q+\rho)}{\theta_{1}(Q+\rho)} \right) + \frac{1}{2}\left( \frac{\theta_{1}''(Q-\rho)}{\theta_{1}(Q-\rho)} +\frac{\theta_{1}''(Q+\rho)}{\theta_{1}(Q+\rho)}\right) \\
 & =  2\pi i \partial_{\tau} \log \T_{CM}-(2\pi i)^{2} a^2 -\frac{(2\pi i)^2}{6} + 2\pi i \frac{d}{d\tau} \log \left( \frac{\theta_{1}(Q-\rho) \theta_{1}(Q+\rho)}{\eta(\tau)^2}\right).
 \end{split}\label{proof:thm1}
\end{equation}
In the last line we used the heat equation for $\theta_1$
\begin{equation}
    4\pi i\partial_\tau\theta_1(Q\pm\rho)=\theta_1''(Q\pm\rho),\label{eq:heat_theta1}
\end{equation}
as well as the fact that $P=2\pi i\frac{dQ}{d\tau}$, leading to
\begin{equation}
\begin{split}
    4\pi i\frac{d}{d\tau}\log\theta_1(Q\pm\rho) & =\frac{1}{\theta_1(Q\pm\rho)}
4\pi i\frac{d}{d\tau}\theta_1(Q\pm\rho)\\
&=4\pi i\frac{dQ}{d\tau}\frac{\theta_1'(Q\pm\rho)}{\theta_1(Q\pm\rho)}+\frac{1}{\theta_1(Q\pm\rho)}\partial_\tau\theta_1(Q\pm\rho)\\
&=2P\frac{\theta_1'(Q\pm\rho)}{\theta_1(Q\pm\rho)}+\frac{\theta_1''(Q\pm\rho)}{\theta_1(Q\pm\rho)},
\end{split}
\end{equation}
and the expression for Dedekind eta function
\begin{align}
    \eta(\tau)=\left(\frac{\theta_1'}{2\pi}\right)^{1/3}, && -4\pi i\frac{d}{d\tau}\log\eta(\tau)=-\frac{1}{3\theta_1'} 4\pi i\frac{d\theta_1'}{d\tau}=-\frac{1}{3}\frac{\theta_1'''}{\theta_1'}. \label{eq:eta_id}
\end{align}
Integrating \eqref{proof:thm1} on both sides, we obtain \eqref{eq:prop1} where the explicit form of the kernel $K_{1,1}$ is in \eqref{eq:kerker_11}.
\end{proof}
\begin{remark}\label{rmk:zerostau11} 
Due to the factor $\frac{\theta_{1}(Q+\rho)\theta_1(Q-\rho)}{\eta(\tau)^2}$ in \eqref{eq:prop1} between the isomonodromic tau function $\T_{CM}$ and the determinant tau function $\T^{(1,1)}$, we have the following statement:
 \begin{equation}
     \T^{(1,1)}\vert_{\rho=\pm Q}=0, \label{eq:remark2}
 \end{equation}
 i.e. the zero locus of the Fredholm determinant in $\rho$ computes the solution to the equation \eqref{eq:EllPainleve}. This is an isomonodromic version of Krichever's solution of the isospectral elliptic Calogero-Moser model \cite{Krichever1980,Bonelli:2019boe,Bonelli:2019yjd}, and justifies the introduction of the extra parameter $\rho$.
 \end{remark}
 Further comments are in order:

\begin{itemize}
\item[1.] We see that, in contrast to the spherical case, now there are two different tau functions, \(\mathcal{T}_{CM}\) and \(\mathcal{T}^{(1,1)}\).
It is usually supposed that the object called 'tau function' is related to free fermions, has a determinant representation, and satisfies some bilinear relations.
It turns out that only \(\mathcal{T}^{(1,1)}\) has such properties, in particular, it was shown in \cite{BGG} that equation \eqref{eq:EllPainleve} is equivalent to some bilinear relations on the two \(\rho\)-independent parts of \(\mathcal{T}^{(1,1)}\).
These bilinear relations are the consequences of the blow-up relations for the theory with adjoint matter (for other examples of such equations see \cite{Gu:2019pqj}). The
free-fermionic nature of \(\mathcal{T}^{(1,1)}\) was shown in \cite{Bonelli:2019boe}.
Instead, \(\mathcal{T}_{CM}\) has one property, which \(\mathcal{T}^{(1,1)}\) does not have: its derivative gives the Hamiltonian.

\item[2.] A consequence of the determinant expression \eqref{eq:prop1} for $\T_{CM}$ is the following relation between the solution $Q$ to the nonautonomous elliptic Calogero-Moser equation and the determinant (see equations 3.56 and F.3, \cite{Bonelli:2019boe}):
   \begin{gather}
        \frac{\theta_{3}(2Q(\tau)\vert 2\tau)}{\theta_{2}(2Q(\tau)\vert 2\tau)} = i e^{3 i \pi \tau/2} \frac{\det\left( \mathbb{I}- K_{1,1}\vert_{\rho = \frac{1}{4} + \frac{\tau}{2}} \right)}{\det\left( \mathbb{I}- K_{1,1}\vert_{\rho = \frac{1}{4}} \right)}.
    \end{gather}
    \end{itemize}

\section{Generalization to the $n$-punctured torus}\label{sec:GL}
We now generalize the discussion of the previous section to the $GL(N)$ linear system \eqref{lax_general} on a torus with $n$ punctures, using the expressions derived in \cite{Levin:2013kca} for the matrices $L,M_{z_k},M_{\tau}$. In this case the matrix elements $L_{ij}$ of the Lax matrix $L(z)$ are
\begin{equation}\label{eq:nptL}
\begin{split}
    L_{ij}(z) & =\delta_{ij}\left[P_i+\sum_{k=1}^n\frac{\theta_1'(z-z_k)}{\theta_1(z-z_k)}\left(S^{(k)}_{ii}+\Lambda_{k} \right) \right]\\
    &+(1-\delta_{ij})\sum_{k=1}^n\frac{\theta_1(z-z_k+Q_i-Q_j)\theta_1'(0)}{\theta_1(z-z_k)\theta_1(Q_i-Q_j)}S_{ij}^{(k)}
\end{split}
\end{equation}
while the matrix elements of the $M$-matrices \eqref{lax_general} are
\eqs{\label{eq:nptMk}
    (M_{z_k})_{ij}(z)=-\delta_{ij}\frac{\theta_1'(z-z_k)}{\theta_1(z-z_k)}\left(S_{ii}^{(k)}+\Lambda_{k}\right)- (1-\delta_{ij})\frac{\theta_1(z-z_k+Q_i-Q_j)\theta_1'(0)}{\theta_1(z-z_k)\theta_1(Q_i-Q_j)}S_{ij}^{(k)},}
\eq{\label{eq:nptMt}(M_\tau)_{ij}(z)=\frac{1}{2}\delta_{ij}\sum_{k=1}^n\frac{\theta_1''(z-z_k)}{\theta_1(z-z_k)}\left(S_{ii}^{(k)}+\Lambda_{k}\right) +\sum_{k=1}^ny(Q_j-Q_i,z-z_k)S_{ij}^{(k)},
}
where the function $y(u,z)$ is defined in \eqref{eq:xyP}. The dynamical variables\footnote{In the interest of brevity, we omit writing the $\tau$, $z_{1}... z_{n}$ dependence of the functions $L(z), M(z), Y(z)$ and the dynamical variables $Q_{i}$'s, $P_{i}$'s.} $Q_1,\dots,Q_N$, $P_1,\dots,P_N$ satisfy $\sum_iQ_i=\sum_iP_i=0$ and are canonically conjugated, and the matrices $S^{(k)}$ satisfy the Kirillov-Kostant Poisson bracket
\begin{align}
    \left\{Q_i,P_j \right\}=\delta_{ij}, && \left\{S^{(k)}_a,S^{(m)}_b \right\}=\delta^{km}f{_{ab}}^cS^{(k)}_c,
\end{align}
where we defined $S^{(k)}:=S^{(k)}_at^a$ in terms of a set of generators $t_a$ of $\mathfrak{sl}_N$, and $f{_{ab}}^{c}$ are the $\mathfrak{sl}_N$ structure constants. The residues take value in $\mathfrak{gl}(N)$ due to the $U(1)$ factors $\Lambda_{k}$.
\begin{figure}[h]
\begin{center}
\begin{tikzpicture}[scale=0.9]

\draw[thick,decoration={markings, mark=at position 0.75 with {\arrow{>}}}, postaction={decorate}] (-5,1.5) circle[x radius=0.5, y radius=0.2];
\draw(-6,0.5) to[out=0,in=270] (-5.5,1.5);
\draw(-4,0.5) to[out=180,in=270] (-4.5,1.5);
\draw(-6,-0.5) to[out=0,in=180] (-4,-0.5);
\draw[thick, red,decoration={markings, mark=at position 0.6 with {\arrow{>}}}, postaction={decorate}]  (-4,0.5) to[out=10,in=-10] (-4,-0.5);
\draw[thick,dashed, red] (-4,0.5) to[out=220,in=-220] (-4,-0.5);
\draw[thick,red,decoration={markings, mark=at position 0 with {\arrow{>}}}, postaction={decorate}] (-6,0) circle[y radius=0.5, x radius=0.2];

\draw[thick,blue,decoration={markings, mark=at position 0 with {\arrow{<}}}, postaction={decorate}] (-3.3,0) circle[y radius=0.5, x radius =0.2];
\draw (-3.3,-0.5) to (-2.7,-0.5);
\draw (-3.3,0.5) to (-2.7,0.5);
\draw[thick, blue,decoration={markings, mark=at position 0.6 with {\arrow{<}}}, postaction={decorate}]  (-2.7,0.5) to[out=10,in=-10] (-2.7,-0.5);
\draw[dashed, blue] (-2.7,0.5) to[out=220,in=-220] (-2.7,-0.5);

\draw[thick,decoration={markings, mark=at position 0.75 with {\arrow{>}}}, postaction={decorate}] (-1,1.5) circle[x radius=0.5, y radius=0.2];
\draw(-2,0.5) to[out=0,in=270] (-1.5,1.5);
\draw(0,0.5) to[out=180,in=270] (-0.5,1.5);
\draw(-2,-0.5) to[out=0,in=180] (0,-0.5);
\draw[thick,red,decoration={markings, mark=at position 0 with {\arrow{>}}}, postaction={decorate}] (-2,0) circle[y radius=0.5, x radius=0.2];
\draw[thick, red,decoration={markings, mark=at position 0.6 with {\arrow{>}}}, postaction={decorate}]  (0,0.5) to[out=10,in=-10] (0,-0.5);
\draw[thick,dashed, red] (0,0.5) to[out=220,in=-220] (0,-0.5);

\draw[thick,blue,decoration={markings, mark=at position 0 with {\arrow{<}}}, postaction={decorate}] (0.7,0) circle[y radius=0.5, x radius =0.2];
\draw (0.7,-0.5) to (1.3,-0.5);
\draw (0.7,0.5) to (1.3,0.5);
\draw[thick, blue,decoration={markings, mark=at position 0.6 with {\arrow{<}}}, postaction={decorate}]  (1.3,0.5) to[out=10,in=-10] (1.3,-0.5);
\draw[dashed, blue] (1.3,0.5) to[out=220,in=-220] (1.3,-0.5);

\node at (2.3,0) {{\Large\dots}};

\draw[thick,blue,decoration={markings, mark=at position 0 with {\arrow{<}}}, postaction={decorate}] (3,0) circle[y radius=0.5, x radius =0.2];
\draw (3,-0.5) to (3.6,-0.5);
\draw (3,0.5) to (3.6,0.5);
\draw[thick, blue,decoration={markings, mark=at position 0.6 with {\arrow{<}}}, postaction={decorate}]  (3.6,0.5) to[out=10,in=-10] (3.6,-0.5);
\draw[dashed, blue] (3.6,0.5) to[out=220,in=-220] (3.6,-0.5);

\draw[thick,decoration={markings, mark=at position 0.75 with {\arrow{>}}}, postaction={decorate}] (5.3,1.5) circle[x radius=0.5, y radius=0.2];
\draw(4.3,0.5) to[out=0,in=270] (4.8,1.5);
\draw(6.3,0.5) to[out=180,in=270] (5.8,1.5);
\draw(4.3,-0.5) to[out=0,in=180] (6.3,-0.5);
\draw[thick,red,decoration={markings, mark=at position 0 with {\arrow{>}}}, postaction={decorate}] (4.3,0) circle[y radius=0.5, x radius=0.2];
\draw[thick, red,decoration={markings, mark=at position 0.6 with {\arrow{>}}}, postaction={decorate}]  (6.3,0.5) to[out=10,in=-10] (6.3,-0.5);
\draw[thick,dashed, red] (6.3,0.5) to[out=220,in=-220] (6.3,-0.5);

\draw[thick, blue,decoration={markings, mark=at position 0.6 with {\arrow{<}}}, postaction={decorate}]  (-6.7,0.5) to[out=10,in=-10] (-6.7,-0.5);
\draw[dashed, blue] (-6.7,0.5) to[out=220,in=-220] (-6.7,-0.5);

\draw[thick,blue,decoration={markings, mark=at position 0 with {\arrow{<}}}, postaction={decorate}] (7,0) circle[y radius=0.5, x radius =0.2];

\draw[thick] (-6.7,0.5) to[out=180,in=90] (-8,-0.5) to[out=270,in=180] (0.15,-3.5) to[out=0,in=270] (8.3,-0.5) to[out=90,in=0] (7,0.5);

\draw[thick] (-6.7,-0.5) to[out=220,in=170] (-5,-2.2) to[out=-10,in=180] (0.15,-2.5) to[out=0,in=190] (5.3,-2.2) to[out=10,in=-30] (7,-0.5);


\node at (-5,0.2) {{\Large$\mathscr{T}^{[1]}$}};
\node at (-1,0.2) {{\Large$\mathscr{T}^{[2]}$}};
\node at (5.3,0.2) {{\Large$\mathscr{T}^{[n]}$}};

\node at (-6,-0.9) {{\footnotesize${\color{red} \underline{\mathcal{C}}_{in}^{[1]}}$}};
\node at (-4,-0.9) {{\footnotesize${\color{red} \underline{\mathcal{C}}_{out}^{[1]}}$}};

\node at (-3.4,0.9) {{\footnotesize${\color{blue} \overline{\mathcal{C}}_{out}^{[1]}}$}};
\node at (-2.5,0.9) {{\footnotesize${\color{blue} \overline{\mathcal{C}}_{in}^{[2]}}$}};

\node at (-2,-0.9) {{\footnotesize${\color{red} \underline{\mathcal{C}}_{in}^{[2]}}$}};
\node at (0,-0.9) {{\footnotesize${\color{red} \underline{\mathcal{C}}_{out}^{[2]}}$}};

\node at (0.5,0.9) {{\footnotesize${\color{blue} \overline{\mathcal{C}}_{out}^{[2]}}$}};
\node at (1.4,0.9) {{\footnotesize${\color{blue} \overline{\mathcal{C}}_{in}^{[3]}}$}};

\node at (2.9,0.9) {{\footnotesize${\color{blue} \overline{\mathcal{C}}_{out}^{[n-1]}}$}};
\node at (3.8,0.9) {{\footnotesize${\color{blue} \overline{\mathcal{C}}_{in}^{[n]}}$}};

\node at (4.3,-0.9) {{\footnotesize${\color{red} \underline{\mathcal{C}}_{in}^{[n]}}$}};
\node at (6.3,-0.9) {\footnotesize{${\color{red} \underline{\mathcal{C}}_{out}^{[n]}}$}};

\node at (7,0.9) {{\footnotesize${\color{blue} \overline{\mathcal{C}}_{out}^{[n]}}$}};
\node at (-7,0.9) {{\footnotesize${\color{blue} \overline{\mathcal{C}}_{in}^{[n+1]}=\overline{\mathcal{C}}_{in}^{[1]} }$}};

\node at (-3,-0.9) {{\large$\mathscr{A}^{[1]}$}};
\node at (1,-0.9) {{\large$\mathscr{A}^{[2]}$}};
\node at (3.2,-0.9) {{\large$\mathscr{A}^{[n-1]}$}};
\node at (0.15,-3) {{\large$\mathscr{A}^{[n]}$}};

\end{tikzpicture}
\end{center}
\caption{Pants decomposition for the $n$-punctured torus}\label{fig:nptTorusPants}
\end{figure}
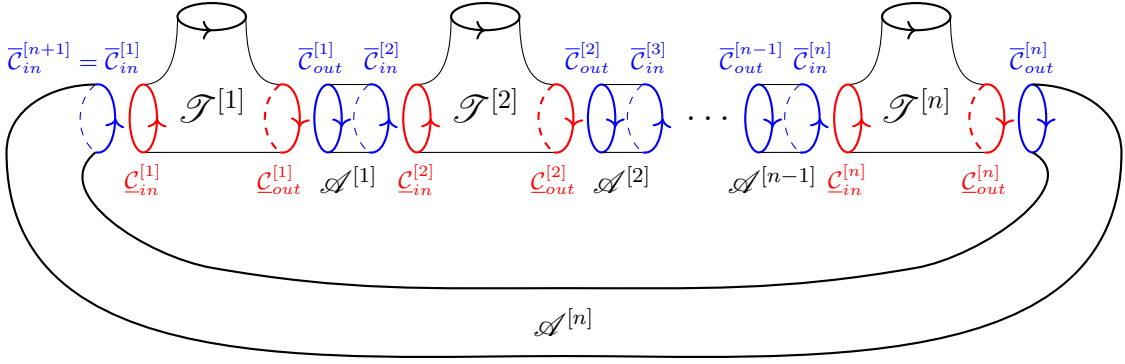
\begin{notation}\label{not:bold}
 Given an N-tuple of parameters $(\xi_1,\dots\xi_N)$, and a function $g(\xi_i)$, $i=1,\dots,N$ of these parameters, we define
 \begin{equation}
     g(\bs\xi):=\diag\left(f(\xi_1),\dots,g(\xi_N)\right).
 \end{equation}
 In particular, when $g(\xi_i)=\xi_i$, this is
 \begin{equation}
     \bs\xi=\diag\left(\xi_1,\dots,\xi_N \right).
 \end{equation}
\end{notation}
\begin{remark}\label{rmk:GLNSLN}
The generic isomonodromic problem on genus one surfaces is formulated in \cite{Levin:2013kca} under the requirement that the matrices $S^{(k)}$, parametrizing the $\mathfrak{sl}_N$ residues at the punctures $z_k$, satisfy
\begin{equation}\label{cons:CONSTRAINT}
    \sum_{k=1}^nS_{ii}^{(k)}=0.
\end{equation}
For consistency of the construction, \eqref{cons:CONSTRAINT} will be imposed on the $\mathfrak{sl}_N$ component of the residues, $S^{(k)}$.
\end{remark}
The matrices $L,M_{z_k},M_{\tau}$ are not single-valued on the torus, but rather under the shift $z\rightarrow z+\tau$ behave as (using notation \ref{not:bold})
\begin{gather}
    L(z+\tau)=e^{-2\pi i\bs Q}L(z)e^{2\pi i\bs Q}-2\pi i\sum_{k=1}^n\Lambda_{k} \nonumber \\
    M_{z_k}(z+\tau)=e^{-2\pi i\bs Q}M_{z_k}(z)e^{2\pi i\bs Q}+2\pi i\sum_{k=1}^n\Lambda_{k}, \label{eq:LMMTransform} \\
        M_\tau(z+\tau)=e^{-2\pi i\bs Q}\left(M_\tau(z)+L(z) \right)e^{2\pi\bs Q}-2\pi i\bs P-(2\pi i)^2\frac{1}{2}\sum_{k=1}^n\Lambda_{k} \nonumber, 
\end{gather}
so that the solution of the linear system \eqref{linear_system} will transform as follows:
\begin{equation}
   \mathcal{Y}(z+\tau)=M_Be^{-2\pi i\sum_{j=1}^n\left(z-z_j+\frac{\tau}{2}+\frac12\right)\Lambda_j}\mathcal{Y}(z)e^{2\pi i\bs Q}, \label{eq:zdep_MB}
\end{equation}
where $M_B\in SL(N)$. The pants decomposition corresponding to the $n$-punctured torus consists of $n$ trinions, as shown in Figure \ref{fig:nptTorusPants}, with each trinion  $\mathscr{T}^{[k]}$ associated to its own three-point problem. 
\begin{gather}
    \partial_z\widetilde{\mathcal{Y}}^{[k]}(z)=\widetilde{\mathcal{Y}}^{[k]}(z) L_{3pt}^{[k]}(z), \nonumber\\ L_{3pt}^{[k]}(z)=-2\pi i A_0^{[k]}-2\pi i\frac{A_1^{[k]}}{1-e^{2\pi iz}},\label{eq:3ptk}
\end{gather}
where
\begin{align}
    A_0^{[k]}\sim\bs\sigma_k, && A_1^{[k]}\sim\bs\mu_k,
\end{align}
\begin{align}
    \bs \sigma_{k} = {\bs a_k} - \sum_{j=0}^{k-1} \Lambda_j \mathbb{I}, && \bs \mu_{k} = {\bs m_{k}} + \Lambda_{k} \mathbb{I} \label{eq:sigma_mu_def}
\end{align}
for $k =1, \dots n$.
As in the 1-point case, we choose \(\widetilde{\mathcal{Y}}^{[k]}(z)\) in such a way that
\begin{equation*}
\widetilde{\mathcal{Y}}^{[k]}(z)^{-1}\mathcal{Y}(z)
\end{equation*}
is regular and single-valued around \(z=z_k\) and has no monodromies around two closest A-cycles.

In \eqref{eq:sigma_mu_def} we introduced a $U(1)$ parameter $\Lambda_0$ shifting the monodromy exponent $\bs\sigma_1$ around $\cC_{in}^{[1]}$, whose significance will become apparent in sections \ref{subsec:Minor_CM} and \ref{subsec:Minor_GL}\footnote{From the point of view of the dynamical system, the monodromy exponents on $\cC_{in}^{[1]}$ have the role of initial conditions, so that it is natural that $\Lambda_0$ doesn't appear in the Lax matrix, contrary to $\Lambda_1,\dots,\Lambda_N$, which are residues at the punctures.}. The monodromy exponents $\bs m_k$, $\bs a_k$ parametrize the $SL(N)$ component of the monodromy, and the $\bs a_k$'s satisfy $\bs a_{n+1} = \bs a_{1}$. 
In terms of the original problem on the torus, the monodromy exponents $\bs \sigma_{k}$, $\bs \mu_{k}$ in equation \eqref{eq:3ptk} are defined by the conjugacy class of the monodromies around the punctures $\left\lbrace z_k\right\rbrace_{k=1}^{n}$, and around the circles $\cC_{in}$, $\cC_{out}$ being glued in the pants decomposition (see Figure \ref{fig:nptTorusPants}), which are respectively 
\begin{align}
    M_k\sim e^{2\pi i\,\bs \mu_k}, && M_{\mathcal{C}_{in}^{[k]}}=M_{\mathcal{C}_{out}^{[k-1]} }^{-1}=G_k^{-1} e^{2\pi i\, \bs \sigma_k}G_k, \label{eq:1nmonodromy}
\end{align}
for $k= 1,\dots n $, and
\begin{align}
    M_{\cC_{out}^{[n]}}  =M_B^{-1} e^{-2\pi i (\bs \sigma_{1} - \sum_{j=1}^{n} \Lambda_{j} \mathbb{I}) }M_B. 
    \label{eq:monodromies1_n}
\end{align} 
The matrix $G_k$ is the matrix that diagonalizes $M_{\cC_{in}^{[k]}}=M_{\cC_{out}}^{[k-1]}$, while $G_{n+1}:=M_B$ is the matrix that diagonalizes $M_{\cC_{out}}^{[n]}$ as in the one-punctured case, and we fixed $G_1=\mathbb{I}$. The total Hilbert space $\cH$ is decomposed into a direct sum of spaces $\cH^{[k]}$ corresponding to each pair of pants:
\eqs{\cH := \bigoplus_{k=1}^{n} \cH^{[k]} = \cH_{+} \oplus \cH_{-}, }
where
\eqs{\cH_{\pm}:=\bigoplus_{k=1}^{n} \left( \cH^{[k]}_{in, \mp} \oplus \cH^{[k]}_{out, \pm} \right),}
\begin{definition}\label{def:YcalYnpoint}
 Corresponding to the solutions $\mathcal{Y}(z)$, $\widetilde{\mathcal{Y}}^{[k]}(z)$ of the linear problems \eqref{lax_general}, \eqref{eq:3ptk} respectively, we define two matrix-valued functions: $Y(z)$ with diagonal monodromies around the boundary circles $\cC^{[k]}_{in}$ and $\cC^{[k]}_{out}$, and $\Yt^{[k]}(z)$ with diagonal monodromies around $\cC^{[k]}_{in}$ and $\cC^{[k]}_{out}$ (see Figure \ref{fig:nptTorusPants}), by the following equations:
 \begin{align}
    Y(z)\vert_{\cC^{[1]}_{in}}:=\mathcal{Y}(z)\vert_{\cC^{[1]}_{in}}\in\cH^{[1]}_{in}, && Y(z)\vert_{\cC^{[n]}_{out}}:=M_B^{-1}\mathcal{Y}(z)\vert_{\cC^{[n]}_{out}}\in\cH^{[n]}_{out}.
\end{align}
 \begin{align}
    \Yt^{[k]}(z)\vert_{\cC^{[k]}_{in}}\equiv \Yt^{[k]}_{in}(z):=G_{k}^{-1}\widetilde{\mathcal{Y}}^{[k]}(z)\vert_{\cC^{[k]}_{in}}\in\cH^{[k]}_{in}, \\ \Yt^{[k]}(z)\vert_{\cC^{[k]}_{out}}\equiv \Yt^{[k]}_{out}(z):=G_{k+1}^{-1}\widetilde{\mathcal{Y}}^{[k]}(z)\vert_{\cC^{[k]}_{out}}\in\cH^{[k]}_{out},
\end{align}
with $G_{1}= \mathbb{I}$, and $G_{n+1}= M_{B}$.
\end{definition}
The functions $f^{[k]}(z) \in \cH^{[k]}$ are decomposed as
\eqs{f^{[k]}(z)= \left(\begin{array}{c}
     f^{[k]}_{in,-}  \\
     f^{[k]}_{out,+}
\end{array}  \right) \oplus \left(\begin{array}{c}
     f^{[k]}_{in,+}  \\
     f^{[k]}_{out,-}
\end{array}  \right).}
The generalization of definition \ref{def:cP_oplus} to the $n$-punctured case is as follows:
\begin{gather}\label{eq:Pplusn}
    \cP_{\oplus}:= \bigoplus_{k=1}^{n} \cP^{[k]}_{\oplus}
\end{gather}
where $\cP^{[k]}$ is the Plemelj operator 
given by the solution to the three-point problem \eqref{eq:3ptk} in the pants decomposition,
\begin{equation}\label{pk}
\begin{split}
    \left(\cP_{\oplus}^{[k]} f^{[k]}\right)(z) & = \int_{\cin^{[k]} \cup\cout^{[k]}} dw \frac{\Yt^{[k]}(z) \Yt^{[k]} (w){^{-1}}}{1-e^{-2\pi i(z-w)}} f^{[k]}(w) \\
    & := \int_{\cC^{[k]}} dw \frac{\Yt^{[k]}(z) \Yt^{[k]} (w){^{-1}}}{1-e^{-2\pi i(z-w)}} f^{[k]}(w).
    \end{split}
\end{equation}
Equivalently,
\eqs{
  \mc P^{[k]}_{\oplus} :  \left(\begin{array}{c}
     f^{[k]}_{in,-}  \\
     f^{[k]}_{out,+}
\end{array}  \right) \oplus \left(\begin{array}{c}
     f^{[k]}_{in,+}  \\
     f^{[k]}_{out,-}
\end{array}  \right) \mapsto
  \begin{pmatrix}f^{[k]}_{in,-}\\f^{[k]}_{out,+}  \end{pmatrix}
  \oplus\begin{pmatrix}{\sf a^{[k]}}&{\sf b^{[k]}}\\{\sf c^{[k]}}&{\sf d^{[k]}}\end{pmatrix}\begin{pmatrix}f^{[k]}_{in,-}\\f^{[k]}_{out,+}\end{pmatrix},}
where
\begin{equation}\begin{gathered}\begin{split}
\label{eq:1nkernels}
  ({\sfa^{[k]}} g)(z)=&\oint_{\cin^{[k]}}dw\frac{\Yt_{in}^{[k]}(z) \Yt_{in}^{[k]}(w)^{-1}-\mathbb{I}}{1-e^{-2\pi i(z-w)}}g_{in}(w),   && z\in \cin^{[k]},\\
  ({\sfb^{[k]}} g)(z)=&\oint_{\cout^{[k]}}dw\frac{ \Yt_{in}^{[k]}(z) \Yt_{out}^{[k]}(w)^{-1}}{1-e^{-2\pi i(z-w)}}g_{out}(w), && z\in \cin^{[k]},\\
  ({\sfc^{[k]}} g)(z)=&\oint_{\cin^{[k]}}dw\frac{ \Yt_{out}^{[k]}(z) \Yt_{in}^{[k]}(w)^{-1}}{1-e^{-2\pi i(z-w)}}g_{in}(w), && z\in \cout^{[k]},\\
  ({\sfd^{[k]}} g)(z)=&\oint_{\cout^{[k]}} dw\frac{ \Yt_{out}^{[k]}(z) \Yt_{out}^{[k]}(w)^{-1}-\mathbb{I}}{1-e^{-2\pi i(z-w)}}g_{out}(w), && z\in \cout^{[k]}.
\end{split}\end{gathered}\end{equation}
The functions $\Yt_{in}^{[k]},\Yt_{out}^{[k]}$ are the local solutions of the $k$-th three-point problem around $\mp i\infty$, respectively, defined in Definition \ref{def:YcalYnpoint}. In the case of a semi-degenerate system (i.e. a linear system with a single independent local monodromy exponent at $z=0$ instead of $N$) these solutions are described by generalized hypergeometric functions $_NF_{N-1}$ (see eq. 19 in \cite{Gavrylenko:2020gjb}).
Similar to \eqref{eq:keroplus}, $ \cP_{\oplus}^2 = \cP_{\oplus}$, and
\begin{equation}
\textrm{ker}~\cP_{\oplus} = \cH_{-}.
\end{equation}

Generalizing definition \ref{def:cP_Sigma}, we now introduce the Plemelj operator described by the solution to the $n$-point linear system \eqref{lax_general},
\eqs{\label{eq:PSigman}
    \left(\cP_{\Sigma_{1,n}}f\right)(z) = \oint_{\cC_{\Sigma}} \frac{dw}{2\pi i} Y(z) \Xi_{N}(z,w) Y (w){^{-1}} f(w), }
where 
\eqs{\cC_{\Sigma} := \bigcup_{k=1}^{n} \cout^{[k]} \cup \cin^{[k+1]}, && \cin^{[n+1]} := \cin^{[1]},}
and
\eq{\label{eq:Xi_N_exp}
\Xi_{N}(z,w)=\diag\left( \frac{\theta_1(z-w+Q_1-\widetilde{\rho})\theta_1'(0)}{\theta_1(z-w)\theta_1(Q_1-\widetilde{\rho})},\dots,\frac{\theta_1(z-w+Q_N-\widetilde{\rho})\theta_1'(0)}{\theta_1(z-w)\theta_1(Q_N-\widetilde{\rho})}\right),
}
where 
\begin{equation}\label{eq:tilderho} \widetilde{\rho}:= \rho -\sum_{j=1}^{n} \Lambda_{j} \left(z_{j}- \frac{\tau}{2}-\frac12 \right),
\end{equation}
and as before $\rho$ is an arbitrary parameter, and $\Xi_N$ transforms as
\begin{align}\label{eq:TauShiftXiN}
    \Xi_N(z+\tau,w)=e^{-2\pi i\bs Q+2\pi i\widetilde{\rho}}\, \Xi_N(z,w), && \Xi_N(x,w+\tau)=\Xi_N(z,w) e^{2\pi i \bs Q-2\pi i\widetilde{\rho}}.
\end{align}
The shift of the parameter $\rho$ in $\eqref{eq:tilderho}$ makes the monodromies of the Cauchy kernel time-independent (see equation \eqref{eq:zdep_MB}), and the following is true:
\eq{\cH_{\mathscr{A}} \subseteq \textrm{ker}\cP_{\Sigma_{1,n}},}
where $\mathscr{A}:= \bigcup_{k=1}^{n} \mathscr{A}^{[k]}$ in Figure \ref{fig:nptTorusPants}, and the space of functions are defined by the equation \eqref{eq:A14}. It is straightforward to check that $    \cP_{\Sigma_{1,n}}^2 = \cP_{\Sigma_{1,n}}$, and one can further prove:
    \eq{
        \cP_{\oplus}\cP_{\Sigma_{1,n}} = \cP_{\Sigma_{1,n}} ,\quad \cP_{\Sigma_{1,n}}\cP_{\oplus}= \cP_{\oplus}.}
The space of functions on $\mathscr{T}:=\bigcup_{k=1}^{n} \mathscr{T}^{[k]}$ in Figure \ref{fig:nptTorusPants} is
    \begin{gather}
   \cH_{\mathscr{T}}:= \textrm{im}~ \cP_{\oplus} = \textrm{im}~ \cP_{\Sigma_{1,n}} .
\end{gather}
\begin{definition}
The determinant tau function $\T^{(1,n)}$ is defined in terms of the Plemelj operators in equations \eqref{eq:Pplusn}, \eqref{eq:PSigman}, as 
\eq{\label{def:T1n} \T^{(1,n)} := \det_{\cH_{+}}\left[\cP_{\Sigma_{1,n}}^{-1} \cP_{\oplus} \right].}
\end{definition}
We now proceed to formulate the generalization of proposition \ref{prop:CM_tau} to the present case. 

\subsection{Block-determinant representation of the tau function}\label{subsec:Fred_GL}
\begin{proposition}\label{thm:GL_det}
The tau function $\T^{(1,n)}$ in \eqref{def:T1n} is the Fredholm determinant of a block operator acting on $L^{2}(S^{1}) \otimes \left(\bigoplus_{k=1}^{2n} \mathbb{C}^N \right)$:
 \eq{\label{eq:TAU1n} \T^{(1,n)}(\tau,z_{1},...,z_{n})= \det\left[ \mathbb{I} - K_{1,n}\right],}
 where
 \eq{{\Large K_{1,n}} = \adjincludegraphics[width=0.8\linewidth, valign=c]{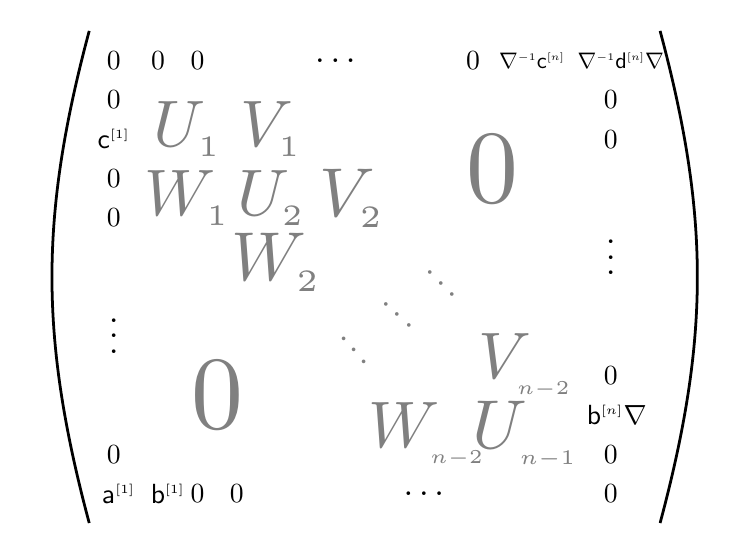},\label{eq:Kthm1}}
and 
\eqs{
    U_{k} = \left( \begin{array}{cc}
       0  & \sfa^{[k+1]}  \\
       \sfd^{[k]} & 0 
    \end{array} \right), && V_{k} = \left( \begin{array}{cc}
       \sfb^{[k+1]}  & 0 \\
       0  & 0
    \end{array}  \right), &&  W_{k} = \left( \begin{array}{cc}
        0 & 0  \\
        0 & \sfc^{[k+1]}
    \end{array} \right). \label{eq:UVW}}
    The operators $\sfa^{[k]}, \sfb^{[k]}, \sfc^{[k]}, \sfd^{[k]}$ defined in \eqref{eq:1nkernels}, 
    and  $\nabla$ is the shift operator defined in \eqref{eq:GLnabla}.
\end{proposition}
\begin{proof}
The proof goes along the same lines as that of Proposition \ref{prop:CM_tau}. Recalling the definition of the tau function in \eqref{def:T1n} and of the Plemelj operators in \eqref{pk}, \eqref{eq:PSigman}, we compute the action of $\cP_{\Sigma_{1,n},+}^{-1} \cP_{\oplus,+}$ on a function $f\in \cH_{+}$:
\begin{align}
F:=\mc P_{\Sigma_{1,n},+}^{-1}\mc P_{\oplus,+}f 
 \quad \Rightarrow \quad \mc P_{\Sigma_{1,n}}F=\mc P_{\oplus}f \, , \qquad\quad F \in \cH_{+}.
\end{align}
Now we use that for any projector $\mc P$ acting on a vector $x$, one has $x-\mc P x\in \textrm{ker } \mc P$, and that\footnote{As in the previous section, when $\left(\mathbb{I}- K_{1,n}\right)$ is invertible, $\cH = \cH_{\mathscr{T}} \oplus \cH_{A}$, and therefore $\textrm{ker } \cP_{\Sigma_{1,n}}=\cH_\mathscr{A}$.} $ \textrm{ker } \cP_{\Sigma_{1,n}}=\cH_\mathscr{A}$: 
\begin{align} \label{eq:F_1n}
    F=( F-\mc P_{\Sigma_{1,n}}  F) + \mc P_{\Sigma_{1,n}}  F= A+\mc P_\oplus  f, &&  A:=  F-\mc P_{\Sigma_{1,n}}  F\in \mc H_{\mathscr{A}}.
\end{align}
The orthogonal decomposition of $ A$ is
\begin{gather}\label{eq:AinHA}
    A= \bigoplus_{k=1}^{n} A^{[k]} =\left( \begin{array}{c}
         A^{[k]}_{in,-} \\
         A^{[k]}_{out,+} 
    \end{array}\right) \oplus \left( \begin{array}{c}
         A^{[k]}_{in,+} \\
         A^{[k]}_{out,-} 
    \end{array}\right).
\end{gather}
The $z$-dependent B-cycle monodromy in \eqref{eq:zdep_MB} implies that the monodromies around $\cC_{in}^{[1]}$ and $\cC_{out}^{[n]}$ (see Figure \ref{fig:nptTorusPants}) are given by \eqref{eq:monodromies1_n}, prompting the following expression for the shift operator $\nabla:\cH^{[1]}_{in} \rightarrow \cH^{[n]}_{out}$ 
\begin{align}
    \nabla g(z)=e^{2\pi i\left(\rho-(z-\tau) \sum_{j=1}^{n} \Lambda_{j}\right)}g(z-\tau), \label{eq:GLnabla}
\end{align}
in order to 'glue' the boundary spaces on $\cC_{in}$, $\cC_{out}$. The factor $e^{2\pi i z \sum_{j=1}^{n} \Lambda_{j}}$ in the above definition of $\nabla$ leads to the following action of $\nabla^{-1}$:
\begin{equation}
    \nabla^{-1}h(z)=e^{-2\pi i\left(\rho-z\sum_{j=1}^{n} \Lambda_{j}\right)}h(z+\tau). \label{eq:GLnablainv}
\end{equation}

Identifying the boundaries
$\cC^{[n]}_{out}$ with $\cC^{[1]}_{in}$, and $\cC^{[k]}_{out}$ with $\cC^{[k+1]}_{in}$ for $k = 1...n-1$, produces the torus from the pants decomposition as in Figure \ref{fig:nptTorusPants}, and translates to the following constraints on $A$ in \eqref{eq:AinHA}:
\begin{align}
    A^{[1]}_{in, \pm} = \nabla^{-1} A^{[n]}_{out, \pm}; &&
    A^{[k]}_{out,\pm} = A^{[k+1]}_{in,\pm},\quad\quad k = 1, \ldots, n-1, \label{eq:con2}
\end{align}
where the translation operator $\nabla:\cH^{[1]}_{in}\rightarrow \cH^{[n]}_{out}$ is defined as in \eqref{eq:GLnabla}.
Component-wise, equation \eqref{eq:F_1n} reads
\begin{equation} \label{eq:Fplus}
\begin{split}
    F & = \left( \begin{array}{c}
         F_{in,-} \\
         F_{out,+} 
    \end{array}\right) \oplus \left( \begin{array}{c}
         0 \\
         0
    \end{array}\right) \\
    & = \bigoplus_{k = 1}^{n} \left( \begin{array}{c}
         A^{[k]}_{in,-} \\
         A^{[k]}_{out,+} 
    \end{array}\right) \oplus \left( \begin{array}{c}
         A^{[k]}_{in,+} \\
         A^{[k]}_{out,-} \end{array} \right) + \left( \begin{array}{c}
         f^{[k]}_{in,-} \\
         f^{[k]}_{out,+} 
    \end{array}\right) \oplus \left(\begin{array}{cc}
       \sfa^{[k]}  & \sfb^{[k]}  \\
       \sfc^{[k]}  & \sfd^{[k]}
    \end{array} \right) \left( \begin{array}{c}
         f^{[k]}_{in,-} \\
         f^{[k]}_{out,+} 
    \end{array}\right).
\end{split}
\end{equation}
Imposing the condition that the $\mathcal{H}_{-}$ component of $F$ is zero, and using the constraints in \eqref{eq:con2} we get
\eqs{
    A^{[1]}_{in,+} = -\sfa^{[1]}f^{[1]}_{in,-}  - \sfb^{[1]}f^{[1]}_{out,+} = \nabla^{-1} A^{[n]}_{out, +},
    \\
    A^{[k]}_{in, +} = -\sfa^{[k]}f^{[k]}_{in,-}  - \sfb^{[k]}f^{[k]}_{out,+} = A^{[k-1]}_{out, +}, \quad \textrm{for} \,\, k=2...n,
    \\
    A^{[k]}_{out, - } = -\sfc^{[k]} f^{[k]}_{in,-} - \sfd^{[k]} f^{[k]}_{out,+} =A^{[k+1]}_{in, -}, \quad \textrm{for} \,\, k=1...n-1,
    \\
    A^{[n]}_{out, - } = -\sfc^{[n]} f^{[n]}_{in,-} - \sfd^{[n]} f^{[n]}_{out,+} = \nabla A^{[1]}_{in,-}.
    \label{eq:A14}
}
Substituting \eqref{eq:A14} in \eqref{eq:Fplus},
\begin{gather}
  F= \bigoplus_{k=1}^{n}  \left( \begin{array}{c}
        A^{[k]}_{in,-} \\
         A^{[k]}_{out,+} 
    \end{array}\right) + \left( \begin{array}{c}
         f^{[k]}_{in,-} \\
         f^{[k]}_{out,+} 
    \end{array}\right)  
    =\bigoplus_{k=1}^{n}\left( \begin{array}{c}
         f^{[k]}_{in,-} \\
         f^{[k]}_{out,+} 
    \end{array}\right) \nonumber\\ -\left(\begin{array}{cccccccccccccc}
        0 &0 & 0 & 0 & 0 & 0 & 0 & 0 &\ldots & 0 & \nabla^{-1} \sfc^{[n]} &\nabla^{-1} \sfd^{[n]}    \\
        0 & 0 & \sfa^{[2]} & \sfb^{[2]} & 0 & 0& 0 & 0 & \ldots & 0 & 0 & 0 \\
        \sfc^{[1]}& \sfd^{[1]} & 0& 0& 0 & 0& 0 & 0 & \ldots &0 & 0 & 0 \\ 
        0 & 0 & 0 & 0 & \sfa^{[3]} & \sfb^{[3]} & 0 & 0 &  \ldots  &0 & 0 & 0 \\
        0 & 0 & \sfc^{[2]} & \sfd^{[2]} & 0 & 0 & 0 & 0 &\ldots  &0 & 0 & 0 \\
        0 & 0 & 0& 0 & 0 & 0 & \sfa^{[4]} & \sfb^{[4]} & \ldots  &0 & 0 & 0 \\
        \vdots & \vdots& \vdots &\vdots & & & & \ddots & & \vdots &\vdots &  \vdots\\
        0 & 0 & 0 & 0 & \ldots & 0 & 0& 0 & 0 & 0&  \sfa^{[n]} & \sfb^{[n]} \\
        0 & 0 & 0 & 0 & \ldots & 0 & 0 & 0 & \sfc^{[n-1]} & \sfd^{[n-1]} & 0 & 0\\
       \nabla \sfa^{[1]} &  \nabla \sfb^{[1]} & 0 & 0  &\ldots & 0 & 0 & 0 & 0 & 0  & 0 & 0\\
    \end{array}  \right) \left( \begin{array}{c}
         f^{[1]}_{in, -} \\
         f^{[1]}_{out,+} \\
         f^{[2]}_{in, -} \\
         f^{[2]}_{out,+} \\
         f^{[3]}_{in, -} \\
         f^{[3]}_{out,+} \\
         \vdots \\
         f^{[n]}_{in, -} \\
         f^{[n]}_{out,+} \\
    \end{array} \right) \nonumber \\
    := \left(\mathbb{I}- \widehat{K}_{1,n}\right) f. \label{eq:AHplus}
\end{gather}
Similar to \eqref{eq:nabla_conj11}, we conjugate $\widehat{K}_{1,n}$ with the diagonal operator ${\rm \diag}\left(1,1,...,\nabla^{-1}  \right)$
\eqs{K_{1,n}:={\rm diag}(1,...,\nabla^{-1})\, \widehat{K}_{1,n}\, {\rm diag}(1,...,\nabla) }
obtaining equation \eqref{eq:Kthm1}.
\end{proof}
\begin{remark} It is straightforward to recover \eqref{eq:TAU11} from
\begin{figure}[H]
\centering
\includegraphics[width=12cm]{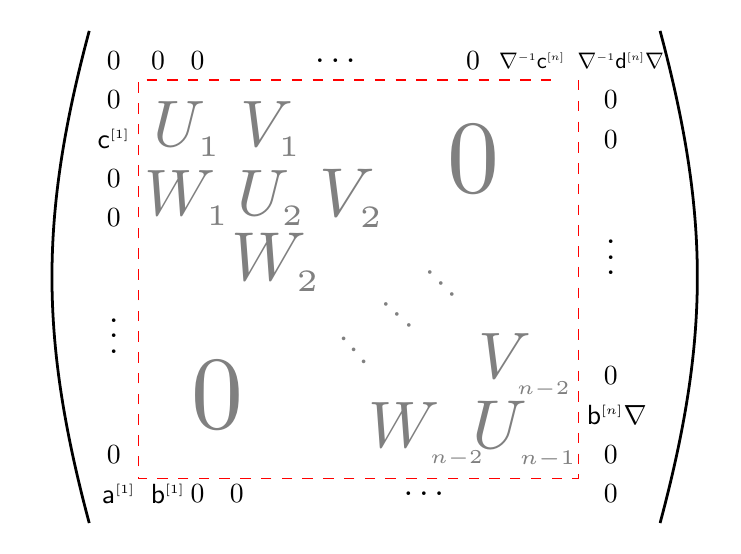}
\caption{Kernel \label{kernel}}
\end{figure}
Moreover, the block form of the tau function $\T^{(1,n)}$ includes naturally an $(n-1)\times (n-1)$ sub-block identical to the tau function appearing in pg. 18 of \cite{Gavrylenko:2016zlf} for the $n+2$-punctured sphere, as emphasised in Figure \eqref{kernel}. This is a consequence of the fact that if we cut the tube that joins the first and last trinion in Figure \ref{fig:nptTorusPants}, (i.e. if we take the limit $\tau\rightarrow+i\infty$), we obtain a Fuchsian problem for an $n+2$-punctured sphere:
\begin{equation}
    \lim_{\tau\rightarrow+i\infty}\T^{(1,n)}\propto\T^{(0,n+2)}.
\end{equation}
\end{remark}
\subsection{Relation to the Hamiltonians}\label{subsec:Ham_GL}
\begin{theorem}\label{thm:GL_Ham} 
The  isomonodromic tau function $\T_{H}$ in \eqref{eq:IsomHam} is related to the Fredholm determinant of the operator $K_{1,n}$ in \eqref{eq:Kthm1} as
    \begin{equation}
    \begin{split}\T_{H}(\tau) &= \det\left[\mathbb{I} - K_{1,n} \right] e^{i\pi\tau\tr\left(\bs \sigma_1^2+\frac{\mathbb{I}}{6} \right)} e^{-i \pi N \widetilde{\rho}} \prod_{i=0}^N\frac{\eta(\tau)}{\theta_1\left(Q_i-\widetilde{\rho} \right)}\prod_{k=1}^n e^{-i\pi z_k\left(\tr\bs \sigma_{k+1}^2-\tr\bs \sigma_{k}^2\right)}\Upsilon_{1,n}, \end{split}\label{eq:Thm2}
    \end{equation}    
where $\Upsilon_{1,n}$ is an arbitrary function of the monodromy data of the linear system \eqref{lax_general}, $Q_{i} \equiv Q_{i}(\tau, z_{1},...,z_{n})$ are the Calogero-like dynamical variables in the linear system \eqref{eq:nptL}, $ \bs \sigma_{k}=\bs a_k+\sum_{j<k}\Lambda_{j}$ are the monodromy exponents defined in \eqref{eq:1nmonodromy}, and $\bs a_{n+1}=\bs a_1$,
\begin{equation}
   \widetilde{\rho} = \rho - \sum_{k=1}^{n}\Lambda_{k} \left(z_{k} - \frac{\tau}{2} -\frac12 \right),
\end{equation}
and $\rho$ is an arbitrary parameter.
\end{theorem}
\begin{proof}
Let us recall equation \eqref{def:T1n}:
\eqs{ \T^{(1,n)} = \det_{\cH_{+}}\left[ \cP_{\Sigma_{1,n}}^{-1} \cP_{\oplus} \right],} where the operators $\cP_{\oplus}$, $\cP_{\Sigma_{1,n}}$ are defined in \eqref{pk} and \eqref{eq:PSigman} respectively. The logarithmic derivative of the tau function $\T^{(1,n)}$, has two main components: the derivatives with respect to the modular parameter $\tau$, and the position of the singularities $\lbrace z_{k}\rbrace_{k=1}^{n}$:
\begin{gather}
 \delta \log\T^{(1,n)} =2\pi i d\tau \, \partial_{\tau} \log\T^{(1,n)} + \sum_{k=1}^n dz_{k} \, \partial_{z_{k}} \log\T^{(1,n)}. \label{eq:deltaTau}
\end{gather}
Computation of this derivative can be done exactly in the same way as in \cite[page 20]{Gavrylenko:2016zlf}:
\eqs{ \label{eq:1nTr}
\delta \log\T^{(1,n)} 
= - 2\pi i  \tr_{\cH} \left[\cP_{\oplus} \partial_{\tau} \cP_{\Sigma_{1,n}} \right]\, d\tau- \sum_{l=1}^{n}\tr_{{\mc H}^{[l]}} \left[ \cP_{\oplus}^{[l]} \partial_{z_{k}} \cP_{\Sigma_{1,n}}. \right] \, dz_{k}
}
The computation for the first term in \eqref{eq:1nTr} is the same as in the proof for Theorem \ref{prop:CM_Ham} in section \ref{subsec:Ham_CM}: the $\tau$-derivative is given by
\begin{equation}\label{eq:1nI_tau}
   -  \tr_{\cH} \left[\cP_{\oplus} \partial_{\tau} \cP_{\Sigma_{1,n}} \right]=\sum_{l=1}^n -I_\tau^{(l)}-I_w-I_z,
\end{equation}
where
\begin{gather}
    I_\tau^{(l)}=\oint_{\CP^{[l]}}dw\oint_{\CS^{[l]}}\frac{dz}{2\pi i} \frac{1}{1-e^{-2\pi i(z-w)}}\tr\left[\Yt^{[l]}(z)\Yt^{[l]}(w)^{-1}\partial_\tau\left(Y(w)\Xi_{N}(w,z)Y(z)^{-1} \right) \right], \nonumber \\
    I_z=\oint_{\CP^{[n]}}dw \oint_{\CS_{out}^{[n]} }\frac{dz}{2\pi i} \frac{1}{1-e^{-2\pi i(z-w)}}\tr\left[\Yt^{[n]}(z)\Yt^{[n]}(w)^{-1}\partial_z\left(Y(w)\Xi_{N}(w,z)Y(z)^{-1} \right) \right], \nonumber  \\
    I_w=\oint_{\CP_{out}^{[n]}}dw \oint_{\CS^{[n]} }\frac{dz}{2\pi i} \frac{1}{1-e^{-2\pi i(z-w)}}\tr\left[\Yt^{[n]}(z)\Yt^{[n]}(w)^{-1}\partial_w\left(Y(w)\Xi_{N}(w,z)Y(z)^{-1} \right) \right].
\end{gather}
Note that the contours of $I_{z}$, $I_{w}$ involve only the final trinion, because the identification $z\sim z+\tau$ glues $\mc C_{out}^{[n]}$ to $\mc C_{in}^{[1]}$, as in Figure \ref{fig:nptTorusPants}. Like in the case of one puncture, $I_\tau=\sum_lI_\tau^{(l)}$ vanishes because $M_\tau$ in \eqref{eq:nptMt} has zero residue at the punctures, while $I_w$ vanishes because the $z$-loop is contractible. Using the notation \ref{not:bold}, 
\begin{equation}
    \sum_i \frac{\theta_1'(Q_i-\widetilde{\rho})}{\theta_1(Q_i-\widetilde{\rho})}(L)_{ii} \equiv\tr\left[\frac{\theta_1'(\bs Q-\widetilde{\rho})}{\theta_1(\bs Q-\widetilde{\rho})}L \right],
\end{equation}
we are then left with the following expression  (see \eqref{eq:I_zF} for comparison) for the first term in \eqref{eq:1nTr}:
\begin{gather}
     -  \tr_{\cH} \left[\cP_{\oplus} \partial_{\tau} \cP_{\Sigma_{1,n}} \right]= - I_z=\frac{1}{2} \oint_{\CS^{[n]}_{out}} \frac{dz}{2\pi i}  \tr\left[-L_{3pt}^{[n]}(z)^2+  L^2(z)+2 i\pi\tr L_{3pt}^{[n]}\right] \nonumber\\
     + \frac{1}{2}\oint_{\CS^{[n]}_{out}} \frac{dz}{2\pi i}\left[2 \frac{\theta_{1}'(\bs Q-\widetilde{\rho})}{\theta_{1}(\bs Q-\widetilde{\rho})} L(z) + \frac{\theta_{1}''(\bs Q-\widetilde{\rho})}{\theta_{1}(\bs Q-\widetilde{\rho})} -\mathbb{I}\left( \frac{1}{3} \frac{\theta_{1}'''}{\theta_{1}}+\frac{(2\pi i)^2}{6} \right)\right] \nonumber \\
    =\frac{H_\tau}{2\pi i}+\frac{d}{d\tau}\log\left(e^{-i\pi\tau \tr\left(\bs \sigma_1^2+\mathbb{I}/6\right)}e^{i\pi \tau N\sum_{j=0}^{n}\Lambda_{j}}\prod_{i=1}^N\frac{\theta_1(Q_i-\widetilde{\rho})}{\eta(\tau)} \right)
    \\
    =\frac{H_\tau}{2\pi i}+\frac{d}{d\tau}\log\left(e^{-i\pi\tau \tr\left(\bs \sigma_1^2+\mathbb{I}/6\right)}e^{i\pi N\tilde{\rho}}\prod_{i=1}^N\frac{\theta_1(Q_i-\widetilde{\rho})}{\eta(\tau)} \right).\label{eq:1nTerm1}
\end{gather}
In the last line we used 
\begin{equation}
    \sum_{j=0}^N\Lambda_j=\frac{1}{2}\sum_{j=1}^N\Lambda_j,
\end{equation}
so that
\begin{equation}
    i\pi\tr L_{3pt}^{[n]}=\frac{N}{2}(2\pi i)^2\sum_{j=0}^n\Lambda_j\mathop{=}^{\eqref{eq:tilderho}}2\pi i\partial_\tau\log\left(e^{i\pi N\widetilde{\rho}} \right).
    \end{equation}
Let us now compute the second term in \eqref{eq:1nTr}:
\begin{gather}
   - \sum_{l=1}^{n} \tr_{\cH^{[l]}} \left[ \cP_{\oplus}^{[l]} \partial_{z_{k}} \cP_{\Sigma_{1,n}} \right] \nonumber\\
   =  -\sum_{l=1}^{n} \oint_{\CP^{[l]}}dw  \oint_{\CS^{[l]}} \frac{dz}{2\pi i} \frac{1}{1-e^{-2\pi i(z-w)}} \tr\left[\Yt^{[l]}(z) \Yt^{[l]}(w)^{-1} \partial_{z_{k}} \left(  Y(w) \Xi_{N}(w,z) Y(z)^{-1} \right)   \right] \nonumber \\
    = \sum_{l=1}^{n} \left(I_{z_k}^{(l,1)} + I_{z_k}^{(l,2)} + I_{z_{k}}^{(l,3)}\right), \label{eq:z_k}
\end{gather} 
where
\begin{gather}
    I_{z_{k}}^{(l,1)} :=- \oint_{\CP^{[l]}}dw  \oint_{\CS^{[l]}} \frac{dz}{2\pi i} \frac{1}{1-e^{-2\pi i(z-w)}} \tr\left[\Yt^{[l]}(z) \Yt^{[l]}(w)^{-1} \partial_{z_{k}}Y(w) \Xi_{N}(w,z) Y(z)^{-1}  \right]  =0 \label{eq:I_zk1}
\end{gather}
since the $z$-loop is contractible, and 
\begin{gather}
     I_{z_{k}}^{(l,2)} :=- \oint_{\CP^{[l]}}dw  \oint_{\CS^{[l]}} \frac{dz}{2\pi i} \frac{1}{1-e^{-2\pi i(z-w)}} \tr\left[\Yt^{[l]}(z) \Yt^{[l]}(w)^{-1} Y(w) \partial_{z_{k}}\Xi_{N}(w,z) Y(z)^{-1}  \right]  =0
\end{gather}
since $\partial_{z_k}\Xi_N(w,z)$ is regular in $w \sim z$.
The term $I_{z_{k}}^{(l,3)}$ is computed by expanding $\Xi_{N}$ for $w\sim z$ as in \eqref{eq:Xi_N_exp}, and using \eqref{eq:ExpExpansion} : 
\begin{gather}
    I_{z_{k}}^{(l,3)} := -\oint_{\CP^{[l]}}\frac{dw}{2\pi i}  \oint_{\CS^{[l]}} \frac{dz}{2\pi i} \frac{1}{1-e^{-2\pi i(z-w)}} \tr\left[\Yt^{[l]}(z) \Yt^{[l]}(w)^{-1}  Y(w) \Xi_{N}(w,z) \partial_{z_{k}}Y(z)^{-1}    \right] \nonumber \\
    = \oint_{\CP^{[l]}}\frac{dw}{2\pi i}  \oint_{\CS^{[l]}} \frac{dz}{2\pi i}   \tr\left[\Yt^{[l]}(w)^{-1}  Y(w) \frac{1}{(w-z)^2} \partial_{z_{k}}Y(z)^{-1} \Yt^{[l]}(z)    \right] \nonumber \\
    \oint_{\CP^{[l]}}\frac{dw}{2\pi i}  \oint_{\CS^{[l]}} \frac{dz}{2\pi i}  \tr\left[\frac{\Yt^{[l]}(w)^{-1}  Y(w) }{w-z} \left(\frac{\theta_{1}'(\bs Q-\widetilde{\rho})}{\theta_{1}(\bs Q-\widetilde{\rho})} -i\pi\right) \partial_{z_{k}}Y(z)^{-1} \Yt^{[l]}(z)    \right] \nonumber \\
=\oint_{\CS^{[l]}}\frac{dz}{2\pi i} \tr \left\{\left(Y(z)^{-1}\Yt^{[l]}(z)L_{3pt}^{[l]}(z)(Y(z)^{-1}\Yt^{[l]}(z))^{-1}-L(z) \right)M_{z_k}(z)\right\}\nonumber\\
 -\oint_{\CS^{[k]}}\frac{dz}{2\pi i} \tr\left\{\left(\frac{\theta_1'(\bs Q-\widetilde{\rho})}{\theta_1(\bs Q-\widetilde{\rho})}-i\pi\right)M_{z_k}(z) \right\}.\label{eq:I_zk2}
\end{gather}
Note that \eqref{eq:I_zk2} can be different from zero only for $l=k$, because the integrand is regular for $l\ne k$. To compute the first and second term, we use the regular parts $L(z)_{reg}$ and $L_{3pt}^{[k]}(z)_{reg}$ of $L(z)$ (eq. \eqref{eq:nptL}) and $L_{3pt}^{[k]}$ (eq. \eqref{eq:3ptk}), as well as the explicit expression \eqref{eq:nptMk}: for $M_{z_k}$:
\begin{gather}
    \oint_{\CS^{[k]}}\frac{dz}{2\pi i} \tr \left\{\left(Y(z)^{-1}\Yt^{[k]}(z)L_{3pt}^{[k]}(z)_{reg}(Y(z)^{-1}\Yt^{[k]}(z))^{-1}\right)M_{z_k}(z)\right\} \nonumber\\
    =2\pi i\tr\left(A_1^{[k]}A_0^{[k]} \right)+ i\pi \tr {A_1^{[k]}}^{2}=i\pi\left(\tr\bs \sigma_{k+1}^2-\tr\bs \sigma_k^2\right)=\frac{d}{dz_k}\log\left(e^{i\pi z_k(\tr\bs \sigma_{k+1}^2-\tr\bs \sigma_{k}^2)} \right), \label{eq:zkfactor}
\end{gather}
where we used the identity
\begin{equation}
    \tr\left(A_0^{[k]}A_1^{[k]} \right)=\frac{1}{2}\tr\left({A_\infty^{[k]}}^2-{A_0^{[k]}}^2-{A_1^{[k]}}^2 \right)=\frac{1}{2}\left(\tr\bs \sigma_{k+1}^2-\tr\bs \sigma_k^2-\tr\bs \mu_k^2\right).
\end{equation}
To compute the second term in \eqref{eq:I_zk2}, we note that $M_{z_k}$ in \eqref{eq:nptMk} is simply the singular part at $z_k$ of $L$ in \eqref{eq:nptL} with a negative sign, so that
\begin{equation}\label{eq:HkLM}
   -\oint_{\CS^{[k]}}\tr L(z)M_{z_k}(z)\frac{dz}{2\pi i}=\oint_{\CS^{[k]}}\frac{1}{2}\tr L^2(z)\frac{dz}{2\pi i}=\tr\left(S_kL(z_k)_{reg} \right)= H_{z_k}
\end{equation}
The last term in \eqref{eq:I_zk2}:
\begin{equation}
\begin{split}
    -\oint_{\CS^{[k]}}\frac{dz}{2\pi i} &\tr\left\{\left(\frac{\theta_1'(\bs Q-\widetilde{\rho})}{\theta_1(\bs Q-\widetilde{\rho})}-i\pi\right)M_{z_k}(z) \right\} =\sum_{i=1}^N\left(S_{ii}^{(k)}+\Lambda_{k}\right)\left(\frac{\theta_1'(Q_i-\widetilde{\rho})}{\theta_1(Q_i-\widetilde{\rho})}-i\pi \right)\\
    & =\sum_{i=1}^N\left(S_{ii}^{(k)}+\Lambda_{k}\right)\frac{\theta_1'(Q_i-\widetilde{\rho})}{\theta_1(Q_i-\widetilde{\rho})}-i\pi N\Lambda_{k}
\end{split}\label{eq:Mtheta1}
\end{equation}
since $\tr S^{(k)} =0$. To simplify \eqref{eq:Mtheta1} further, let us first substitute \eqref{eq:nptL}, \eqref{eq:nptMk} in \eqref{eq:HkLM}:
\begin{equation}
    H_{z_k}=\sum_{i=1}^NS_{ii}^{(k)}P_i+\left(P\text{-independent part}\right),
\end{equation}
which implies, together with the canonical Poisson bracket $\{Q_i,P_j\}=\delta_{ij}$, that
\begin{equation}\label{eq:Sii}
    \frac{dQ_i}{dz_k}=\frac{\partial H_{z_k}}{\partial P_i}=S_{ii}^{(k)}\end{equation}
Then \eqref{eq:Mtheta1} becomes
\begin{equation}
    \begin{split}
   \oint_{\CS^{[k]}}\frac{dz}{2\pi i} \tr\left\{\frac{\theta_1'(\bs Q-\widetilde{\rho})}{\theta_1(\bs Q-\widetilde{\rho})}M_{z_k}(z) \right\} 
   & = \sum_{i=1}^N\left(S_{ii}^{(k)}+\Lambda_{k}\right)\frac{\theta_1'(Q_i-\widetilde{\rho})}{\theta_1(Q_i-\widetilde{\rho})}   -i\pi N\Lambda_{k}\\
   &\mathop{=}^{\eqref{eq:tilderho}}\sum_{i=1}^N\left(\frac{dQ_i}{dz_k}-\frac{d\widetilde{\rho}}{dz_k}\right)\frac{\theta_1'(Q_i-\widetilde{\rho})}{\theta_1(Q_i-\widetilde{\rho})}-i\pi N\Lambda_{k}  \\ &=\frac{d}{dz_k}\log\left(e^{-i\pi Nz_k \Lambda_{k}}\prod_{i=1}^N\theta_1(Q_i-\widetilde{\rho} )\right)\\
   & \mathop{=}^{\eqref{eq:tilderho}}\frac{d}{dz_k}\log\left(e^{i\pi N \widetilde{\rho}}\prod_{i=1}^N\theta_1(Q_i-\widetilde{\rho} )\right). \end{split}\label{eq:theta_eta}
\end{equation}
Substituting \eqref{eq:zkfactor}, \eqref{eq:HkLM}, \eqref{eq:theta_eta} back in \eqref{eq:I_zk2},
\begin{gather} 
   I_{z_{k}}^{(l,3)} =\delta_k^l\left[  \frac{d}{dz_k}\log\left(e^{i\pi z_k\left(\tr\bs \sigma_{k+1}^2-\tr\bs \sigma_{k}^2\right)}e^{i\pi N\widetilde{\rho}} \right) + H_{z_k} + \frac{d}{dz_k}\log\prod_{i=1}^N\theta_1(Q_i-\widetilde{\rho})\right]. \label{I_zk2F}
\end{gather}
Putting it all together \eqref{eq:1nTr}:
\begin{equation}
    \begin{split}
    \delta\log\det\left[\mathbb{I}-K_{1,n}\right] &\mathop{=}^{\eqref{eq:TAU1n}} \delta \log \T^{(1,n)}  = - 2\pi i  \tr_{\cH} \left[\cP_{\oplus} \partial_{\tau} \cP_{\Sigma_{1,n}} \right]\, d\tau- \sum_{l=1}^{n}\tr_{{\mc H}^{[l]}} \left[ \cP_{\oplus}^{[l]} \partial_{z_{k}} \cP_{\Sigma_{1,n}} \right] \, dz_{k}  \\
    &\mathop{=}^{\eqref{eq:1nI_tau},\eqref{eq:z_k}} -2\pi i\left(\sum_{l=1}^nI_\tau^{(l)}+I_w+I_z \right) d\tau +\sum_{l=1}^{n} \left(I_{z_k}^{(l,1)} + I_{z_k}^{(l,2)} +I_{z_k}^{(l,3)} \right) dz_{k}  \\
    &\mathop{=}^{\eqref{eq:1nTerm1}, \eqref{eq:I_zk1}, \eqref{I_zk2F}} H_\tau d\tau + \sum_{k=1}^{n} H_{z_k} dz_{k} \\ 
    &+2\pi i\frac{d}{d\tau}\log\left(e^{-i\pi \tau\tr\left(\bs \sigma_1^2+\frac{\mathbb{I}}{6}\right)} e^{i \pi N \widetilde{\rho}} \prod_{i=1}^N\frac{\theta_1(Q_i-\widetilde{\rho})}{\eta(\tau)} \right) d\tau \\ &+\sum_{k=1}^{n}\left[\frac{d}{dz_k}\log\left(e^{i\pi z_k\left(\tr\bs \sigma_{k+1}^2-\tr\bs \sigma_{k}^2\right)}e^{i\pi N\widetilde{\rho}} \right)  + \frac{d}{dz_k}\log\prod_{i=1}^N\theta_1(Q_i-\widetilde{\rho})\right] dz_{k} \\
    & \mathop{=}^{\eqref{eq:IsomHam}}\delta\log\T_H\\
    &+\delta\log\left[e^{-i\pi\tau\tr\left(\bs \sigma_1^2+\frac{\mathbb{I}}{6}\right)}e^{i\pi N\widetilde{\rho}} \prod_{i=1}^N \frac{\theta_1(Q_i-\widetilde{\rho})}{\eta(\tau)} \prod_{k=1}^n e^{i\pi z_k\left(\tr\bs \sigma_{k+1}^2-\tr\bs \sigma_{k}^2\right)}\right] 
    \end{split}\label{eq:Tau_1n}
\end{equation} 
Integrating \eqref{eq:Tau_1n} and substituting \eqref{eq:tilderho}, we obtain \eqref{eq:Thm2}. 
\end{proof}

\section{Charged partitions and Nekrasov functions}\label{sec:Minors}

In this section, we expand the Fredholm determinant \eqref{eq:TAU1n} in terms of its principal minors labeled by random partitions, and show that the resulting combinatorial expression takes the form of a Fourier series of Nekrasov functions, known as dual Nekrasov-Okounkov partition function \cite{Nekrasov:2003rj} in the self-dual Omega-background, for a class of four-dimensional $\mathcal{N}=2$ supersymmetric gauge theories called circular quiver gauge theories. These are gauge theories with multiple $SU(N)$ gauge groups, each of which is coupled to two matter hypermultiplets in the bifundamental representation, and their partition functions are equal to free fermion conformal blocks on the torus.

\subsection{Minor expansion}\label{subsec:Minor_Pants}
The Hilbert space $L^2(S^1)$ admits a natural orthonormal basis of Fourier modes. We now compute the minor expansion of the Fredholm determinant \eqref{eq:TAU1n} in this particular basis. 
The kernels of the operators $\sfa^{[k]},\sfb^{[k]},\sfc^{[k]},\sfd^{[k]}$ in \eqref{eq:1nkernels} read:
\eqs{\label{eq:abcdYt}
    \sfa^{[k]} (z,w) = \frac{\mathbb{I}-\Yt^{[k]}(z) \Yt^{[k]}(w)^{-1}}{1-e^{-2\pi i(z-w)}}, && z,w\in \cin^{[k]}, \\
    \sfb^{[k]}(z,w)\, = \frac{\Yt^{[k]}(z) \Yt^{[k]}(w)^{-1}}{1-e^{-2\pi i(z-w)}}, && z\in \cin^{[k]}, \, w\in \cout^{[k]} \\
    \sfc^{[k]}(z,w) = -\frac{\Yt^{[k]}(z) \Yt^{[k]}(w)^{-1}}{1-e^{-2\pi i(z-w)}}, && z\in \cout^{[k]}, \, w\in \cin^{[k]} \\
    \sfd^{[k]}(z,w) = \frac{\Yt^{[k]}(z) \Yt^{[k]}(w)^{-1}-\mathbb{I} }{1-e^{-2\pi i(z-w)}}, && z,w\in \cout^{[k]}.}
Since the solution $\Yt^{[k]}$ to the $k$-th three-point problem defined in \eqref{eq:3ptk} is multivalued on $\cC_{in},\cC_{out}$, with monodromy determined by $\bs \sigma_{k},\bs \sigma_{k+1}$ respectively as in equation \eqref{eq:1nmonodromy}, the  matrix elements of the kernels in \eqref{eq:abcdYt} have the following (twisted) Fourier series representation:
\begin{equation}
    \begin{split}
\sfa_{\alpha, \beta}^{[k]}(z,w) = \sum_{-r,s\in\mathbb{Z}'_+} \sfa^{-r;\alpha}_{s;\beta} e^{2\pi iz\left(\frac{1}{2}-r+{ \sigma}_{k}^{(\alpha)}\right)} e^{2\pi iw\left(-\frac{1}{2}-s- { \sigma}_{k}^{(\beta)}\right)}, \\
    \sfb_{\alpha, \beta}^{[k]}(z,w) = \sum_{r,s\in\mathbb{Z}'_+} \sfb^{-r;\alpha}_{-s;\beta} e^{2\pi iz\left(\frac{1}{2}-r+{ \sigma}_{k}^{(\alpha)}\right)} e^{2\pi i w\left(-\frac{1}{2}+s- { \sigma}_{k+1}^{(\beta)}\right)}, \\
    \sfc_{\alpha, \beta}^{[k]}(z,w) = \sum_{r,s\in\mathbb{Z}'_+} \sfc^{r;\alpha}_{s;\beta} e^{2\pi iz\left(\frac{1}{2}+r+{ \sigma}_{k+1}^{(\alpha)}\right)} e^{2\pi iw\left(-\frac{1}{2}-s- { \sigma}_{k}^{(\beta)}\right)}, \\
    \sfd_{\alpha, \beta}^{[k]}(z,w) = \sum_{r,s\in\mathbb{Z}'_+} \sfd^{r;\alpha}_{-s;\beta} e^{2\pi iz\left(\frac{1}{2}+r+{ \sigma}_{k+1}^{(\alpha)}\right)} e^{2\pi iw\left(-\frac{1}{2}+s- { \sigma}_{k+1}^{(\beta)}\right)},\end{split}\label{eq:ker_abcd}\end{equation}
with $\alpha,\beta=1,\dots,N$, and $\mathbb{Z}'_+$ denoting the set of positive half-integers. 
The Fourier coefficients $\sfa^{-r;\alpha}_{s;\beta},\sfb^{-r;\alpha}_{-s;\beta},\sfc^{r;\alpha}_{s;\beta},\sfd^{r;\alpha}_{-s;\beta}$
were computed in \cite{Gavrylenko:2016zlf}, but we will not need their explicit form. A submatrix of either $\sfa^{[k]},\sfb^{[k]},\sfc^{[k]},\sfd^{[k]} $, of size $i\times j$, is denoted by two unordered sets $\{(r,\alpha)_1,\dots,(r,\alpha)_i\} \in 2^{\mathbb{Z}'_{+} \times \lbrace 1,...,N \rbrace}$ and $\{(s,\beta)_1,\dots,(s,\beta)_j\} \in  2^{\mathbb{Z}'_{+} \times \lbrace 1,...,N \rbrace}$ where $r,s$ are the Fourier indices in the expansion \eqref{eq:ker_abcd}, and $\alpha,\beta$ are the matrix ("color") indices. Such sets comprised of positive (negative) Fourier modes will be denoted by $I$ ($J$). Minors of $K$ will then be denoted by collections of such sets $\vec{I}:=\{I_1,\dots,I_n\} $, $\vec{J}:=\{J_1,\dots,J_n\}$, and a generic minor $K_{\vec{I},\vec{J}}$ has the form:
\begin{gather}
    K_{\vec{I},\vec{J}}= \left(\begin{array}{cccccccccc}
        0 &0  & 0 & 0 &  \ldots & 0& \left(\nabla^{-1} \sfc^{[n]}\right)^{J_{1}}_{J_{n}} &\left(\nabla^{-1} \sfd^{[n]}\nabla \right)^{J_{1}}_{I_{1}}  \\
        0 & 0 & \left(\sfa^{[2]} \right)^{I_{2}}_{J_{2}}& \left(\sfb^{[2]}\right)^{I_{2}}_{I_{3}} &   \ldots & 0& 0 & 0 \\
       \left( \sfc^{[1]}\right)^{J_{2}}_{J_{1}} & \left(\sfd^{[1]}\right)^{J_{2}}_{I_{2}} & 0& 0&   \ldots&0& 0 & 0 \\ 
        0 & 0 & 0 & 0  & \ldots &0 & 0 & 0  \\
        0 & 0 & \left(\sfc^{[2]}\right)^{J_{3}}_{J_{2}} & \left(\sfd^{[2]}\right)^{J_{3}}_{I_{3}}  & \ldots & 0 & 0 & 0 \\
        \vdots  &\vdots & &  \ddots & & \vdots &\vdots & \vdots\\
        0 & 0  & \ldots & 0  & 0 & 0&  \left(\sfa^{[n]}\right)^{I_{n}}_{J_{n}} & \left(\sfb^{[n]} \nabla \right)^{I_{n}}_{I_{1}} \\
        0 & 0 & \ldots & 0  & \left(\sfc^{[n-1]} \right)^{J_{n}}_{J_{n-1}}& \left(\sfd^{[n-1]}\right)^{J_{n}}_{I_{n}} & 0 & 0\\
        \left(\sfa^{[1]}\right)^{I_{1}}_{J_{1}} & \left( \sfb^{[1]}\right)^{I_{1}}_{I_{2}} & \ldots  & 0  & 0 & 0& 0 &0\\
    \end{array}  \right). \label{eq:kijiji}
\end{gather}
A combinatorial interpretation in terms of Maya diagrams and charged partitions proves vital in expressing the minors as Nekrasov functions:
the multi-indices $(I,J)$ can be viewed as the positions $\sf h(m^{(\alpha)})$ and $\sf p(m^{(\alpha)})$ of 'holes' and 'particles'  respectively, of a coloured Maya diagram $\sf m^{(\alpha)}$, where $\alpha=1,...,N$, see figure \ref{fig:young}. Each particle (hole) carries a positive (negative) unit charge, so that the total charge associated to every Maya diagram is 
\begin{equation}
    \sf Q({\sf m}^{(\alpha)}) := \vert p(m^{(\alpha)}) \vert - \vert h(m^{(\alpha)})\vert.
\end{equation}
Using the notation
\begin{align}\label{eq:charge_maya}
    \vec{\sf m}:=\left(\sf{m}^{(1)},\dots,\sf{m}^{(N)} \right), && \vec{\sf Q}:=\left(\sf{Q}^{(1)},\dots,\sf{Q}^{(N)} \right),
\end{align}
the total charge is
\eq{\label{eq:totalchargeMaya}
    \sfQ:=\sum_{\alpha=1}^N\sfQ^{(\alpha)},
    }
and it is the same for every $N$-tuple of coloured Maya diagrams appearing in our expansions. Each Maya diagram determines uniquely a charged Young diagram $\left(\sf Y,\sf Q\right)$ as exemplified in figure \ref{fig:young}. Consequently, the minors can be labeled by $N$-tuples of charged partitions $\left(\vec{\sfY},\vec{\sfQ}\right)$. 
\begin{definition} With the labels in terms of partitions $\sf Y$ and charges $\sf Q$, let us define the trinion partition function by the following expression:
\begin{equation}
\begin{split}
   Z^{\sfY_{k},\sfQ_{k}}_{\sfY_{k+1},\sfQ_{k+1}}\left(\mathscr{T}^{[k]} \right)&= Z^{I_{k}, J_{k}}_{I_{k+1}, J_{k+1}}\left(\mathscr{T}^{[k]} \right)  
   := (-1)^{\vert I_{k+1} \vert} \det\left( \begin{array}{cc}\left( \sfa^{[k]} \right)^{I_{k}}_{J_{k}} &\left( \sfb^{[k]} \right)^{I_{k}}_{I_{k+1}} \\ & \\
    \left( \sfc^{[k]} \right)^{J_{k+1}}_{J_{k}} & \left( \sfd^{[k]} \right)^{J_{k+1}}_{I_{k+1}}\end{array} \right),
\end{split}\label{eq:YQ_Z}
\end{equation} 
 where $k=1,\dots, n$, and $I_{n+1}= I_{1}, J_{n+1}=J_{1}$. $\mathscr{T}^{[k]}$ is the $k$-th trinion in the pants decomposition in figure \ref{fig:nptTorusPants}.
\end{definition}
\begin{figure}[H]
\centering
\includegraphics[width=12.5cm]{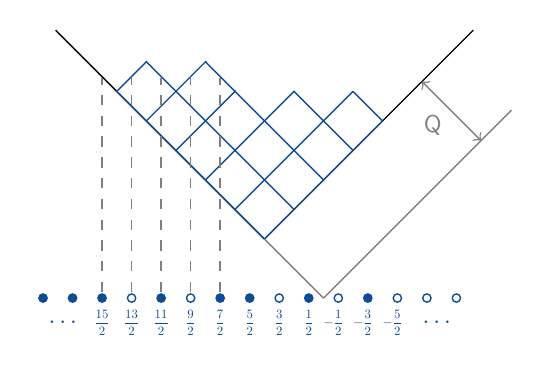}
\caption{Young and Maya diagrams \label{fig:young}}
\end{figure}
Note that the determinant in \eqref{eq:YQ_Z} is non zero for $\vert I_{k+1} \vert + \vert J_{k} \vert = \vert I_{k} \vert + \vert J_{k+1} \vert $, which in turn implies that all the Maya diagrams carry the same charge $\sf Q$.
\begin{proposition}\label{thm:Minors}
The determinant tau function $\T^{(1,n)}$ in \eqref{eq:Tau_1n} has the following minor expansion in terms of the trinion partition functions in \eqref{eq:YQ_Z}:
\begin{gather}\label{eq:MinorExp}
    \T^{(1,n)}=\sum_{\vec{\sfQ}_1,\dots \vec{\sfQ}_n}\sum_{\vec{\sfY }_1,\dots \vec{\sfY }_n\in\mathbb{Y}^N} e^{2\pi i\tau\left[\frac{1}{2} \left({\vec {\sf Q} }+ {\vec \sigma}_{1}  \right)^2 - \frac{1}{2} {\vec \sigma}_{1}^2 + \vert \vec{\sf Y}\vert\right]  - 2\pi i \left(\rho - \frac{\tau}{2}+\tau\sum_{j=1}^{n}\Lambda_j \right) {\sf Q} }\prod_{k=1}^nZ^{\vec{\sfY }_{k},\vec{\sfQ}_{k} }_{\vec{\sfY }_{k+1},\vec{\sfQ}_{k+1} }(\mathscr{T }^{[k]})
    \end{gather}
where the $Z^{\vec{\sfY }_{k},\vec{\sfQ}_{k} }_{\vec{\sfY }_{k+1},\vec{\sfQ}_{k+1} }(\mathscr{T }^{[k]})$ is the determinant defined in \eqref{eq:YQ_Z}, $\vec{\sigma}_1:=\left(\sigma_1^{(1)},\dots,\sigma_1^{(N)} \right)$\footnote{Note that here, differently from \eqref{eq:1nmonodromy} where we collected the monodromy exponents $\sigma_i^{(1)},\dots,\sigma_i^{(N)}$ into diagonal matrices denoted by $\bs \sigma_i$, we organize them into vectors $\vec{a}_i$, since they are summed with the charges $\vec{\sfQ}_i$, that are vectors in the root lattice $\mathbb{Z}^N$ of $\mathfrak{gl}_N$.} is the vector of monodromy exponents along the A-cycle of the torus, with modular parameter $\tau$.
\end{proposition}
\begin{proof}
From \eqref{eq:kijiji}, we can read off the minor expansion of the tau function \eqref{eq:TAU1n} in terms of the trinion partition functions in \eqref{eq:YQ_Z}:
\begin{equation}\label{eq:TAU_IJ}
    \T^{(1,n)} = \sum_{(\vec{I},\vec{J})} \prod_{k=1}^{n-1} Z^{I_{k}, J_{k}}_{I_{k+1}, J_{k+1}}\left(\mathscr{T}^{[k]} \right)   \times (-1)^{|I_1|} \det\left( \begin{array}{cc}\left( \sfa^{[n]} \right)^{I_{n}}_{J_{n}} &\left( \sfb^{[n]} \right)^{I_{n}}_{I_{1}} \\ & \\
    \left( \sfc^{[n]} \right)^{J_{1}}_{J_{n}} & \left( \sfd^{[n]} \right)^{J_{1}}_{I_{1}}\end{array} \right).  
    \end{equation}
Additionally, we can write the last factor in \eqref{eq:TAU_IJ} as follows
\begin{gather}
     \det\left( \begin{array}{cc}\left( \sfa^{[n]} \right)^{I_{n}}_{J_{n}} &\left( \sfb^{[n]} \nabla \right)^{I_{n}}_{I_{1}} \\  \\
    \left( \nabla^{-1} \sfc^{[n]} \right)^{J_{1}}_{J_{n}} & \left(\nabla^{-1} \sfd^{[n]} \nabla \right)^{J_{1}}_{I_{1}}\end{array} \right) \nonumber \\
    \mathop{=}^{\eqref{eq:GLnabla},\eqref{eq:GLnablainv}} \left(\prod_{(r,\alpha)\in J_{1}}  e^{-2\pi i \rho} e^{2\pi i \tau\left( \frac{1}{2}+ r+\sigma_{1}^{(\alpha)}-\sum_{j=1}^{n}\Lambda_j \right)}\right) \det\left( \begin{array}{cc}\left( \sfa^{[n]} \right)^{I_{n}}_{J_{n}} &\left( \sfb^{[n]} \right)^{I_{n}}_{I_{1}} \\  \\
    \left( \sfc^{[n]} \right)^{J_{1}}_{J_{n}} & \left(\sfd^{[n]} \right)^{J_{1}}_{I_{1}}\end{array} \right) \nonumber\\
\times \left(\prod_{(s,\beta)\in I_{1}} e^{2\pi i \rho}e^{2\pi i \tau\left(-\frac{1}{2} + s - \sigma_{1}^{(\beta)} +\sum_{j=1}^{n} \Lambda_{j} \right)} \right) \nonumber \\
    = \det\left( \begin{array}{cc}\left( \sfa^{[n]} \right)^{I_{n}}_{J_{n}} &\left( \sfb^{[n]}  \right)^{I_{n}}_{I_{1}} \\  \\
    \left(  \sfc^{[n]} \right)^{J_{1}}_{J_{n}} & \left( \sfd^{[n]}  \right)^{J_{1}}_{I_{1}}\end{array} \right) ~\textrm{exp}\left\{\sum_{(r,\alpha)\in J_{1}} \left[-2\pi i \left(\rho-\frac{\tau}{2}+\tau\sum_{j=1}^{n}\Lambda_j \right) +2\pi i \tau\left(r + \sigma_{1}^{(\alpha)} \right)\right]\right\} \nonumber\\
    \times\exp\left\{\sum_{(s,\beta)\in I_{1}} \left[2\pi i \left(\rho-\frac{\tau}{2}+\tau\sum_{j=1}^{n}\Lambda_j \right) + 2\pi i \tau\left(s - \sigma_{1}^{(\beta)}\right) \right] \right\} \nonumber \\
    = \det\left( \begin{array}{cc}\left( \sfa^{[n]} \right)^{I_{n}}_{J_{n}} &\left( \sfb^{[n]}  \right)^{I_{n}}_{I_{1}} \\  \\
    \left( \sfc^{[n]} \right)^{J_{1}}_{J_{n}} & \left( \sfd^{[n]}  \right)^{J_{1}}_{I_{1}}\end{array} \right) e^{2\pi i\tau\left[\frac{1}{2} \left({\vec {\sf Q} }+ {\vec \sigma}_{1}  \right)^2 - \frac{1}{2} {\vec \sigma}_{1}^2 + \vert \vec{\sf Y}\vert\right]  - 2\pi i \left(\rho - \frac{\tau}{2}+\tau\sum_{j=1}^{n}\Lambda_j \right) {\sf Q} }.\label{eq:Zn_block}
\end{gather}
In the second line of \eqref{eq:Zn_block}, we used the fact that if $\bs\sigma_1$ is the monodromy exponent on $\cC_{in}^{[1]}$, then the monodromy exponent on $\cC_{out}$ is $\bs\sigma_1-\sum_{j=1}^n\Lambda_j$ 
Since $s\in I_{1}$, the hole positions in the corresponding Maya diagram $\sf{m}$ are  ${\sf h(m)}= \left\{ -s_{1},...,-s_{k} \right\}$, and since $r\in J_{1}$, the particle positions are ${\sf p(m)} =\left\{ r_1,...,r_l  \right\} $. 
\begin{figure}[H]
\centering
\includegraphics[width=10cm]{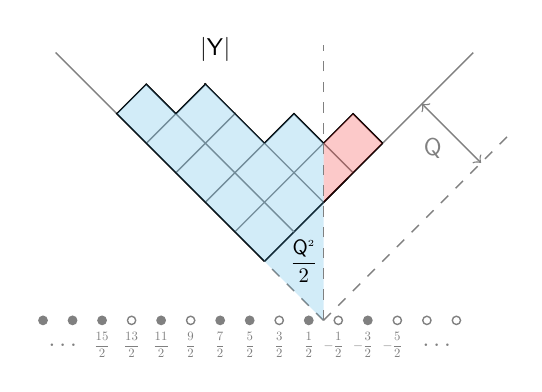}
\caption{Pictorial proof of \eqref{eq:srQY} \label{fig:maya}}
\end{figure}
To obtain the last line in \eqref{eq:Zn_block}, we use the following equalities:
\begin{gather}
\label{eq:srQY}
\sum_{l} r_{l} + \sum_{k} s_{k} = \frac{\sf Q(m)^2}{2} + \vert \vec{\sf Y} \vert , \qquad \#r-\#s= \sf Q(m), 
\end{gather}
which can be read off from Figure \ref{fig:maya} noting that the $r$'s and $s$'s are to the left and right sides of the axis respectively.
As an example, in the Figure \ref{fig:maya}, ${\sf p(m)}=\left\{ \frac{13}{2}, \frac{9}{2}, \frac{3}{2} \right\}$, ${\sf h(m)}=\left\{ \frac{-3}{2} \right\}$. $\vert Y \vert$ is the $\#$boxes in the Young diagram which in the present example is 12. The charge $\sf Q(m)=2$. $
\sum r$ is the blue area and $\sum s$ is the red area in the Figure \ref{fig:maya}. Equations \eqref{eq:TAU_IJ}, \eqref{eq:Zn_block} imply \eqref{eq:MinorExp}.
\end{proof}
Although the determinant tau function $\T^{(1,n)}$ in \eqref{eq:TAU1n} admits the expansion \eqref{eq:MinorExp}, the trinion partition functions \eqref{eq:YQ_Z} are known explicitly in terms of Nekrasov functions only in the case where the Lax matrix residues are of rank-1. We denote the determinant tau function for a generic Fuchsian system on the torus with rank-1 residues,i.e. residues of the form $\bs\mu_k=\left(\mu_1,0,\dots,0 \right)$, and monodromy exponents around $\cC_{in}^{[k]}$ given by $\vec{\sigma}_k$ by $\widetilde{\T}^{(1,n)}:=\det\left[\mathbb{I}-\widetilde{K}_{1,n} \right]$. Using the expressions for $Z^{\sfY _{k},\sfQ_{k} }_{\sfY_{k+1},\sfQ_{k+1} }(\mathscr{T }^{[k]})$ computed in \cite{Gavrylenko:2016zlf,Gavrylenko:2018fsm} for the rank-1 case, we obtain\footnote{Time-independent term \(\frac12\sum_{j=1}^{n}\Lambda_j\) comes from the ratios of the asymptotics of \(U(1)\) corrections to solutions of the 3-pt problems, given explicitly by $\left(\sin\pi(z-z_k)\right)^{\Lambda_k}$.}
\begin{equation}
    \begin{split}
    \widetilde{\T}^{(1,n)}&=\sum_{\vec{\sfQ}_1,\dots \vec{\sfQ}_n}\sum_{\vec{\sfY }_1,\dots \vec{\sfY }_n}e^{-2\pi i\sfQ\left(\rho-\frac{\tau}{2}+\tau\sum_{j=1}^{n}\Lambda_j+\frac12\sum_{j=1}^{n}\Lambda_j \right)}e^{2\pi i\tau\left[\frac{1}{2}(\vec\sfQ_1+\vec{\sigma}_1)^2-\frac{1}{2}\vec{\sigma}_1^2+\vert\vec{\sfY }_n\vert \right]} \\ & \times
   \prod_{j=1}^{n}e^{-2\pi i(z_{j}-z_{j-1})\left[\frac{1}{2}(\vec{\sigma}_{j}+\vec{\sfQ}_{j})^2-\frac{1}{2}\vec\sigma_j^2+\vert\vec{\sfY}_{j}\vert\right]} \\
   & \times \prod_{k=1}^ne^{2\pi i\vec{\sfQ}_i\cdot\vec{\nu}_i}\frac{Z_{pert}\left(\vec{\sigma}_{k}+\vec{\sfQ}_{k},\vec{\sigma}_{k+1}+\vec{\sfQ}_{k+1} \right)}{Z_{pert}\left(\vec{\sigma}_k,\vec{\sigma}_{k+1} \right)}Z_{inst}\left(\vec{\sigma}_{k}+\vec{\sfQ}_{k},\vec{\sigma}_{k+1}+\vec{\sfQ}_{k+1}\vert\vec{\sfY}_k,\vec{\sfY}_{k+1}\right)\\
   & =\det\left[\mathbb{I}-\widetilde{K}_{1,n} \right]
\end{split}  \label{eq:T1nNek}
\end{equation} 
where we set $z_{0}:=z_{n}$, the Fourier series parameters $\vec{\nu}_i$ were defined in \cite{Gavrylenko:2016zlf,Gavrylenko:2018fsm} in terms of the normalization of the three-point solution, and we have used introduced the functions
\begin{equation}\label{eq:Zpertk}
    Z_{pert}\left(\vec{\sigma},\vec{\mu}\right):= \prod_{\alpha,\beta=1}^N\frac{ G\left(1+ \sigma^{(\alpha)}- \mu^{(\beta)}\right)}{G\left(1+\sigma^{(\alpha)}-\sigma^{(\beta)}\right)},
\end{equation}
$G(x)$ being the Barnes' G-function, and
{\begin{gather}
    Z_{inst}\left(\vec{\sigma},\vec{\mu}\vert\vec{\sfY},\vec{\sf{W}}\right)
    :=\prod_{\alpha,\beta=1}^N\frac{Z_{\textrm{\sf bif}} \left( \sigma^{(\alpha)}- \mu^{(\beta)} \vert \sfY^{(\alpha)}, \sf{W}^{(\beta)} \right)}{Z_{\textrm{\sf bif}} \left( \sigma^{(\alpha)}- \sigma^{(\beta)}  \vert \sfY^{(\alpha)}, \sfY^{(\beta)} \right)}, \label{eq:Zinstdef}
\end{gather} }
with
\begin{equation}\label{eq:Zbif}
    Z_{\sf bif}\left(x|\sfY',\sfY \right):=\prod_{\Box\in\sfY}\left(x+1+a_{\sfY'}(\Box)+l_{\sfY}(\Box) \right)\prod_{\Box'\in\sfY'} \left(x-1-a_{\sfY}(\Box')-l_{\sfY'}(\Box') \right).
\end{equation}
In the above equations, $\vec{\sigma},\vec{\mu}\in\mathbb{C}^N$, $\vec{\sfY},\vec{\sf{W}}\in\mathbb{Y}^N$, where $a_{\sfY}(\Box)$ and $l_{\sfY}(\Box)$ denote respectively the arm and leg length of the box $\Box$ in the Young diagram $\sfY$, as in figure \ref{fig:armleg}. 
\begin{remark}
In \eqref{eq:T1nNek}, the expression
\begin{equation}
    Z^D:=e^{2\pi i\tau\vec{\sigma}_1^2}\det\left[\mathbb{I}-\widetilde{K}_{1,n} \right] \label{eq:ZDual}
\end{equation}
is the Nekrasov-Okounkov dual partition function \cite{Nekrasov:2003rj} of a circular quiver $\mathcal{N}=2$, $SU(N)$ gauge theory. By the AGT correspondence \cite{Alday:2009aq}, $Z^D$ is equal to a conformal block of $N$ free fermions on the torus, as in \cite{Bonelli:2019yjd}. Consequently, we expect $Z^D$ in \eqref{eq:ZDual} to satisfy appropriate bilinear equations, along the lines of \cite{Bershtein:2014yia,Bershtein:2016uov}.
\end{remark}
\begin{figure}
    \centering
\begin{tikzpicture}[scale=2.5]
\draw[ultra thick,darkspringgreen,pattern color=darkspringgreen,pattern=north east lines] (0.4,-0.4) rectangle (0.8,-0.8) (0.6,-0.3) node{$s$};
\draw[clip](0,0) -- ++(2,0) -- ++(0,-0.4)-- ++(-0.4,0)-- ++(0,-0.4)-- ++(-0.4,0)-- ++(0,-0.8)--(0,-1.6) --cycle;
\draw[scale=0.4](0,-1.6*3) grid (5,5);
\draw[red](1.2,-0.45) node[fill=white](){ $a(s)$} ;
\draw[ultra thick,red,<->] (0.8,-0.6)--(1.6,-0.6) ;
\draw[blue] (0.8,-1.2) node[fill=white](){ $l(s)$};
\draw[ultra thick,blue,<->] (0.55,-0.8) to (0.55,-1.6)  ;
\end{tikzpicture}
    \caption{Arm and leg length}
    \label{fig:armleg}
\end{figure}
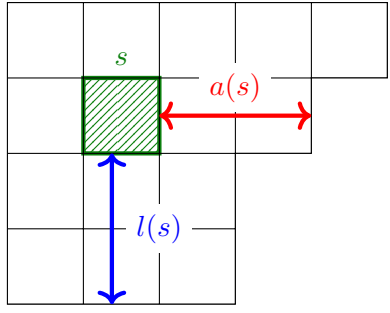
Our next goal is to relate the explicit expression \eqref{eq:T1nNek} for the tau function $\widetilde{\T}^{(1,n)}$ of a linear system on the torus with rank-1 residues, to the tau function $\T_H$ of an isomonodromic problem, where the residues are generic and satisfy the constrain \eqref{cons:CONSTRAINT}. With the observation that any $SL(2)$ matrix can be reduced to rank-1 by a scalar transformation, we will do this for the cases of the 2-particle nonautonomous Calogero-Moser system and of the elliptic Garnier system, which is the restriction to $N=2$, $\Lambda^{(j)}=0$, $j=1,\dots,n$ of the linear system \eqref{lax_general}.

\subsection{Reduction to rank-1 residues: the case of 2-particle nonautonomous Calogero-Moser system}\label{subsec:Minor_CM}

With the above considerations in mind, we formulate the tau function of the equation \eqref{eq:EllPainleve} in terms of the dual Nekrasov-Okounkov partition function \eqref{eq:T1nNek} for the $\mathcal{N}=2^*$ gauge theory\footnote{This is the $SU(2)$, $\mathcal{N}=2$ Super Yang-Mills theory with one massive adjoint hypermultiplet.}: the Lax matrix $L_{CM}$ in \eqref{eq:linear_systemCM} behaves as follows around the puncture $z=0$
\begin{equation}
        L_{CM}(z)=\left(\begin{array}{cc}
        P & mx(2Q,z) \\
        mx(-2Q,z) & -P
    \end{array}\right)= \frac{m\sigma_1}{z}+\mathcal{O}(1),
\end{equation}
so that it has rank-2 residue. To make it rank-1, we perform the scalar gauge transformation
\begin{align} \label{eq:11_Transformation}
    L_{CM}(z)\rightarrow \widetilde{L}_{CM} := L_{CM}(z)-\lambda_{CM}(z)^{-1}\partial_z \lambda_{CM}(z)\mathbb{I}_2, && \lambda_{CM}(z)=\theta_1(z)^m, \end{align}
after which the Lax matrix and its behavior around the puncture become
\begin{equation}\label{eq:gaugetransfLax_11}
        \widetilde{L}_{CM}(z)=\left(\begin{array}{cc}
        P-m\frac{\theta_1'(z)}{\theta_1(z)} & mx(2Q,z) \\
        mx(-2Q,z) & -P-m\frac{\theta_1'(z)}{\theta_1(z)}
    \end{array}\right)= \frac{m}{z}\left(\begin{array}{cc}
        -1 & 1 \\
        1 & -1
    \end{array} \right)+ \mathcal{O}(1).
\end{equation}
As a consequence of \eqref{eq:11_Transformation}, the monodromies will be dressed by additional scalar factors that we denote by $g_B(z),g_1$ for the B-cycle and for the monodromy around the puncture respectively. The absence of a factor $g_A$ for the A-cycle, as well as the expression for $g_B(z)$, are determined by the periodicity of theta functions:
\begin{gather}
 \lambda_{CM}(z+\tau)=\theta_1(z+\tau)^{m}=e^{-2\pi i\left(z+\frac{\tau+1}{2}\right)m}\lambda_{CM}(z) := g_{B}(z) \lambda_{CM}(z), \label{eq:LambdaCMB} \\
    \lambda_{CM}(z+1)=\theta_1(z+1)^{m}=e^{i\pi m}\lambda_{CM}(z):=g_A\lambda_{CM}(z).
\end{gather}
The $z$-dependence of the factor $g_B(z)$ leads to a nontrivial factor $g_1$ for the monodromy around $z=0$:
\begin{equation}\label{eq:LambdaCM1}
    g_1= e^{-2\pi im}.
\end{equation} 
The Hamiltonian tau function $\widetilde{\T}_{CM}$ associated to the gauge-transformed Lax matrix \eqref{eq:gaugetransfLax_11} is:
\begin{equation}\label{eq:tildeTCM}
    2\pi i\partial_\tau\log\widetilde{\T}_{CM} := \frac{1}{2} \oint_{A}dz \tr \widetilde{L}_{CM}^{2}.
\end{equation}    
\begin{proposition}
The tau function \eqref{eq:prop1} of the equation \eqref{eq:EllPainleve} is related to the tau function \eqref{eq:tildeTCM} of the rank-1 Lax matrix in \eqref{eq:11_Transformation} as
\begin{equation} \label{eq:prop_tildetau_tau}
\T_{CM}(\tau) =\widetilde{\T}_{CM}(\tau) \left(\eta(\tau)e^{\frac{i\pi\tau}{6}}\right)^{-2m^2},
\end{equation}
where $m$ is the monodromy exponent at the puncture and $\tau$ is the isomonodromic time.
\end{proposition}
\begin{proof}
We begin with the equation \eqref{eq:tildeTCM}:
\begin{equation}
\begin{split}
    2\pi i\partial_\tau\log\widetilde{\T}_{CM} & = \frac{1}{2} \oint_{A}dz \tr \widetilde{L}_{CM}^{2} \mathop{=}^{\eqref{eq:11_Transformation}}\frac{1}{2} \oint_{A}dz \tr L_{CM}^{2}+ \oint_{A}  dz\left(\lambda_{CM}^{-1}(z)\partial_z \lambda_{CM}(z) \right)^2\\
    &=2\pi i\partial_\tau\log\T_{CM}+m^2\int_0^1dz\left(\frac{\theta_1'(z)}{\theta_1(z)} \right)^2.
\end{split}\label{eq:tildeCM}
\end{equation}
 \begin{figure}[H]
\centering
\begin{tikzpicture}[scale = 4]
\tikzmath{\eps = 0.1;}
\begin{scope}
\draw[ultra thick,blue,decoration={markings, mark=at position 0.5 with {\arrow{>}}}, postaction={decorate}] (0,0) to (1,0);
\draw[thick] (0.72,0.5) to (1.72,0.5);
\draw[thick] (0,0) to (0.72,0.5);
\draw[thick] (1,0) to (1.72,0.5);
\draw[ultra thick,red,decoration={markings, mark=at position 0.15 with {\arrow{>}}}, postaction={decorate}] (0,0+\eps) to (1,0+\eps) to (1.72,0.5+\eps) to (0.72,0.5+\eps) to (0,0.0+\eps);

\node at (0.5,-0.1) {{\color{blue}$A$}};
\node at (1.22,0.6+\eps) {{\color{red}$\cC_\varepsilon$}};

\node at (-0,-0.06) {$0$};
\node at (1,-0.06) {$1$};
\node at (0.74,0.43) {$\tau$};
\node at (-\eps,\eps+0.02) {$i\varepsilon$};
\end{scope}
\end{tikzpicture}
\caption{Contour of integration}
\label{fig:ModContour_11}
\end{figure}
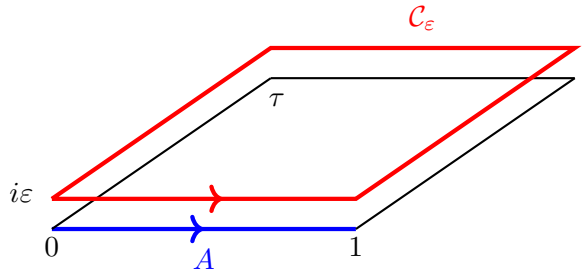
To compute the last term in \eqref{eq:tildeCM}, consider the following integral over the deformed contour $\cC_{\varepsilon}$ as in Figure \ref{fig:ModContour_11}
\begin{equation} \label{eq:tildeCM_int}\begin{split}
   2\pi i~\res_{z=\tau}\left(\frac{\theta_1'(z)}{\theta_1(z)} \right)^3 & =  \oint_{\mathcal{C}_\varepsilon}\left(\frac{\theta_1'(z)}{\theta_1(z)} \right)^3dz \\  &=\left[\int_{i\epsilon}^{1+i\epsilon}+\int_{1+i\epsilon}^{1+i\epsilon+\tau}+\int_{1+i\epsilon+\tau}^{i\epsilon+\tau}+\int_{i\epsilon+\tau}^{i\epsilon} \right]\left(\frac{\theta_1'(z)}{\theta_1(z)} \right)^3dz  \\
    & =\int_{i\epsilon}^{1+i\epsilon}\left[\left(\frac{\theta_1'(z)}{\theta_1(z)} \right)^3-\left(\frac{\theta_1'(z)}{\theta_1(z)}-2\pi i \right)^3 \right]dz  \\
 & = 6\pi i\int_{i\epsilon}^{1+i\epsilon}dz \left(\frac{\theta_1'(z)}{\theta_1(z)} \right)^2-3(2\pi i)^2\int_{i\epsilon}^{1+i\epsilon}\frac{\theta_1'(z)}{\theta_1(z)}+(2\pi i)^3 \\
 & =6\pi i\int_{i\epsilon}^{1+i\epsilon}dz \left(\frac{\theta_1'(z)}{\theta_1(z)} \right)^2+\frac{5}{2}(2\pi i)^3.
\end{split}
\end{equation}
To obtain the last line we use that
\begin{equation}\label{eq:tildeCM_1}
\int_{i\epsilon}^{1+i\epsilon}\frac{\theta_1'(z)}{\theta_1(z)}dz=\left[\log\theta_1(z)\right]^{1+i\epsilon}_{i\epsilon}=-i\pi.
\end{equation}
The residue on the left hand side of \eqref{eq:tildeCM_int} is computed shifting $z$ by $\tau$ and expanding around $0$:
\begin{gather}
    \res_{z=\tau}\left(\frac{\theta_1'(z)}{\theta_1(z)} \right)^3 = \res_{z=0}\left(\frac{\theta_1'(z+\tau)}{\theta_1(z+\tau)} \right)^3,
\end{gather}
and
\begin{equation}
    \left(\frac{\theta_1'(z+\tau)}{\theta_1(z+\tau)}\right)^3=\left(\frac{\theta_1'(z)}{\theta_1(z)}-2\pi i\right)^3=\frac{1}{z^3}-\frac{6\pi i}{z^2}+\frac{1}{z}\left(\frac{\theta_1'''}{\theta_1'}+3(2\pi i)^2\right)+\mathcal{O}(1).
\end{equation}
Therefore,
\begin{equation} \label{eq:tildeCM_res}
\res_{z=\tau}~\left(\frac{\theta_1'(z)}{\theta_1(z)} \right)^3=\frac{\theta_1'''}{\theta_1}+3(2\pi i)^2.
\end{equation}
Substituting \eqref{eq:tildeCM_res} in \eqref{eq:tildeCM_int}, and taking the limit $\epsilon\rightarrow0$ we get
\begin{equation}\begin{split}
   \frac{2\pi i }{m^2}\partial_{\tau} \log\left(  \frac{\widetilde{\T}_{CM}}{\T_{CM}}\right) &
  \mathop{=}^{\eqref{eq:tildeCM}}\lim_{\epsilon\rightarrow 0} \int_{i\epsilon}^{1+i\epsilon}dz \left(\frac{\theta_1'(z)}{\theta_1(z)} \right)^2\\
   &=\frac{1}{3}\frac{\theta_1'''}{\theta_1'}+\frac{1}{6}(2\pi i)^2\\ & \mathop{=}^{\eqref{eq:eta_id}}2\pi i\partial_\tau\log\eta(\tau)^2e^{i\pi\tau/3}.\label{eq:I2epsilon}
\end{split}\end{equation}
Therefore,
\begin{equation} \label{eq:tildetau_tau}
\eta(\tau)^{-2m^2}e^{-i\pi\tau m^2/3}\widetilde{\T}_{CM}=\T_{CM},
\end{equation}
having set the integration constant to $1$ without any loss of generality.
\end{proof}
\begin{theorem}\label{prop:Minor_CM}
The isomonodromic tau function $\T_{CM}$ admits the following combinatorial expansion:
\begin{equation}\begin{gathered}
\T_{CM}(\tau)  =\frac{\left(\eta(\tau)e^{-i\pi\tau/12} \right)^{2(1-m^2)}}{\theta_1\left(Q(\tau)+\rho-\frac{m(\tau+1)}{2}\right)\theta_1\left(Q(\tau)-\rho+\frac{m(\tau+1)}{2}\right)}e^{-2\pi i\left[\rho-\frac{\tau}{2}\left(m+\frac{1}{2}\right)- \frac{m}{2}\right]}\\
\times
\sum_{\vec{\sfQ}}\sum_{\vec{\sfY}\in\mathbb{Y}^2}e^{2\pi i\tau\left[\frac{1}{2}(\vec{\sfQ}+\vec{a})^2+|\vec{\sfY}|\right]}e^{2\pi i\left[\vec{\sfQ}\cdot\vec{\nu}-\sfQ\left(\rho-\frac{m(\tau+1)}{2}-\frac{\tau}{2}\right)\right]} \\
\times\frac{Z_{pert}\left(\vec{a}+\vec{\sfQ},\vec{a}+\vec{\sfQ}+m \right)}{Z_{pert}\left(\vec{a},\vec{a}+m\right)}Z_{inst}\left(\vec{a}+\vec{\sfQ},\vec{a}+\vec{\sfQ}+m\vert\vec{\sfY},\vec{\sfY} \right)\widetilde{\Upsilon}_{1,1} ,\label{eq:PropTCM}
\end{gathered}\end{equation}
where the functions $Z_{inst}$, $Z_{pert}$ are defined in \eqref{eq:Zinstdef}, \eqref{eq:Zpertk} respectively, $\vec{a}=\left(a,-a \right)$, with $a$ the local monodromy exponent around the A-cycle of the torus, $m$ is the monodromy exponent at the puncture $z=0$, $\rho$ is an arbitrary parameter, $Q$ is the solution of the equations of motion of the 2-particle nonautonomous Calogero-Moser system \eqref{eq:EllPainleve}, $\vec{\sfQ}$ is the vector of charges \eqref{eq:charge_maya}, $\sfQ$ is the total $U(1)$ charge \eqref{eq:totalchargeMaya}, and $\widetilde{\Upsilon}_{1,1}$ is an integration constant depending on monodromy data.
\end{theorem}
\begin{proof}
The linear system \eqref{eq:gaugetransfLax_11} is the specialisation of \eqref{eq:nptL} to the case $n=1$ (with the puncture at 0), $N=2$. The corresponding monodromy exponents $\vec{\sigma}_{1}$, $\vec{\sigma}_{2}$, and the $U(1)$ shifts $\Lambda_{0}$, $\Lambda_{1}$ in \eqref{eq:sigma_mu_def}, and the parameter $\widetilde{\rho}$ in \eqref{eq:tilderho}, for the present case are
\begin{align}
\vec{\sigma}_1=\left(a-\frac{m}{2},-a-\frac{m}{2} \right), && \vec{\sigma}_2=\left(a+\frac{m}{2},-a+\frac{m}{2} \right), \label{eq:sigma12} \end{align}\begin{align}  \Lambda_{0}= \frac{m}{2}, && \Lambda_1=-m, && \widetilde{\rho} = \rho -\frac{m (\tau+1)}{2}. \label{eq:Lambda01}
\end{align}
Theorem \ref{thm:GL_Ham} then implies that the tau function $\widetilde{\T}_{CM}$ in \eqref{eq:prop_tildetau_tau} can be written as a Fredholm determinant of an operator we call $\widetilde{K}_{1,1}$ whose minor expansion has an interpretation through Nekrasov functions as in \eqref{eq:T1nNek}, of the tau-function in \eqref{eq:Thm2}. Therefore,
\begin{equation}
\begin{split}
    \T_{CM} & \mathop{=}^{\eqref{eq:prop_tildetau_tau}}\eta(\tau)^{-2m^2}e^{-i\pi\tau m^2/3}\widetilde{\T}_{CM} \\ 
    &\mathop{=}^{\eqref{eq:Thm2}} \eta(\tau)^{-2m^2}e^{-i\pi\tau m^2/3}  e^{i\pi\tau\tr\left(\bs \sigma_1^2+\frac{\mathbb{I}}{6} \right)} e^{-2 \pi i \widetilde{\rho}}\frac{\eta(\tau)^2}{\theta_1(Q+\widetilde{\rho})\theta_1(Q-\widetilde{\rho})} \det\left[\mathbb{I} - \widetilde{K}_{1,1} \right]\widetilde{\Upsilon}_{1,1},\\
   &\mathop{=}^{\eqref{eq:sigma12},\eqref{eq:Lambda01}} \frac{\left(\eta(\tau)e^{-\frac{i\pi\tau}{12}}\right)^{2-2m^2}e^{2\pi i\tau a^2}}{\theta_1(Q+\rho-\frac{m(\tau+1)}{2})\theta_1(Q-\rho+\frac{m(\tau+1)}{2})}e^{-2\pi i\left[\rho-\frac{\tau}{2}\left(m+\frac{1}{2}\right)- \frac{m}{2}\right]} \det\left[\mathbb{I}-\widetilde{K}_{1,1} \right]\widetilde{\Upsilon}_{1,1} \\
    & \mathop{=}^{\eqref{eq:T1nNek}} \frac{\left(\eta(\tau)e^{-\frac{i\pi\tau}{12}}\right)^{2-2m^2} e^{-2\pi i\left[\rho-\frac{\tau}{2}\left(m+\frac{1}{2}\right)- \frac{m}{2}\right]}}{\theta_1\left(Q+\rho-\frac{m(\tau+1)}{2}\right)\theta_1\left(Q-\rho+\frac{m(\tau+1)}{2}\right)}
\sum_{\vec{\sfQ}}\sum_{\vec{\sfY}\in\mathbb{Y}^2}e^{2\pi i\tau\left[\frac{1}{2}(\vec{\sfQ}+\vec{a})^2+|\vec{\sfY}|\right]}\\
&\times e^{2\pi i\left[\vec{\sfQ}\cdot\vec{\nu}-\sfQ\left(\rho-\frac{m(\tau+1)}{2}-\frac{\tau}{2}\right)\right]} 
\frac{Z_{pert}\left(\vec{a}+\vec{\sf Q},\vec{a}+\vec{\sf Q}+m \right)}{Z_{pert}\left(\vec{a},\vec{a}+m\right)}Z_{inst}\left(\vec{a}+\vec{\sf Q},\vec{a}+\vec{\sf Q}+m\vert\vec{\sf Y},\vec{\sf Y} \right)\widetilde{\Upsilon}_{1,1}. \\
\end{split}
\end{equation}
\end{proof}
\begin{remark}
Equation \eqref{eq:PropTCM} coincides with equations (3.48) (4.10) in \cite{Bonelli:2019boe}, obtained by CFT methods. To compare the two expressions, one has to set $\sigma=0$ and send $\rho+\frac{1}{2}+\frac{\tau}{2}\rightarrow-\rho+\frac{m(\tau+1)}{2}$ in the expressions of \cite{Bonelli:2019boe}. 
\end{remark}
\subsection{Elliptic Garnier system and Nekrasov functions}\label{subsec:Minor_GL}
For the $N\times N$ case, it is in general only possible, with a scalar gauge transformation, to reduce the rank of the residues to $N-1$, which means that the minors can be written in terms of Nekrasov functions only in the case of semi-degenerate residues, as in \cite{Gavrylenko:2018ckn,Gavrylenko:2018fsm}.\footnote{In the context of class S theories \cite{Gaiotto:2009we,Gaiotto:2009hg} these are called \textit{minimal punctures}. The $A_{N-1}$ six-dimensional theory compactified on a torus with $n$ minimal punctures  gives rise to a four-dimensional circular quiver gauge theory.}
Therefore, we restrict the Lax matrix in \eqref{eq:nptL} to $N=2$, which can always be reduced to rank-1 by the scalar gauge transformation
\begin{equation}\label{eq:transf_n}
    \lambda(z)=\prod_{k=1}^{n}\theta_1 (z-z_k)^{m_k}
\end{equation}
where $m_k$ is the local monodromy exponent at the puncture $z_k$.
The new Lax matrix is
\begin{equation}\label{eq:rk1npt}
    \widetilde{L} := L - \lambda(z)^{-1} \partial_{z} \lambda(z) \mathbb{I}_{2}=L(z)-\sum_{k=1}^n m_k\frac{\theta_1'(z-z_k)}{\theta_1(z-z_k)}.
\end{equation}
The $U(1)$ factors around the punctures are given by
\begin{equation}
    g_k=e^{-2\pi im_k},
\end{equation}
while $g_A,g_B(z)$, are induced as before by the periodicity of theta functions:
\begin{equation}
\begin{split}
\lambda(z+\tau)&=\prod_{k=1}^n\theta_1(z-z_k+\tau)^{m_k}\\
    &=e^{-2\pi i\sum_{k=1}^n\left(z-z_k+\frac{\tau+1}{2}\right)m_k}\lambda(z):=g_B(z)\lambda(z) ,
    \end{split}
\end{equation}
\begin{align}
    \lambda(z+1)=\prod_{k=1}^n\theta_1(z-z_k+1)^{m_k}=e^{i\pi m}\lambda(z) :=g_A\lambda(z),
\end{align}
where we defined 
\begin{equation}
    m:=\sum_{j=1}^nm_j.\label{eqq:Defm}
\end{equation}
Again, we want to find the relation between the isomonodromic tau function $\T_H$ of the $SL(2)$ elliptic Garnier system \cite{Korotkin:1999xx,Takasaki:2001fr,Levin:2013kca}, and the $GL(2)$ tau function $\widetilde{\T}_H$ for the system with rank-1 residues obtained from the scalar gauge transformation \eqref{eq:transf_n}, defined by
\begin{align}
    2\pi i\partial_\tau\log\widetilde{\T}_H=\oint_Adz\frac{1}{2}\tr\widetilde{L}(z)^2, && \partial_{z_k}\log\widetilde{\T}_H=\res_{z_k}\frac{1}{2}\tr\widetilde{L}(z)^2. \label{eq:tildeTH}
\end{align}
\begin{proposition}\label{prop:THTHtilde}
The tau function $\widetilde{\T}_H$ \eqref{eq:tildeTH} of the rank-1 system is related to the tau function $\T_H$ \eqref{eq:IsomHam} of the Garnier system (whose Lax matrix is \eqref{eq:nptL} restricted to $N=2$, $\Lambda_i=0$ for $i=1,\dots,n$) as
\begin{gather}\label{eq:Rank1Garnier}
    \T_{H}(\tau) =  \widetilde{\T}_{H}(\tau) \prod_{k=1}^n\left(\eta(\tau) e^{\frac{i\pi\tau}{12}} \right)^{-2m_k^2}\prod_{l\ne k}\left(\frac{\theta_1(z_k-z_l)}{\eta(\tau)e^{-\frac{i \pi\tau}{3}}}\right)^{-m_km_l},
\end{gather}
where $m_k$ is the local monodromy exponent at the puncture $z_k$, $k=1,\dots, n$, and $\tau$ is the modular parameter of the torus.
\end{proposition}
\begin{proof}
Under the transformation \eqref{eq:transf_n}, the $z_k$-derivative of $\widetilde{\T}_H$ is
\begin{gather}
            \partial_{z_k}\log\widetilde{\T}_{H} = \res_{z_k}\frac{1}{2}\tr \widetilde{L}^{2} = \res_{z_k}\frac{1}{2}\tr L^{2}+ \res_{z_k}\left(\lambda^{-1}(z)\partial_z \lambda(z) \right)^2\nonumber\\
    =\partial_{z_k}\log\T_{H}+\sum_{j=1}^n\res_{z_k}\left[m_j^2\left(\frac{\theta_1'(z-z_j)}{\theta_1(z-z_j)} \right)^2+\sum_{l\ne k}m_jm_l\frac{\theta_1'(z-z_j)}{{\theta_1(z-z_j)}}\frac{\theta_1'(z-z_l)}{{\theta_1(z-z_l)}} \right] \nonumber \\
    = \partial_{z_k}\log\T_{H}+ \partial_{z_k}\log\left[\prod_{l\ne k}\left(\theta_1(z_k-z_l) \right)^{2m_km_l}\right] . \label{eq:tau_der_zk}
\end{gather}
In the last line we use that
\begin{equation}
\sum_{j=1}^{n}\sum_{l\ne k}m_jm_l\res_{z_k}\left(\frac{\theta_1'(z-z_k)}{{\theta_1(z-z_k)}}\frac{\theta_1'(z-z_l)}{{\theta_1(z-z_l)}}\right)=\partial_{z_k}\log\left(\prod_{l\ne  k}\theta_1(z_k-z_l)^{2 m_km_l} \right),
\end{equation}
and
\begin{equation}
    \res_{z_k}\left(\frac{\theta_1'(z-z_k)}{\theta_1(z-z_k)} \right)^2 =0.
\end{equation}
We now turn to the computation of the $\tau$-derivative:
\begin{gather}
            2\pi i\partial_\tau\log\widetilde{\T}_{H}  = \oint_{A}dz \frac{1}{2}\tr\widetilde{L}^{2} = \oint_{A}dz\frac{1}{2}\tr L^{2}+ \oint_{A}  dz\left(\lambda^{-1}(z)\partial_z \lambda(z) \right)^2\nonumber\\
    =2\pi i\partial_\tau\log\T_{H}+\sum_{k=1}^n\int_0^1dz\left[m_k^2\left(\frac{\theta_1'(z-z_k)}{\theta_1(z-z_k)} \right)^2+\sum_{l\ne k}m_km_l\frac{\theta_1'(z-z_k)}{{\theta_1(z-z_k)}}\frac{\theta_1'(z-z_l)}{{\theta_1(z-z_l)}} \right] \label{eq:tau_der_minor}
\end{gather}
Let us consider the A-cycle integral of the first term in equation \eqref{eq:tau_der_minor}.
 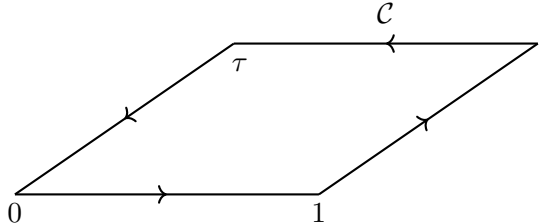
\begin{figure}[H]
\centering
\begin{tikzpicture}[scale = 4]
\tikzmath{\eps = 0.1;}
\begin{scope}
\draw[thick,decoration={markings, mark=at position 0.5 with {\arrow{>}}}, postaction={decorate}] (0,0) to (1,0);
\draw[thick,decoration={markings, mark=at position 0.5 with {\arrow{>}}}, postaction={decorate}] (1,0) to (1.72,0.5);
\draw[thick,decoration={markings, mark=at position 0.5 with {\arrow{>}}}, postaction={decorate}] (1.72,0.5) to (0.72,0.5);
\draw[thick,decoration={markings, mark=at position 0.5 with {\arrow{>}}}, postaction={decorate}] (0.72,0.5) to (0,0.0);
\node at (1.22,0.6) {$\cC$};

\node at (-0,-0.06) {$0$};
\node at (1,-0.06) {$1$};
\node at (0.74,0.43) {$\tau$};
\end{scope}
\end{tikzpicture}
\caption{Contour of integration}
\label{fig:ModContour_1n}
\end{figure}
The computation goes along the same lines of the $n=1$ case \eqref{eq:tildeCM_int}, but in the present case we do not shift the contour by $i\epsilon$, since the singularity $z_l$ is now in the interior of the contour $\cC$ in figure \ref{fig:ModContour_1n}:
\begin{equation}
    \begin{split}
        2\pi i~\res_{z=z_l}\left(\frac{\theta_1'(z-z_l)}{\theta_1(z-z_l)} \right)^3 & =\left[\int_0^1+\int_1^{\tau+1}+\int_{\tau+1}^\tau+\int_\tau^0\right]\left(\frac{\theta_1'(z-z_l)}{\theta_1(z-z_l)} \right)^3dz \\
& =6\pi i\int_0^1\left(\frac{\theta_1'(z-z_l)}{\theta_1(z-z_l)} \right)^2dz-3(2\pi i)^2\int_0^1\frac{\theta_1'(z-z_l)}{\theta_1(z-z_l)}dz+(2\pi i)^3 \\
& =6\pi i\int_0^1\left(\frac{\theta_1'(z-z_l)}{\theta_1(z-z_l)} \right)^2dz-\frac{1}{2}(2\pi i)^3, 
    \end{split}\label{eq:theta2RHS}
\end{equation}
while
\begin{gather}
2 \pi i~\res_{z=z_l}\left(\frac{\theta_1'(z-z_l)}{\theta_1(z-z_l)} \right)^3=2\pi i\frac{\theta_1'''}{\theta_1'}\mathop{=}^{\eqref{eq:eta_id}}3(2\pi i)^2\partial_\tau\log\eta(\tau)^2. \label{eq:theta3LHS}
\end{gather}
Equating \eqref{eq:theta2RHS}, \eqref{eq:theta3LHS}, we see that the first term of \eqref{eq:tau_der_minor} simply consists of $n$ copies of the 1-point computation \eqref{eq:I2epsilon}: 
\begin{equation}
\oint_Adz\sum_{k=1}^nm_k^2\left(\frac{\theta_1'(z-z_k)}{\theta_1(z-z_k)} \right)^2=2\pi i\partial_\tau\log\left[\prod_{k=1}^n\left(\eta(\tau)e^{\frac{i\pi\tau}{6}} \right)^{2m_k^2}\right]. \label{eq:theta2_int}
\end{equation}
We then turn to the computation of the second term of \eqref{eq:tau_der_minor}:
\begin{equation}
I_{kl}:=\oint_Adz\frac{\theta_1'(z-z_k)}{{\theta_1(z-z_k)}}\frac{\theta_1'(z-z_l)}{{\theta_1(z-z_l)}}.
\end{equation}
To compute $I_{kl}$, we consider the following integral over the deformed contour $\cC$ in figure \ref{fig:ModContour_1n}:
\begin{gather}
\left[\int_0^1+\int_1^{1+\tau}+\int_{1+\tau}^{\tau}+\int_\tau^0 \right] dz \frac{\theta_1'(z-z_k)}{{\theta_1(z-z_k)}}\left(\frac{\theta_1'(z-z_l)}{{\theta_1(z-z_l)}}\right)^2 \nonumber \\
=  2\pi i\left( ~\res_{z=z_k}+~\res_{z=z_l} \right)\frac{\theta_1'(z-z_k)}{{\theta_1(z-z_k)}}\left(\frac{\theta_1'(z-z_l)}{{\theta_1(z-z_l)}}\right)^2. \nonumber\\
= 2\pi i\left[\left(\frac{\theta_1'(z_k-z_l)}{\theta_1(z_k-z_l)} \right)^2+\frac{d}{dz_k}\left(\frac{\theta_1'(z_k-z_l)}{\theta_1(z_k-z_l)} \right) \right]\mathop{=}^{\eqref{eq:heat_theta1}}2\pi i\frac{\theta_1''(z_k-z_l)}{\theta_1(z_k-z_l)}\nonumber
\\=(2\pi i)^2\partial_\tau\log\theta_1(z_k-z_l)^2.
\label{eq:nRHS}
\end{gather}
The left-hand side of \eqref{eq:nRHS} is
\begin{gather}
\left[\int_0^1+\int_1^{1+\tau}+\int_{1+\tau}^{\tau}+\int_\tau^0 \right]dz\frac{\theta_1'(z-z_k)}{{\theta_1(z-z_k)}}\left(\frac{\theta_1'(z-z_l)}{{\theta_1(z-z_l)}}\right)^2\nonumber \\
=\int_0^1 dz\frac{\theta_1'(z-z_k)}{{\theta_1(z-z_k)}}\left(\frac{\theta_1'(z-z_l)}{{\theta_1(z-z_l)}}\right)^2-\int_0^1 dz \left(\frac{\theta_1'(z-z_k)}{\theta_1(z-z_k)}-2\pi i \right)\left(\frac{\theta_1'(z-z_l)}{\theta_1(z-z_l)}-2\pi i \right)^2 \nonumber \\
= 4\pi i I_{kl}-\frac{1}{2}(2\pi i)^3+2\pi i\int_0^1dz\left(\frac{\theta_1'(z-z_l)}{\theta_1(z-z_l)} \right)^2 \nonumber \\
\mathop{=}^{\eqref{eq:theta2_int}} 4\pi iI_{kl}-\frac{1}{2}(2\pi i)^3+(2\pi i)^2\partial_\tau\log\left(\eta(\tau)^2e^{\frac{i\pi\tau}{3}} \right) . \label{eq:nLHS}
\end{gather}
Equating \eqref{eq:nLHS} and \eqref{eq:nRHS}, we find
\begin{gather}
I_{kl}=2\pi i\partial_\tau\log\left(\frac{\theta_1(z_k-z_l)}{\eta(\tau)e^{-\frac{i\pi\tau}{3}}} \right).
\end{gather}
Therefore, the second term of \eqref{eq:tau_der_minor} reads
\begin{gather}
    \sum_{k=1}^{n} \sum_{l\ne k} m_{k} m_{l} \int_{0}^{1} dz \frac{\theta_{1}'(z-z_{k}) \theta_{1}'(z-z_{l})}{\theta_{1}(z-z_{k}) \theta_{1}(z-z_{l})} = 2\pi i \sum_{k=1}^{n} \sum_{l\ne k} m_{k} m_{l} \partial_{\tau} \log\left(\frac{\theta_1(z_k-z_l)}{\eta(\tau)e^{-\frac{i\pi\tau}{3}}} \right) \nonumber \\
    =2\pi i  \partial_{\tau} \log\left[\prod_{k=1}^n\prod_{l\ne k}\left(\frac{\theta_1(z_k-z_l)}{\eta(\tau)e^{-\frac{i\pi\tau}{3}}} \right)^{m_{k} m_{l}} \right].\label{eq:Tau_t_term2}
\end{gather}
Substituting \eqref{eq:theta2_int} and \eqref{eq:Tau_t_term2} in \eqref{eq:tau_der_minor},
\begin{equation}
\begin{split}
    2\pi i \partial_{\tau} \log \widetilde{\T}_{H} & =
    2\pi i \partial_{\tau} \log \T_{H} + 2\pi i\partial_\tau\log\left[\prod_{k=1}^{n}\left(\eta(\tau) e^{\frac{i\pi\tau}{6}}\right) ^{2m_k^2}\prod_{l\ne k}\left(\frac{\theta_1(z_k-z_l)}{\eta(\tau)e^{-\frac{i\pi\tau}{3}}} \right)^{  m_{k} m_{l}}\right]. 
\end{split} \label{eq:Tau_t}
\end{equation}
Combining \eqref{eq:tau_der_zk} and \eqref{eq:Tau_t} we find
\begin{gather}
     2\pi i \partial_{\tau} \log \widetilde{\T}_{H} +\sum_{k=1}^{n}\partial_{z_k}\log\widetilde{\T}_{H} =  2\pi i \partial_{\tau} \log \T_{H} + \sum_{k=1}^{n}\partial_{z_k}\log\T_{H} \nonumber \\
     2\pi i\partial_\tau\log\left[\prod_{k=1}^{n}\left(\eta(\tau) e^{\frac{i\pi\tau}{6}}\right) ^{2m_k^2}\prod_{l\ne k}\left(\frac{\theta_1(z_k-z_l)}{\eta(\tau)e^{\frac{-i\pi\tau}{3}}} \right)^{  m_{k} m_{l}}\right] \nonumber\\
     + \sum_{k=1}^n\partial_{z_k}\log\left[ \prod_{l\ne k}\theta_1(z_k-z_l)^{2m_km_l}\right]
\end{gather}
Integrating the above equation on both sides and setting the integration constant to $1$, we obtain
\begin{equation}
    \frac{\widetilde{\T}_H}{\T_H}=\prod_{k=1}^n\left(\eta(\tau) e^{\frac{i\pi\tau}{6}} \right)^{2m_k^2}\prod_{l\ne k}\left( \frac{\theta_1(z_k-z_l)}{\eta(\tau)e^{-\frac{i \pi\tau}{3}}}\right)^{m_km_l}.
\end{equation}
\end{proof}
\begin{remark}
Note that \eqref{eq:Rank1Garnier} takes the form of the partition function for a Coulomb gas on a torus, with the first term encoding the self-interaction of the particles, while the second term encodes the pairwise interactions.
\end{remark}
Using Proposition \ref{prop:THTHtilde}, it is possible to write the tau function of the elliptic Garnier system as a Fourier series of Nekrasov partition functions.
\begin{theorem}\label{thm:GL_minor}
The isomonodromic tau function of the elliptic Garnier system (see \eqref{eq:Thm2} restricted to $N=2$) admits the following combinatorial expression:
\begin{equation}
\begin{split}
    \T_{H}(\tau) &=\widetilde{\Upsilon}_{1,n}\frac{ e^{-2\pi i\left(\widetilde{\rho}-\frac{\tau}{4}\right)}}{ \theta_1\left(Q-\widetilde{\rho} \right) \theta_{1}\left(Q+\widetilde{\rho} \right)}\prod_{k=1}^n\left(\eta(\tau) e^{\frac{-i\pi\tau}{12}} \right)^{2-2m_k^2}e^{-2\pi iz_k m_k^2}
    \prod_{l\ne k}\left( \frac{\theta_1(z_k-z_l)}{\eta(\tau)e^{-\frac{i \pi\tau}{6}}}e^{-i\pi\left(z_k-z_l \right)}\right)^{-m_km_l}\\
    &\times   \sum_{\vec{\sfQ}_1,\dots \vec{\sfQ}_n}\sum_{\vec{\sfY }_1,\dots \vec{\sfY }_n}e^{-2\pi i\sfQ\left(\widetilde{\rho}-\frac{\tau}{2}\right)} e^{2\pi i\tau\left[\frac{1}{2}(\vec\sfQ_1+\vec{a}_1)^2+\vert\vec{\sfY }_n\vert \right]} \prod_{j=1}^{n}e^{-2\pi i(z_{j}-z_{j-1})\left[\frac{1}{2}(\vec{a}_{j}+\vec{\sfQ}_{j})^2+\vert\vec{\sfY}_{j}\vert\right]} \\
   &\prod_{k=1}^ne^{2\pi i\vec{\sfQ}_i\cdot\vec{\nu}_i}\frac{Z_{pert}\left(\vec{a}_{k}+\vec{\sfQ}_{k},\vec{a}_{k+1}+m_k+\vec{\sfQ}_{k+1} \right)}{Z_{pert}\left(\vec{a}_k,\vec{a}_{k+1}+m_k \right)}Z_{inst}\left(\vec{a}_{k}+\vec{\sfQ}_{k},\vec{a}_{k+1}+m_k+\vec{\sfQ}_{k+1}\vert\vec{\sfY}_k,\vec{\sfY}_{k+1}\right),\label{eq:ThmTHT1n}
   \end{split}
\end{equation}
where the functions $Z_{inst}$, $Z_{pert}$ are defined in \eqref{eq:Zinstdef}, \eqref{eq:Zpertk} respectively, $\vec{a}_k=\left(a_k,-a_k \right)$, $a_k$ being the $\mathfrak{sl}_2$ local monodromy exponent on the circle $\cC_{in}^{[k]}$ in figure \ref{fig:nptTorusPants}, $m_k$ is the $\mathfrak{sl}_2$ monodromy exponent at the puncture $z_k$, $Q\equiv Q(\tau; z_{1},...,z_{n})$ is the Calogero-like variable in the Lax matrix \eqref{eq:nptL} specialized to $N=2$, $\tau$ is the modular parameter, $\widetilde{\Upsilon}_{1,n}$ is an integration constant that depends on the monodromy data, $(\vec{\sfY}, \vec{\sfQ})$ are charged partitions, 
\begin{equation}
    \widetilde{\rho} = \rho - \sum_{j=1}^n \Lambda_{j} \left( z_{j} -\frac{(\tau+1)}{2}\right),
\end{equation}
and $\rho$ is an arbitrary parameter.
\end{theorem}
\begin{proof}
 The Lax matrix \eqref{eq:rk1npt} is the same as \eqref{eq:nptL} specialised to $n$-punctures, $N=2$. The monodromy exponents $\vec{\sigma}_{k}$, the $U(1)$ shifts $\Lambda_{k}$ in \eqref{eq:sigma_mu_def}, and the parameter $\widetilde{\rho}$ defined in \eqref{eq:tilderho} read as follows for the present case:
  \begin{align}
  \Lambda_{j}=-m_{j} \quad \textrm{for} \quad j=1...n, &&  \Lambda_{0}= \frac{m}{2}, \label{eq:1nLambdaj}
 \end{align}
 with $m$ defined in \eqref{eqq:Defm}, and
 \begin{align}
    \vec{\sigma}_{k}= \left( a_{k} -\sum_{j=0}^{k-1} \Lambda_{j}, \, -a_{k} -\sum_{j=0}^{k-1} \Lambda_{j} \right), && \widetilde{\rho} = \rho - \sum_{j=1}^n \Lambda_{j} \left( z_{j} -\frac{(\tau+1)}{2}\right). \label{eq:1nsigmakrho}
 \end{align}
Theorem \ref{thm:GL_Ham} then implies that the tau function $\widetilde{\T}_{H}$ in \eqref{eq:Rank1Garnier} can be written in terms of a Fredholm determinant of an operator we call $\widetilde{K}_{1,n}$ which in turn can be written in terms of  Nekrasov functions as in \eqref{eq:T1nNek}.
\begin{equation}
\begin{split}
    \T_H & \mathop{=}^{\eqref{eq:Rank1Garnier}}\left(\prod_{k=1}^n\left(\eta(\tau) e^{\frac{i\pi\tau}{6}} \right)^{-2m_k^2}\prod_{l\ne k}\left( \frac{\theta_1(z_k-z_l)}{\eta(\tau)e^{-\frac{i \pi\tau}{3}}}\right)^{-m_km_l}\right)\widetilde{\T}_H \\
    & \mathop{=}^{\eqref{eq:Thm2}} \left(\prod_{k=1}^n\left(\eta(\tau) e^{\frac{i\pi\tau}{6}} \right)^{-2m_k^2}\prod_{l\ne k}\left( \frac{\theta_1(z_k-z_l)}{\eta(\tau)e^{-\frac{i \pi\tau}{3}}}\right)^{-m_km_l}\right)e^{i\pi\tau\tr\left(\bs \sigma_1^2+\frac{\mathbb{I}}{6} \right)} e^{-2 i \pi \widetilde{\rho}}  \\ &\times\frac{\eta(\tau)^2}{ \theta_1(Q-\widetilde{\rho}) \theta_{1}(Q+\widetilde{\rho})} \left(\prod_{k=1}^n e^{-i\pi z_k\left(\tr\bs \sigma_{k+1}^2-\tr\bs \sigma_{k}^2\right)}\right)\det\left[\mathbb{I} - \widetilde{K}_{1,n} \right] \widetilde{\Upsilon}_{1,n},\\
    & \mathop{=}^{\eqref{eq:1nLambdaj}, \eqref{eq:1nsigmakrho}} \frac{e^{2\pi i\tau  a_1^2} e^{-2
        \pi i \left(\widetilde{\rho}-\frac{\tau}{4} \right)}}{ \theta_1(Q-\widetilde{\rho}) \theta_{1}(Q+\widetilde{\rho})}\prod_{k=1}^n\left(\eta(\tau) e^{\frac{-i\pi\tau}{12}} \right)^{2-2m_k^2}\prod_{l\ne k}\left( \frac{\theta_1(z_k-z_l)}{\eta(\tau)e^{-\frac{i \pi\tau}{6}}}\right)^{-m_km_l}\\
    & \times\prod_{k=1}^n e^{-2\pi i  z_k\left(a_{k+1}^2- a_{k}^2+m_k^2+m_k\sum_{j=1}^{k-1}m_j +m_{k}\right)}\det\left[\mathbb{I}-\widetilde{K}_{1,n}\right]\widetilde{\Upsilon}_{1,n} \\
    & \mathop{=}^{\eqref{eq:T1nNek}}\frac{ e^{-2
        \pi i \left(\widetilde{\rho}-\frac{\tau}{4} \right)}}{ \theta_1(Q-\widetilde{\rho}) \theta_{1}(Q+\widetilde{\rho})}\prod_{k=1}^n\left(\eta(\tau) e^{\frac{-i\pi\tau}{12}} \right)^{2-2m_k^2}e^{-2\pi iz_km_k^2}\prod_{l\ne k}\left( \frac{\theta_1(z_k-z_l)}{\eta(\tau)e^{-\frac{i \pi\tau}{6}}}e^{-i\pi\left(z_k-z_l \right)}\right)^{-m_km_l}\\
    & \times\prod_{k=1}^n e^{-2\pi i  z_k\left(a_{k+1}^2- a_{k}^2+m_{k}\right)}\widetilde{\Upsilon}_{1,n}  \times \sum_{\vec{\sfQ}_1,\dots \vec{\sfQ}_n}\sum_{\vec{\sfY }_1,\dots \vec{\sfY }_n}e^{-2\pi i\sfQ\left(\widetilde{\rho}-\frac{\tau}{2}\right)}\\
    &\times e^{2\pi i\tau\left[\frac{1}{2}(\vec\sfQ_1+\vec{a}_1)^2+\vert\vec{\sfY }_n\vert \right]}
   \prod_{j=1}^{n}e^{-2\pi i(z_{j}-z_{j-1})\left[\frac{1}{2}(\vec{a}_{j}+\vec{\sfQ}_{j})^2-\frac{1}{2}\vec{a}_j^2+\vert\vec{\sfY}_{j}\vert\right]} \\
   &\prod_{k=1}^ne^{2\pi i\vec{\sfQ}_i\cdot\vec{\nu}_i}\frac{Z_{pert}\left(\vec{a}_{k}+\vec{\sfQ}_{k},\vec{a}_{k+1}+m_k+\vec{\sfQ}_{k+1} \right)}{Z_{pert}\left(\vec{a}_k,\vec{a}_{k+1}+m_k \right)}Z_{inst}\left(\vec{a}_{k}+\vec{\sfQ}_{k},\vec{a}_{k+1}+m_k+\vec{\sfQ}_{k+1}\vert\vec{\sfY}_k,\vec{\sfY}_{k+1}\right).\\
\end{split}
\end{equation}
\end{proof}

\section{Comments on the Riemann-Hilbert problem on the torus and generalization of the Widom constant}\label{sec:RHP}

It is known from \cite{Cafasso:2017xgn} that the  isomonodromic tau function of a linear system on a sphere can be identified with the so-called Widom constant \cite{WIDOM19761}, which depends only on the \(N\times N\) jump matrix \(G(z)\) of the associated Riemann-Hilbert problem (RHP) with unit determinant.
For the 4-point isomonodromic problem on a sphere, the general RHP can be recast as a RHP on a unit circle with the jump \(G(z)\) is given by the ratio of the solutions $\psi_{in,out}$ of the auxiliary 3-point problems:
\begin{equation}G(z)=\psi_{in}(z)^{-1}\psi_{out}(z), \label{eq:jump1}\end{equation} and admits Birkhoff factorization. Finding the opposite factorization to \eqref{eq:jump1} is equivalent to solving the RHP with the above jump condition, which in turn is the same as solving the general 4-point RHP.

Following the same logic, we rewrite\footnote{ In this section, we use the $SL(N)$ analogues of the objects studied in Section \ref{sec:CM} for the one-punctured torus.} the Fredholm determinant \eqref{eq:prop1} in terms of the ratio of the solutions to the 3-point problem $\Yt$ as defined in \eqref{eq:3ptcyl}, and the projectors \(\Pi_{\pm}\) onto \(\mathcal{H}^0_{\pm}\), where \(\mathcal{H}^0_{\pm}\) are spaces of functions with identical A-cycle monodromies (in this section ``\(+\)'' and ``\(-\)'' are understood in the sense of \eqref{eq:functionComponents}):
\begin{equation}
\label{eq:4}
\mathcal{T}^{(1,1)}=\det_{\mathcal{H}_-\oplus \mathcal{H}_+} \left( \mathbb{I}- 
\begin{pmatrix}
\widetilde{Y}_- \nabla^{-1}\Pi_- \widetilde{Y}_+^{-1} & \Pi_- - \widetilde{Y}_-\Pi_-\widetilde{Y}_-^{-1}\\
\Pi_+-\widetilde{Y}_+\Pi_+\widetilde{Y}_+^{-1} & \widetilde{Y}_+\nabla \Pi_+ \widetilde{Y}_-^{-1}
\end{pmatrix} \right),
\end{equation}
where we set
\begin{equation}
\label{eq:Ypm}
\widetilde{Y}_-(z):=\widetilde{Y}_{out}(z+\tau),\quad \widetilde{Y}_+(z):=\widetilde{Y}_{in}(z).
\end{equation}
Using the fact that \(\Pi_\mp\) vanishes on \(\mathcal{H}_\pm^0\), and conjugating \eqref{eq:4} by the matrix \(\operatorname{diag}\left(\widetilde{Y}_-^{-1},\widetilde{Y}_+^{-1}\right)\), which is an isomorphism \(\mathcal{H}_-\oplus \mathcal{H}_+\to \mathcal{H}^0_-\oplus \mathcal{H}^0_+\), we get:
\begin{equation}
\label{eq:4b}
\mathcal{T}^{(1,1)}=\det_{\mathcal{H}_-^0\oplus \mathcal{H}_+^0} \left( \mathbb{I}-
\begin{pmatrix}
\nabla^{-1}\Pi_- \widetilde{Y}_+^{-1}\widetilde{Y}_- & -\Pi_-\widetilde{Y}_-^{-1}\widetilde{Y}_+\\
-\Pi_+\widetilde{Y}_+^{-1}\widetilde{Y}_- & \nabla \Pi_+ \widetilde{Y}_-^{-1}\widetilde{Y}_+
\end{pmatrix} \right).
\end{equation}
Defining the jump function
\begin{equation}
\label{eq:RHPdual}
J(z):=\widetilde{Y}_-^{-1}(z)\widetilde{Y}_+(z),
\end{equation}
and writing \(\nabla\) explicitly as in \eqref{eq:def_nabla}, we propose the following definition for a tau function generalizing the Widom constant to the case of the torus.
\begin{definition}
The torus tau function for the jump $J(z)$ defined in \eqref{eq:RHPdual} is
\begin{equation}
\label{eq:toricTau}
\mathcal{T}^{(1,1)}[\rho,J]
:=\det_{\mathcal{H}_-^0\oplus \mathcal{H}_+^0}
\begin{pmatrix}
\mathbb{I}-e^{-2\pi i\rho+\tau\partial_z}\Pi_- J^{-1} & \Pi_-J\\
\Pi_+J^{-1} & \mathbb{I}-e^{2\pi i\rho-\tau\partial_z} \Pi_+ J
\end{pmatrix},
\end{equation}
where
\[\det J(z)=1,\]
$\rho$ is an arbitrary parameter, and $\tau$ is the modular parameter.
\end{definition}
Let us now find an appropriate definition of the corresponding RHP.
Consider the solution to the the torus one-point linear system \(\mathcal{Y}(z)\) in \eqref{eq:linear_systemCM}, whose monodromies inside the fundamental domain are the same as the monodromies of the solution to the 3-pt problem \(\widetilde{\mathcal{Y}}(z)\). The function
\begin{equation}
\label{eq:8}
\Psi(z):=\widetilde{\mathcal{Y}}(z)^{-1}\mathcal{Y}(z)
\end{equation}
is analytic inside the fundamental domain, and satisfies the following relation on the A-cycle,
\begin{equation}
\label{eq:RHP}
\Psi(z+\tau)=J(z)\Psi(z)e^{2\pi i \boldsymbol{Q}[J]},
\end{equation}
where
\begin{equation}
\label{eq:10}
\boldsymbol{Q}[J]=\diag(Q_1[J],\ldots ,Q_N[J]),\qquad \sum_{i=1}^NQ_i[J]=0,
\end{equation}
making it a natural candidate for solution of the RHP.
An important observation here is that unlike in the spherical case, on the torus we have an extra diagonal twist \(e^{2\pi i \boldsymbol{Q}[J]}\), which implies that \(\Psi(z)\) is not a function, but a section of some non-trivial vector bundle on the torus. So we propose the following
\begin{definition}
The solution of the RHP \(\Psi(z)\) \eqref{eq:8} with jump \(J(z)\) on the A-cycle of the torus, is a section of a vector bundle, which is analytic in the fundamental domain and satisfies \eqref{eq:RHP} at the boundaries.
The complex moduli of the bundle \(Q_i[J]\) are some functionals of \(J\).
\end{definition}
We conjecture that the RHP for generic jump \(J(z)\) is solvable, and moreover, that complex moduli of the vector bundle are given by zeroes of \(\mathcal{T}^{(1,1)}[\rho,J]\) in \(\rho\) as in \eqref{eq:remark2}:
\begin{equation}
\label{eq:tauZero}
\mathcal{T}^{(1,1)}\left[Q_i[J],J\right]=0.
\end{equation}
It should not be hard to verify that the \(\rho\)-dependence of \(\mathcal{T}^{(1,1)}[\rho,J]\) should be factorizable as in the 1-point case \eqref{eq:prop1}, or even in the more general case \eqref{eq:Thm2}, with the $\rho$-independent part that we denote by $\mathcal{T}^{(1,1)}_0[J]$:
\begin{equation}
\label{eq:tauThetaN}
\mathcal{T}^{(1,1)}[\rho,J]=e^{N\pi i\rho}\prod_{i=1}^{N}\theta_1\left(\rho-Q_{i}[J]\right)\mathcal{T}^{(1,1)}_0[J].
\end{equation}
Furthermore, the solution to the RHP \eqref{eq:RHP} is not unique.
Namely, if \(\Psi_{N}(z)\) is a solution to the $N \times N$ RHP, then
\begin{equation}
\label{eq:1}
\Psi'(z)=\Psi(z)\cdot e^{2\pi i \boldsymbol{k} z},
\end{equation}
where \(\sum_{i=1}^N k_i=0\), is also a solution with the moduli of the vector bundle given by
\begin{equation}
\label{eq:5}
\boldsymbol{Q}'=\boldsymbol{Q}+\boldsymbol{k}\tau,
\end{equation}
using the Notation \ref{not:bold}. This demonstrates that \(Q_i\)'s are points on the same complex torus over which the RHP is formulated (they are the Tyurin points that parametrize the vector bundle on the torus, as in \cite{Krichever:2001zg,Krichever:2001cx}), and is consistent with \eqref{eq:tauThetaN}.

The next point of discussion is the distinction between what is called the direct RHP \eqref{eq:RHP} and the dual one \eqref{eq:RHPdual}.
In the spherical case the direct and dual RHPs were identical, and even in the present case it seems that we can rewrite the dual RHP in a similar way, \(\widetilde{Y}_{out}(z+\tau)=J(z)\widetilde{Y}_{in}(z)\).
However, the important difference with respect to the spherical case is that here, \(\Psi(z)\) is analytic for \(\Im z \in[0,\Im \tau]\), \(\widetilde{Y}_{in}(z)\) is analytic for \(\Im z\in(-i\infty,0]\), and \(\widetilde{Y}_{out}(z)\) is analytic for \(\Im z\in[\Im\tau,i\infty)\).
So the dual RHP is actually the same Birkhoff factorization problem as in the spherical case, while the direct one is different.

Another interesting question regards the possible options for \(J(z)\). In the 1-puncture case $J(z)$ was constructed from \eqref{eq:Ypm}, where $\Yt_-$, $\Yt_+$ were the analytic continuations of the same function in two different regions. We do not know the features of the RHPs with such jumps, which jumps are more natural to consider, and what is the natural \(\tau\)-dependence of \(J(z)\).

In the torus case, the jump \(J(z)\) can be in principle \(\tau\)-dependent, but we can also consider \(\tau\)-independent jump \(J(z)\) and find a limit to the usual Widom constant.
Namely, setting \(\rho=\tau/2\), one finds that the eigenvalues of \(e^{-2\pi i\rho+\tau\partial_z}\) acting on \(\mathcal{H}_-^0\) are \(e^{2\pi i (n+1/2)\tau}\) for \(n\geq 0\), and similarly for \(e^{2\pi i\rho-\tau\partial_z}\) acting on \(\mathcal{H}_+^0\) (recall the definition of $\cH_{+}$ in \eqref{eq:functionComponents}).
They vanish in the \(\tau\to i\infty\) limit, so we are left with the standard determinant giving the Widom constant.
We also see this degeneration at the level of the RHP \eqref{eq:RHP}: \(\Psi(z)\) at the boundaries can be first approximated by Fourier series' decaying towards the interior of the fundamental domain, and then consistency conditions around \(z=\tau/2\) will introduce corrections of order \(\mathcal{O}\left(e^{\pi i\tau}\right)\). 

Another observation is that equation \eqref{eq:RHP}, for rational \(J(z)\), can be considered as the solution to a \(q\)-difference linear system.
In this case \(\bs Q[J]\) has a different interpretation: following the \(q\)-difference generalization of the approach of \cite{Krichever:2004bb} to difference equations,
we introduce the following function
\begin{equation}
\label{eq:7}
\widehat{\Psi}(z)=\Psi(z)e^{-2\pi i z\boldsymbol{Q}[J]/\tau},
\end{equation}
which solves an equation
\begin{equation}
\label{eq:9}
\widehat{\Psi}(z+\tau)=J(z)\widehat{\Psi}(z).
\end{equation}
Then, \(\bs Q[J]\) parametrizes the local monodromy of the \(q\)-difference system, since
\begin{equation}
    \widehat{\Psi}(z+1)=\widehat{\Psi}(z)e^{-2\pi i\bs Q[J]/\tau},
\end{equation}
and the monodromy matrix $e^{-2\pi i\bs Q[J]/\tau}$ is defined up to multiplication by \(e^{-2\pi i\bs k/\tau}\), due to the presence of the \(\tau\)-periodic exponentials \(e^{-2\pi i\bs k z/\tau}\).

Another interesting question is then the following: what is the meaning of the torus tau function in this case of the \(q\)-difference system, and what information can we extract about the local monodromy from equation \eqref{eq:tauZero}? What is the generalization of the Widom formula for the variation of \(\mathcal{T}^{(1,1)}[\rho,J]\) under the variations of \(J(z)\)? We leave all these questions for a future work.

\section{Outlook and Discussion}

An immediate question regards the modular properties of tau functions on the torus which provides a way to study the so-called connection constant \cite{Iorgov:2013uoa, Its:2014lga}. A starting point in analyzing the modular properties of the tau function $\T_{H}$ in \eqref{eq:Thm2} is the free fermion conformal block $Z^D$ in \eqref{eq:ZDual}, whose transformations can be obtained by using the results of \cite{Ponsot:1999uf,Ponsot:2000mt,Hadasz:2010xp,Nemkov:2015zha,Nemkov:2016ikx}, where the so-called modular kernel, governing the behavior of the conformal block under modular transformations was derived. The other key ingredient is the variable $Q_i$ that appears in the argument of the theta functions. Its modular properties have been studied in \cite{manin1996sixth} for the $SL(2)$ one-punctured case, for which $Q$ is the solution of the equation \eqref{eq:EllPainleve}.

A natural continuation of this work is the generalization of the Fredholm determinant representation of tau functions on higher genus ($g\ge2$) Riemann surfaces, and for $g\ge 1$ cases with irregular singularities.  The main obstacle in providing explicit formulas for the tau function in both these cases is that the corresponding pants decomposition necessarily contains trinions with no external legs (see Figure \ref{fig:G2Dec}), for which the construction for the matrix elements of the Plemelj operators is not clear. Solving the problem posed by the all-internal trinion would immediately allow us to generalize several results on the Riemann sphere \cite{Gavrylenko:2017lqz,Cafasso:2017xgn,Desiraju:2019vna,2020arXiv200801142D} to the case of the torus.
\begin{figure}[h]
    \centering
    \includegraphics{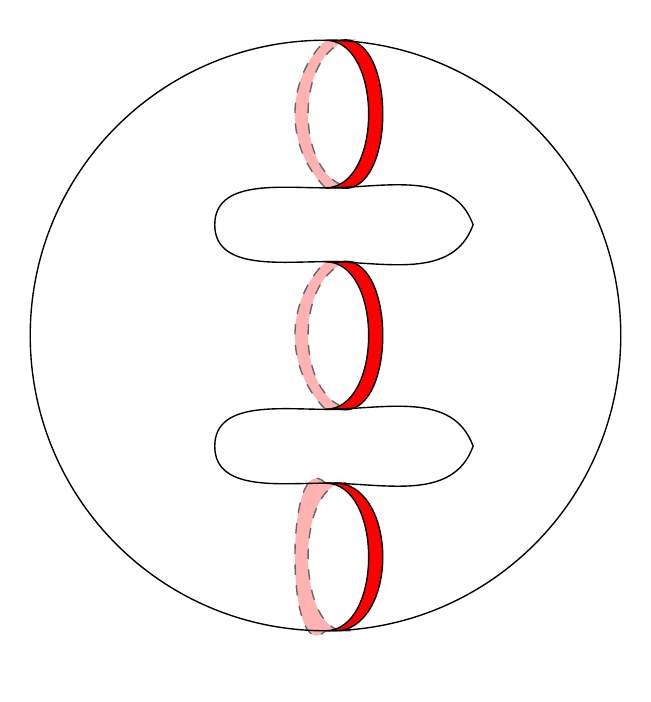}
    \caption{Pants decomposition for a genus-2 Riemann surface}
    \label{fig:G2Dec}
\end{figure}

The explicit formulas for the higher genus case would have important consequences in theoretical physics as well: the relation between the isomonodromic tau function and free fermion conformal blocks argued in \cite{DelMonte:2020nvp} would provide new explicit formulas for higher genus conformal blocks. From yet another perspective, the determinant tau functions studied in this paper coincide with partition functions of topological string theory on certain local Calabi-Yau threefolds, as already observed in \cite{Bonelli:2019boe}, and the identification with determinants provides a powerful nonperturbative definition of such partition functions \cite{Grassi:2014zfa,Bonelli:2016idi,Bonelli:2017ptp,Coman:2018uwk,Coman:2020qgf}. Further extending our construction to higher genus Riemann Surfaces and irregular punctures would provide explicit formulas for cases that have proven inaccessible by usual methods, like the topological vertex \cite{Aganagic:2003db}.
In CFT, these would be given by irregular conformal blocks on Riemann surfaces (for relations of irregular conformal blocks on the sphere with Painlev\'e equations, see also \cite{Bonelli:2016qwg,Nagoya:2015cja,2018arXiv180404782N,Gavrylenko:2020gjb}).

\section*{Acknowledgements}

We thank G. Bonelli, T. Grava, O. Lisovyy, A. Tanzini for their useful comments. Thanks are also due to the organizers of the workshops {\it Geometric Correspondences of Gauge Theories} (Trieste 2019) and {\it Winter School on Integrable Systems and Representation Theory} (Bologna, 2020) for their hospitality. F.D.M. thanks SISSA for providing him support in this taxing time after the end of his PhD.

The work of P.G. is partially supported by the HSE University Basic Research Program, Russian Academic Excellence Project '5-100' and by the RSF Grant No. 19-11-00275 (results of Sections \ref{subsec:Plemelj_CM} and \ref{sec:RHP} were funded by RSF). The work of H.D is partly supported by H2020-MSCA-RISE-2017 PROJECT No. 778010 IPaDEGAN. H.D thanks INFN Iniziativa Specifica GAST, and F.D.M thanks INFN Iniziativa Specifica ST\&FI, for funding the travel to  {\it Winter School on Integrable Systems and Representation Theory} (Bologna, 2020) where a part of the work was done. The work of F.D.M. was funded by SISSA, under the project "Fredholm determinants for tau functions and supersymmetric partition functions".

A special thanks to the apps Microsoft Whiteboard, Telegram, and Zoom, where most of the work was done.
P.G. would also like to thank Emacs+CD\LaTeX+Magit.

\bibliographystyle{JHEP}
\bibliography{Biblio.bib}
\end{document}